\documentclass[11pt,a4paper]{article}

\usepackage{amsmath} 
\usepackage{amsthm} 
\usepackage{amssymb}	
\usepackage{graphicx} 

\usepackage{xurl} 
\usepackage[pagebackref,colorlinks,linkcolor=NavyBlue,anchorcolor=blue,citecolor=NavyBlue]{hyperref}

\hypersetup{breaklinks=true}

\usepackage[dvipsnames]{xcolor}
\usepackage[dvips,margin=1in,bottom=1in]{geometry}

\allowdisplaybreaks[3]

\usepackage{xspace}
\usepackage[T1]{fontenc}
\AtBeginDocument{%
  \DeclareFontShape{T1}{cmr}{m}{scit}{<->ssub*cmr/m/sc}{}%
}
\usepackage[utf8]{inputenc}
\usepackage[english]{babel}

\usepackage{adjustbox} 
\usepackage{footnotehyper} 
\makesavenoteenv{table}  
\usepackage{multirow}
\usepackage{makecell}
\usepackage{booktabs}
\usepackage{float}

\usepackage{enumerate}
\usepackage{enumitem}

\usepackage{authblk}
\usepackage{diagbox}

\usepackage{mathtools}
\usepackage{thmtools}
\declaretheorem[style=plain,numberwithin=section]{theorem}
\declaretheorem[style=plain,numberlike=theorem]{lemma,corollary}
\declaretheorem[style=remark,numberlike=theorem]{remark}
\declaretheorem[style=plain,numberlike=theorem]{definition}

\declaretheorem[style=definition,numberlike=theorem]{problem,conjecture,question}

\numberwithin{equation}{section}

\usepackage[algoruled,algosection]{algorithm2e}
\usepackage[capitalize,noabbrev]{cleveref}
\crefname{algorithm}{protocol}{protocols}
\crefname{conjecture}{conjecture}{conjectures}
\newcommand{\algoref}[1]{Algorithm \ref{#1}}
\newcommand{\circuitref}[1]{Circuit \ref{#1}}

\crefname{appendix}{appendix}{appendices} 
\Crefname{appendix}{Appendix}{Appendices}

\usepackage{thm-restate}
\usepackage{xpatch}
\makeatletter
\xpatchcmd{\thmt@restatable}
  {\csname #2\@xa\endcsname\ifx\@nx#1\@nx\else[{#1}]\fi}
  {\ifthmt@thisistheone
     \csname #2\@xa\endcsname\ifx\@nx#1\@nx\else[{#1}]\fi
   \else
     \csname #2\@xa\endcsname
       \ifx\@nx#1\@nx
         [restated]%
       \else
         [{#1, restated}]%
       \fi
   \fi}
  {}
  {\PackageWarning{thm-restate}{Patch for restated theorem failed}}
\makeatother

\usepackage[most]{tcolorbox}

\usepackage[normalem]{ulem}

\usepackage{comment}

\DeclarePairedDelimiter\rbra{\lparen}{\rparen}
\DeclarePairedDelimiter\sbra{\lbrack}{\rbrack}
\DeclarePairedDelimiter\cbra{\{}{\}}
\DeclarePairedDelimiter\abs{\lvert}{\rvert}
\DeclarePairedDelimiter\norm{\lVert}{\rVert}
\DeclarePairedDelimiter\ceil{\lceil}{\rceil}
\DeclarePairedDelimiter\floor{\lfloor}{\rfloor}
\DeclarePairedDelimiter\ket{\lvert}{\rangle}
\DeclarePairedDelimiter\bra{\langle}{\rvert}

\DeclarePairedDelimiterX{\set}[2]{\{}{\}}{#1 \;\delimsize|\; #2}
\DeclarePairedDelimiterX{\sbraCond}[2]{(}{)}{#1 \;\delimsize|\; #2}

\newcommand{\ketbra}[2]{\ensuremath{\ket{#1}\bra{#2}}}

\makeatletter
\def\@buildmath#1{%
  \expandafter\def\csname bb#1\endcsname{\ensuremath{\mathbb{#1}}}%
  \expandafter\def\csname bf#1\endcsname{\ensuremath{\mathbf{#1}}}%
  \expandafter\def\csname sf#1\endcsname{\ensuremath{\mathsf{#1}}}%
  \expandafter\def\csname cal#1\endcsname{\ensuremath{\mathcal{#1}}}%
  \expandafter\def\csname rm#1\endcsname{\ensuremath{\mathrm{#1}}}%
  \expandafter\def\csname tt#1\endcsname{\ensuremath{\mathtt{#1}}}%
}
\def\@buildmathletters#1{%
  \ifx#1\relax\else
    \@buildmath{#1}%
    \expandafter\@buildmathletters
  \fi
} 
\@buildmathletters ABCDEFGHIJKLMNOPQRSTUVWXYZabcdefghijklmnopqrstuvwxyz\relax
\makeatother

\newcommand{\parheading}[1]{%
  \par\addvspace{1em}%
  \noindent\emph{#1}\enspace\ignorespaces%
}

\newcommand{\Tr} {\operatorname{Tr}}
\newcommand{\poly} {\operatorname{poly}}

\newcommand{\supp} {\operatorname{supp}}

\newcommand{\sign} {\operatorname{sgn}}
\newcommand{\arctanh}{\operatorname{arctanh}}
\newcommand{\yes}{{\rm yes}}
\newcommand{\no}{{\rm no}}

\newcommand{\Out}{\mathrm{out}}
\newcommand{\Real} {\operatorname{Re}}

\newcommand{\binset}{\{0,1\}}
\newcommand{\setcomplement}[1]{\overline{#1}}
\newcommand{\GetYes}[1]{{#1}_{\yes}}
\newcommand{\GetNo}[1]{{#1}_{\no}}

\newcommand{\PromiseProblemI}{\calI}

\newcommand{\Id}{\operatorname{Id}}
\newcommand{\bit}{\mathrm{bit}}
\newcommand{\qubit}{\mathrm{qubit}}

\newcommand{\Naturals}{\mathbb{N}}

\newcommand{\Pos}{\operatorname{Pos}}
\newcommand{\Dens}{\operatorname{Dens}}

\newcommand{\prob}[1]{\Pr\sbra*{#1}}

\newcommand{\MSucc}[1]{\mathrm{MSucc}\rbra*{#1}}
\newcommand{\BSV}[2]{\mathrm{BSV}\rbra*{#1,#2}}
\newcommand{\MSV}[2]{\mathrm{MSV}\rbra*{#1;#2}}

\newcommand{\PauliX}{X}
\newcommand{\PauliZ}{Z}

\newcommand{\BPP}{\textnormal{\textsf{BPP}}\xspace}
\newcommand{\BQP}{\textnormal{\textsf{BQP}}\xspace}

\newcommand{\IP}{\textnormal{\textsf{IP}}\xspace}
\newcommand{\QIP}{\textnormal{\textsf{QIP}}\xspace}
\newcommand{\QIPbit}{\textnormal{\textsf{QIP}\textsubscript{bit}\xspace}}
\newcommand{\IPbit}{\textnormal{\textsf{IP}\textsubscript{bit}\xspace}}
\newcommand{\QMAM}{\textnormal{\textsf{QMAM}}\xspace}

\newcommand{\SZK}{\textnormal{\textsf{SZK}}\xspace}

\newcommand{\QSZK}{\textnormal{\textsf{QSZK}}\xspace}

\newcommand{\QIPtwo}{\textnormal{\textsf{QIP(2)}}\xspace}
\newcommand{\PSPACE}{\textnormal{\textsf{PSPACE}}\xspace}
\newcommand{\QAM}{\textnormal{\textsf{QAM}}\xspace}

\newcommand{\qqQAM}{\textnormal{\textsf{qq}-\textsf{QAM}}\xspace}
\newcommand{\qcQAM}{\textnormal{\textsf{qc}-\textsf{QAM}}\xspace}

\renewcommand{\H}{\mathrm{H}}
\newcommand{\Hmin}{\mathrm{H}_\infty}
\newcommand{\HminCond}[2]{\Hmin \rbra*{#1 \mid #2}}
\renewcommand{\S}{\mathrm{S}}
\newcommand{\D}{\mathrm{D}}
\newcommand{\QJS}{\textnormal{\textrm{QJS}}\xspace}
\newcommand{\JS}{\mathrm{JS}}

\newcommand{\TV}{\mathrm{TV}}
\newcommand{\TD}{\mathrm{TD}}

\newcommand{\SD}{\textnormal{\textsc{SD}}\xspace}
\newcommand{\QSD}{\textnormal{\textsc{QSD}}\xspace}
\newcommand{\MultiQSD}{\textnormal{\textsc{MultiQSD}}\xspace}

\newcommand{\BinSteerVal}{\textnormal{\textsc{BinSteerVal}}\xspace}
\newcommand{\SteerVal}{\textnormal{\textsc{SteerVal}}\xspace}

\newcommand{\ED}{\textnormal{\textsc{ED}}\xspace}
\newcommand{\QED}{\textnormal{\textsc{QED}}\xspace}

\newcommand{\QJSP}{\textnormal{\textsc{QJSP}}\xspace}

\newcommand{\innerprodF}[2]{\langle #1 , #2 \rangle}

\newcommand{\dt}{\mathrm{d}t}
\newcommand{\dd}{\mathrm{d}}

\newcommand{\CNOT}{\textnormal{\textsc{CNOT}}\xspace}
\newcommand{\protocol}[2]{{#1}\!\rightleftharpoons\!{#2}}

\begin{document}
\setlength{\abovedisplayskip}{6pt}
\setlength{\belowdisplayskip}{6pt}


\title{On quantum interactive proofs with a laconic prover}

\author[1]{Zihan Hu\thanks{Email: \href{mailto:zihan.hu@epfl.ch}{zihan.hu@epfl.ch}}}
\author[1]{Yupan Liu\thanks{Email: \href{mailto:yupan.liu@epfl.ch}{yupan.liu@epfl.ch}}}
\affil[1]{%
  School of Computer and Communication Sciences,\protect\\
  \'Ecole Polytechnique F\'ed\'erale de Lausanne%
}

\date{}
\maketitle

\pagenumbering{roman}
\thispagestyle{empty}

\begin{abstract}
Interactive proof systems with a \emph{laconic} prover, studied by \hyperlink{cite.GVW02}{Goldreich, Vadhan, and Wigderson~(CC, 2002)}, provide a model for understanding which problems can be verified with only \emph{logarithmic} prover communication in the classical setting. 
For two-message quantum analogs of such models, an easy observation is that a \emph{single-bit} prover response suffices to contain quantum statistical zero-knowledge ($\sf QSZK$), introduced by \hyperlink{cite.Watrous02}{Watrous~(FOCS 2002)}. However, restricting the verifier's question to classical public coins collapses the corresponding class to $\sf BQP$, as established by \hyperlink{cite.BSW11}{Beigi, Shor, and Watrous~(ToC, 2011)}. 

In this work, we further investigate two-message quantum interactive proof systems with a laconic prover. To this end, we introduce the class ${\sf QIP}_{\ell\text{-}{\rm bit}}(2)$, where $\ell$ denotes the length of the prover's response, and establish the following results: 
\begin{enumerate}[label={(\arabic*)}]
    \item \emph{A natural complete characterization of ${\sf QIP}_{\ell\text{-}{\rm bit}}(2)$ by \textsc{Multi-State Distinguishability}.} In particular, \textsc{Quantum State Distinguishability} (\textsc{QSD}) is ${\sf QIP}_{\rm bit}$-complete. Together with the $\sf QSZK$-hardness of \textsc{QSD} proven by \hyperlink{cite.Watrous02}{Watrous~(FOCS 2002)}, our completeness result places ${\sf QIP}_{\ell\text{-}{\rm bit}}(2)$, for $\ell\geq 2$, in a landscape ``just above'' $\sf QSZK$.
    \item \emph{Easy regimes for ${\sf QIP}_{\ell\text{-}{\rm bit}}(2)$ that collapse to $\sf QSZK$.} We prove that $\textsc{QSD}[a,b]$ (and thus ${\sf QIP}_{\rm bit}[a,b]$) is in $\sf QSZK$ when $a(n)-b(n)\geq 1/O(\log n)$, and combine this containment with an answer compression from ${\sf QIP}_{\ell\text{-}{\rm bit}}[2,c,s]$ to ${\sf QIP}_{\rm bit}$ to obtain another easy regime when the completeness $c$ and soundness $s$ satisfy $2c>(1+2^{\ell/2})s$. Remarkably, our improvement in polarizing the trace distance to the natural regime also applies to the classical setting, showing that $\textsc{SD}[a,b]$ is in $\sf SZK$ for constant $a>b$, and resolving the first open problem posed in \hyperlink{cite.SV97}{Sahai and Vadhan~(JACM, 2003)}. 
    \item \emph{Quantum public coins also make the interaction useless in ${\sf QIP}_{\ell\text{-}{\rm bit}}(2)$}: ${\sf qc}\text{-}{\sf QAM}[O(\sqrt{\log{n}})]$ with constant promise gap is in $\sf BQP$, where ${\sf qc}\text{-}{\sf QAM}[\ell]$ denotes a subclass of ${\sf QIP}_{\ell\text{-}{\rm bit}}(2)$ in which the verifier's question consists solely of halves of EPR pairs. 
\end{enumerate}
\end{abstract}

\newpage
\tableofcontents
\thispagestyle{empty}
\newpage
\pagenumbering{arabic}


\section{Introduction}

Quantum interactive proof systems (\QIP{}) were initially studied in~\cite{Watrous99QIP}, and about a decade after the model was proposed, quantum interactive proof systems were shown to be \emph{no more powerful} than their classical counterparts~\cite{Babai85,GMR89} when the completeness-soundness gap, i.e., the gap between the acceptance probability guaranteed on \emph{yes} instances (\emph{completeness}) and the maximum acceptance probability on \emph{no} instances (\emph{soundness}), is at least inverse-polynomial, as established by $\QIP=\PSPACE$~\cite{JJUW11} and $\IP=\PSPACE$~\cite{LFKN92,Shamir92}. However, quantum interactive proof systems admit \emph{parallelization}: any quantum interactive proof system can be parallelized to use only \emph{three} messages~\cite{KW00}, while establishing such a property for classical interactive proof systems would collapse the polynomial-time hierarchy~\cite{Babai85,GS89,BHZ87}. 
This contrast makes the two-message setting particularly mysterious for quantum interactive proof systems and motivates a closer study of its restricted variants. 

In this work, we investigate two-message quantum interactive proof systems with a \emph{laconic prover}, where the verifier sends a question of polynomial length, but the prover's response has only \emph{logarithmic} length. Our model can be viewed as a quantum counterpart of classical interactive proof systems with a laconic prover~\cite{GH98,GVW02}, where the prover's total communication is restricted to $O(\log{n})$ bits or even shorter.

Such restricted models remain powerful \emph{in general}, and even the extremely restricted versions with a single-bit response contain interactive proof systems with the statistical zero-knowledge property in both the classical and quantum settings, particularly the classes \SZK{}~\cite{SV97,GSV98} and \QSZK{}~\cite{Watrous02,Wat09}, respectively. Nevertheless, requiring the verifier's question to consist only of \emph{classical public coins} makes the interaction \emph{useless} in laconic prover settings: the classical restricted model collapses to \BPP{} when the prover sends $O(\log{n})$ bits~\cite{GH98}, while the quantum restricted model collapses to \BQP{} when the prover sends $O(\log{n})$ qubits~\cite{BSW11}.

Beyond these complexity-theoretic motivations, two-message quantum interactive proof systems with a laconic prover also have a natural connection to succinct non-interactive protocols with setup in cryptography, especially when the verifier's question is further required to be instance-independent.
In particular, starting from such a proof system, one may obtain a candidate succinct non-interactive protocol, with ``succinct'' referring to sublinear communication, by moving the verifier's question to a setup phase and leaving only the laconic prover's response online.
Studying the power of such restricted proof systems therefore provides insight into whether succinctness and statistical soundness can be achieved simultaneously in the non-interactive protocols at the cost of an instance-independent setup, before we relax the soundness requirement from statistical to computational.

\subsection{Main results}

We begin by highlighting two modeling considerations that arise in our models of (two-message) quantum interactive proof systems with a laconic prover:
\begin{description}
    \item[Classical response from the prover.] Whether the prover's response is classical or quantum appears to matter at first glance. However, using the ideas of teleportation~\cite{BBCJPW93} and superdense coding~\cite{BW92}, any $\ell$-\emph{qubit} quantum response from the prover can be simulated by a $2\ell$-\emph{bit} classical response, and vice versa. These transformations preserve both completeness and soundness. Such ideas were used to prove that $\QIP=\mathsf{qcc}\text{-}\QAM$~\cite{KLGN19}, strengthening the earlier result $\QIP=\QMAM$~\cite{MW05} by showing that not only can the verifier's message be made classical public coins (\QMAM{}), but the prover's last response can also be made classical via teleportation. A similar teleportation argument gives the equivalence $\qqQAM=\qcQAM$.\footnote{In the two-message settings, the verifier's message, consisting of halves of EPR pairs, provides the entanglement needed to make the prover's quantum response classical.}
    \item[Quantum public coins.] Beyond classical public coins, sending halves of EPR pairs serves as their \emph{quantum} counterpart. This notion was first implicitly considered in non-interactive quantum statistical zero-knowledge~\cite{Kobayashi03} and then formalized in~\cite{KLGN19} through \qqQAM{} (or equivalently, \qcQAM{}), in which the verifier's message consists of halves of EPR pairs and the prover's response is quantum (\qqQAM{}) or classical (\qcQAM{}), respectively. 
\end{description}

Accordingly, in all our models with a laconic prover, we take the prover's response to be \emph{classical}.\footnote{See \Cref{sec:classical-messages-suffice} for the details of the equivalence between $\ell$-qubit and $2\ell$-bit prover responses in two-message quantum interactive proof systems, including its application to the quantum public coin setting.}
We write $\QIP_{\ell\text{-}\bit}[2,c,s]$ for the class of promise problems possessing two-message quantum interactive proof systems in which the prover sends an $\ell$-bit classical answer with completeness $c$ and soundness $s$, and abbreviate the single-bit class as $\QIPbit$. 
For the verifier's message, we additionally consider a \emph{fully quantum} version of public coins in laconic proof systems. We write $\qcQAM[\ell,c,s]$ for the class of promise problems possessing two-message \emph{quantum public coin} quantum interactive proof systems in which the prover sends an $\ell$-bit classical answer with completeness $c$ and soundness $s$. 
For convenience, we use $\QIP_{\ell\text{-}\bit}(2)$ and $\qcQAM[\ell]$ to denote, respectively, the unions of $\QIP_{\ell\text{-}\bit}[2,c,s]$ and $\qcQAM[\ell,c,s]$ over all $c(n)$ and $s(n)$ satisfying $c(n)-s(n)\geq 1/\poly(n)$. For simplicity, we sometimes abbreviate $\QIP_{\ell\text{-}\bit}(2)$ as $\QIP_{\ell\text{-}\bit}$.

\paragraph{A landscape ``just above'' \QSZK{}.}
Our first main result is a complete characterization of $\QIP_{\ell\text{-}\bit}$ via the \textsc{Multi-state Distinguishability Problem} (\MultiQSD{}), as stated in \Cref{thm-informal:complete-problems}. Specifically, given the description of a quantum circuit $Q$ that prepares an ensemble $\{\rho_r\}_{r\in\binset^\ell}$ after tracing out non-output qubits, $\MultiQSD[\ell,a,b]$ asks whether $\MSucc{Q}$ is at least $a$ (for \emph{yes} instances) or at most $b$ (for \emph{no} instances), where $\MSucc{Q}$ denotes the optimal probability of identifying a label $r$ chosen uniformly at random upon receiving the state $\rho_r$: 
\begin{equation}
    \label{eq:MultiQSD-pacc}
    \MSucc{Q} \coloneqq \max_{\text{POVM }\cbra{E_r}} \frac{1}{2^\ell}\sum_{r\in\binset^\ell}\Tr(E_r\rho_r)\enspace.
\end{equation}
We use the abbreviation $\MultiQSD_{2^\ell}$ for $\MultiQSD[\ell,a,b]$ with $a(n)-b(n) \geq 1/\poly(n)$.
\begin{theorem}[A natural complete characterization of $\QIP_{\ell\text{-}\bit}$, informal version of \Cref{thm:QSD-is-QIPbit-Complete,thm:MultiQSD-is-QIP-ell-bit-Complete}]
    \label{thm-informal:complete-problems}
    For every positive integer-valued function $\ell(n)\leq O(\log n)$, 
    \[\MultiQSD_{2^\ell} \text{ is } \QIP_{\ell\text{-}\bit}\text{-complete}\enspace.\] 
    In particular, \QSD{} is \QIPbit{}-complete. 
\end{theorem}

\sloppypar
Notably, \Cref{thm-informal:complete-problems} not only complements the observation that $\QSZK\subseteq\QIPbit$~\cite[Section~5]{BSW11} by establishing the corresponding \QIPbit{}-hardness, but also places $\QIP_{\ell\text{-}\bit}$ in a landscape ``just above'' \QSZK{}. 
In addition, \MultiQSD{} is a natural problem arising from quantum information theory and captures the standard quantum state discrimination problem (see, e.g.,~\cite[Section~3.1.2]{Watrous18}), whose optimal success probability admits a semidefinite programming (SDP) formulation~\cite[Section~II]{EMV03} (see also~\cite[Section~4.2]{vAG19}). This promise problem not only recovers the \textsc{Quantum State Distinguishability Problem} (\QSD{}) introduced in~\cite{Watrous02} when $\ell=1$,\footnote{When $\ell=1$, \Cref{eq:MultiQSD-pacc} becomes $\max_{\text{POVM } \cbra{\Pi,I-\Pi}} \rbra*{ \Tr(\Pi\rho_0) + \Tr((I-\Pi)\rho_1) }/2 = (1+\TD(\rho_0,\rho_1))/2$, where the underlying optimal measurement is the Holevo--Helstrom measurement~\cite{Holevo73TraceDist,Helstrom69}.}
but also provides a quantum counterpart of the \textsc{Generalized Statistical Difference Problem} (\textsc{GSD}) considered in~\cite{GVW02}.\footnote{One can recover \textsc{GSD} from \MultiQSD{} by simply forcing the underlying states to be diagonal matrices that encode probability distributions $\cbra{D_r}$ over $\Omega$. Consequently, \Cref{eq:MultiQSD-pacc} becomes the desired quantity $2^{-\ell} \sum_{z\in\Omega} \max_{r\in\binset^\ell} D_r(z)$ in \textsc{GSD}.} 

\paragraph{Easy regimes that collapse to \QSZK{}.}
Our second main result concerns polarizing the trace distance (quantum) and the total variation distance (classical) in the natural regime, thereby resolving the first open problem listed in~\cite[Section~6]{SV97}: 

\begin{theorem}[Polarizing the $\ell_1$-norm distances in the natural regime, informal version of \Cref{thm:QSD-in-QSZK-natural-regime,thm:SD-in-SZK-natural-regime}]
\label{thm:quantum-polarization-informal}
For functions $a, b \colon \Naturals \to [0, 1]$ such that $0 \leq b(n) < a(n) \leq 1$ and $a(n)-b(n)\geq1/O(\log n)$, for all sufficiently large $n$, the following inclusions hold:
\[ \QSD[a,b]\in\QSZK \quad\text{and}\quad \SD[a,b]\in\SZK\enspace. \]
\end{theorem}

Here, the promise problem $\QSD[a,b]$ asks whether the trace distance $\TD(\rho_0,\rho_1)$ between two states $\rho_0$ and $\rho_1$ that admit polynomial-size state-preparation circuits, is at least $a$ or at most $b$, and serves as the quantum counterpart of \SD{}~\cite{SV97}, which is defined via the total variation distance between two efficiently samplable distributions.

All known approaches to classical polarization~\cite{SV97,BDRV19} are subject to the \emph{oblivious black-box polarization} barrier in~\cite[Corollary~2]{HR05}: the model in \cite{HR05} captures these approaches and permits polarization for all pairs of input states or distributions with fixed thresholds $a>b$ only if $a^2>b$. Remarkably, \Cref{thm:quantum-polarization-informal} is the first approach to bypass this barrier. 
To put this distinction in context, polarizing the total variation distance, originally proposed in~\cite[Section~3.2]{SV97}, is the key ingredient for establishing $\SD[a,b]\in\SZK$ for constants $a^2>b$. This zig-zag-type construction extends readily to the trace distance~\cite[Section~4.1]{Watrous02} and shows that $\QSD[a,b]\in\QSZK$ for the same regime. A direct inspection shows that this construction extends to the regime $a(n)^2-b(n) \geq 1/O(\log{n})$. A further improvement of the classical polarization, relaxing this requirement to $1/\poly(n)$, was established in~\cite{BDRV19}, with the corresponding quantum adaptation given in \Cref{sec:QSD-in-QSZK}, whereas a direct quantum analog in~\cite{Liu23} applies only to a slightly weaker regime. 

Let $\IP_{\bit}[c,s]$ denote the classical analogue of $\QIPbit[c,s]$. Taking \Cref{thm:quantum-polarization-informal} together with appropriate \QIPbit{} and \IPbit{} hardness results,\footnote{$\IPbit[a,b]$ is $\SD[a,b]$-hard for all $0<b<a<1$ in~\cite[Lemma~4.1]{GVW02}, and we prove the quantum counterpart in~\Cref{thm:QSD-is-QIPbit-hard}, which improves the $\ell=1$ restriction of \Cref{thm-informal:complete-problems} by removing the $1/2$-shift in the parameter regime.} we obtain the following containments and, informally, the equivalence $\QIPbit=\QSZK$, thereby improving~\cite[Theorem 3.1]{GVW02}:

\begin{corollary}[One-bit laconic proof systems and statistical zero-knowledge, informal version of \Cref{thm:QIPbit-in-QSZK,thm:IPbit-in-SZK}]
\label{thm:bit-containments-informal}
For functions $c, s \colon \Naturals \to [0, 1]$ such that $0 \leq s(n) < c(n) \leq 1$ and $c(n)-s(n)\geq 1/O(\log n)$ for all sufficiently large $n$, we have
\[ \QIPbit[c,s]\subseteq\QSZK \quad\text{and}\quad \IPbit[c,s]\subseteq\SZK\enspace. \]
\end{corollary}

\vspace{1em}
Our third main result is an \emph{answer compression} for $\QIP_{\ell\text{-}\bit}$ proof systems that compresses the prover's response from $\ell$ bits to a single bit, provided that the completeness $c$ and soundness $s$ are sufficiently separated: 

\begin{theorem}[Answer compression for $\QIP_{\ell\text{-}\bit}$, informal version of \Cref{thm:QIPlbit-answer-compression}]
\label{thm:answer-compression-informal}
Let $c, s \colon \Naturals \to [0, 1]$  be functions such that $0 \leq s(n) < c(n) \leq 1$. For every positive integer-valued function $\ell(n)\leq O(\log n)$, the following statement holds:
\[ \text{If } c > \frac{1+2^{\ell/2}}{2} \cdot s, \quad 
    \QIP_{\ell\text{-}\bit}[2,c,s] \subseteq \QIPbit\sbra*{\frac{1}{2}+\frac{c}{2^{\ell+1}},\frac{1}{2}+\frac{(1+2^{\ell/2})s}{2^{\ell+2}}}\enspace.\]
\end{theorem}

It is noteworthy that \Cref{thm:answer-compression-informal} is not only a quantum analog of~\cite[Theorem~3.4]{GVW02}, but also covers a \emph{strictly larger} parameter regime. For instance, \Cref{thm:answer-compression-informal} applies to $\ell=2$, $c=1$, and $s=3/5$, while these parameters violate the requirement $s<2^{-\ell/2}$ in~\cite[Theorem~3.4]{GVW02}. 
Combining \Cref{thm:answer-compression-informal} and \Cref{thm-informal:complete-problems} with $\ell=1$, we obtain:
\begin{corollary}[Easy regimes for $\QIP_{\ell\text{-}\bit}$ with $\ell\geq 2$]
\label{corr:easy-regime}
Let $c, s \colon \Naturals \to [0, 1]$  be functions such that $0 \leq s(n) < c(n) \leq 1$ and $c(n)-s(n)\geq 1/O(\log{n})$ for all sufficiently large $n$. For every positive integer-valued function $\ell(n)\leq \floor*{\log\log n}$, the following statement holds: 
\[ \text{If } c(n) - \frac{1+2^{\ell(n)/2}}{2} \cdot s(n) \geq\frac{2^{\ell(n)}}{O(\log n)}, \quad \QIP_{\ell\text{-}\bit}[2,c,s]\subseteq\QSZK. \]
\end{corollary}

\paragraph{Public coins weaken quantum proofs with a laconic prover.} Our last main result strengthens~\cite[Theorem 5.1]{BSW11} from classical public coins to \emph{quantum public coins}: 

\begin{theorem}[Quantum public coins make interaction useless in \qcQAM{}, informal version of \Cref{thm:qcQAM[1]=BQP,thm:qcQAM[sqrt log n]=BQP}]
\label{intro:thm:public-one-bit}
Let $c, s \colon \Naturals \to [0, 1]$  be functions such that $0 \leq s(n) < c(n) \leq 1$. For every positive integer-valued function $\ell(n)\leq O\rbra*{\sqrt{\log{n}}}$, the following inclusion holds:
\[ \text{If } \Delta(n)\coloneqq c(n)-s(n) \geq \Omega(1), \quad \qcQAM[\ell,c,s]\subseteq\BQP\enspace.  \]
Here, the underlying quantum algorithm remains efficient whenever $\ell^2/\Delta^6 = O(\log{n})$. 
Furthermore, in the single-bit case, $\qcQAM[1] \subseteq \BQP$ holds for all $\Delta(n) \geq 1/\poly(n)$. 
\end{theorem}

\subsection{Proof techniques: Hardness results and answer compression}

Our hardness results and answer-compression reduction start by formulating the maximum acceptance probability, denoted by $\omega(V)$, of a $\QIP_{\ell\text{-}\bit}$ proof system $\protocol{P}{V}$. Since the prover returns only a classical string $r \in \binset^\ell$, the prover's strategy is described by a POVM $\cbra*{M_r}_{r \in \binset^\ell}$ on the message register $\sfM$, and $\omega(V)$ can thus be written as an optimization problem (\Cref{lemma:optimal-acceptance-classical-response}):
\begin{align}
\label{eqn:maximum-acceptance-prob-intro}
    \omega(V) =\max_{\textnormal{POVM }\{M_r\}_{r\in\binset^\ell}}\sum_{r\in\binset^\ell}\Tr\rbra*{M_rX_{r}}\enspace.
\end{align} 
For each $r\in\binset^\ell$, $X_r$ is a subnormalized state on $\sfM$ defined by
\begin{align*}
    X_{r}\coloneq\Tr_{\sfM'\sfW}\rbra*{\ketbra{1}{1}_{\sfO}\rbra*{V_2}_{\sfM'\sfW}\rbra*{\ketbra{r}{r}_{\sfM'}\otimes \ketbra{\psi}{\psi}_{\sfM\sfW}}\rbra*{V_2}^\dagger_{\sfM'\sfW}}\enspace.
\end{align*}
Here, $V_1$ and $V_2$ are the verifier's actions and $\ket{\psi} \coloneq V_1\ket{\bar{0}}_{\sfM\sfW}$ is the state on the registers $\sfM$ and $\sfW$ after the verifier's first action, with $\sfW$ denoting the verifier's private register.

\Cref{eqn:maximum-acceptance-prob-intro} is closely related to the optimal discrimination probability in \Cref{eq:MultiQSD-pacc}, except that the $X_{r}$ are subnormalized states and may have different traces $p_{r} \coloneq \Tr(X_{r})$. In the following, we show how to appropriately pad the operators $\cbra*{X_{r}}$ to form efficiently preparable states $\cbra*{\rho_{r'}}$ whose optimal discrimination probability indicates whether $\omega(V) \geq c$ or $\omega(V) \leq s$, thereby reducing a promise problem in $\QIP_{\ell\text{-}\bit}$ to the corresponding state discrimination problem. For the hardness results, $r'\in \binset^\ell$, whereas for the answer compression, $r'\in\binset$.

\paragraph{Hardness: $\ell = 1$.}
Our starting point is the closed-form expression given by the Holevo--Helstrom formula for \emph{binary} generalized state discrimination~\cite{AM14}:
\begin{align}
\label{eqn:omega(V)-intro}
    \omega(V)=\frac{\Tr(X_{0}) + \Tr(X_{1})+\norm{X_{0}-X_{1}}_1}{2}\enspace.
\end{align}
We define two states whose trace distance is $\omega(V)$ by encoding the terms in separate blocks:
\begin{align*}
    \rho_{0} &\coloneq X_{0} \otimes \ketbra{0}{0} + (\Tr_{\sfW}(\ketbra{\psi}{\psi}) - X_{0}) \otimes \ketbra{1}{1},\\
    \rho_{1} &\coloneq X_{1} \otimes \ketbra{2}{2} + (\Tr_{\sfW}(\ketbra{\psi}{\psi}) - X_{1}) \otimes \ketbra{1}{1}.
\end{align*} 
Here, $\Tr_{\sfW}(\ketbra{\psi}{\psi}) -X_{r}$ is the subnormalized rejection state when we run the verifier with the fixed response $r$, and thus both states $\rho_0$ and $\rho_1$ are efficiently preparable.

\paragraph{Hardness: Extending to $\ell = O(\log n)$.}
For larger $\ell$, no closed-form expression like that in \Cref{eqn:omega(V)-intro} is available. We thus use a different padding scheme via the POVM optimization. For simplicity, we omit the extra qubits initialized to $\ket{0}$ that are used to match register sizes.

A natural idea is to set $\rho_r \coloneq X_{r} \otimes \ketbra{0}{0} + \mu_{r} \otimes \ketbra{1}{1}$, in which case an optimal POVM for distinguishing $\cbra*{\rho_r}$ can also be chosen to have the same block structure, with each block optimized separately. Then, it suffices to choose $\mu_{r}$ so that the padding contributes a \emph{fixed} amount to the optimal discrimination probability. However, normalization requires that $\Tr(\mu_{r})=1-p_{r}$, making this contribution \emph{$p_r$-dependent}: the padding contributes zero when $p_{r}=1$ for every $r$, but a positive amount when $p_{r}=0$ for every $r$. 

We thus scale $X_r$ by $1 - 1/2^\ell$, leaving at least $1/2^\ell$ of the total weight for padding, and set
\[\rho_r \coloneq \frac{2^\ell - 1}{2^\ell}X_{r} \otimes \ketbra{0}{0} + \mu_{r} \otimes \ketbra{1}{1}, \quad\text{where } \mu_{r} \coloneq \frac{1}{2^\ell}\ketbra{r}{r} + \sum_{r' \ne r}\frac{1 - p_{r}}{2^\ell}\ketbra{r'}{r'}\enspace.\]
Consequently, the choice of $\mu_{r}$ ensures that $\rho_r$ is normalized, and the padding contributes a fixed amount to the optimal discrimination probability since, for each classical label, the largest weight assigned to that label among $\cbra*{\mu_{r}}$ is exactly $1/2^\ell$. Formally, the following expression relates the optimal discrimination probability over $\cbra*{\rho_r}$, optimized over POVMs $\cbra{N_r}$, to $\omega(V)$:
\begin{align*}
    \max_{\textnormal{POVM }\{N_r\}_{r}} \sum_{r} \Tr(N_r \rho_r) = \frac{2^\ell - 1}{2^\ell}\max_{\textnormal{POVM }\{M_r\}_{r}}\sum_{r}\Tr\rbra*{M_rX_{r}} + 1 = \frac{2^\ell - 1}{2^\ell} \cdot \omega(V) + 1\enspace.
\end{align*}
To prepare $\rho_r$ efficiently, choose $r'\in\binset^\ell$ uniformly at random. If $r'=r$, output the label $r$ with flag $1$. Otherwise, run the verifier with the fixed response $r$: upon acceptance, output $\sfM$ with flag $0$; upon rejection, output the label $r'$ with flag $1$.

\paragraph{Answer compression for $\QIP_{\ell\text{-}\bit}$.}
To avoid the parameter loss from the padding used for \MultiQSD{}, we return to \Cref{eqn:maximum-acceptance-prob-intro} and work directly with the $2^\ell$ subnormalized states $\cbra*{X_{r}}$. The goal is to encode these $2^\ell$ subnormalized states into two normalized states of a \QSD{} instance, which can then be decided by a \QIPbit{} proof system, thereby compressing the prover's response from $\ell$ bits to one bit.

Our construction replaces the task of guessing the $\ell$-bit label $r$ from $X_{r}$ as in \Cref{eqn:maximum-acceptance-prob-intro} with the task of guessing a single-bit label $h_{u, v}(r)$ from two states as in \QSD{} by using a pairwise independent hash function $h_{u, v}(r) = \innerprodF{u}{r}\oplus v$, where $u \in \binset^\ell$ and $v \in \binset$. Let $p \coloneq 2^{-\ell}\sum_r p_{r}$. We include the hash seed $(u,v)$ and add a padding for normalization in the output state, obtaining the following efficiently preparable states for $b\in\binset$:
\begin{align*}
    \rho_b \coloneq \frac{1}{2^{2\ell}}\ketbra{0}{0} \otimes \sum_{u, v, r \text{ s.t. } h_{u, v}(r) = b}\ketbra{u, v}{u, v} \otimes X_{r} + (1 - p)\ketbra{1}{1} \otimes \ketbra{\bar{0}}{\bar{0}}\enspace.
\end{align*}
For \emph{yes} instances, a successful guess of the $\ell$-bit label $r$ also gives the correct single-bit label $h_{u,v}(r)$. For \emph{no} instances, we use the fact that pairwise independent hash families give strong randomness extractors even
in the presence of quantum side information \cite[Theorem~6]{TSSR11}.
Adapting the argument to our construction, we bound how well $h_{u,v}(r)$ can be guessed in terms of how well $r$ can be guessed. Combining the two arguments, we obtain
\[
    \frac{\omega(V)}{2^\ell}\leq \TD(\rho_0,\rho_1)\leq \frac{1+2^{\ell/2}}{2^{\ell+1}}\omega(V)\enspace.
\]
This bound yields the desired answer compression when $c>(1+2^{\ell/2})s/2$.

\subsection{Proof techniques: Polarizing the \texorpdfstring{$\ell_1$}{}-norm distances in the natural regime}

To establish \Cref{thm:quantum-polarization-informal}, we first prove \Cref{thm:QSD-in-QSZK-improved} as a warm-up, showing that $\QSD[a,b] \in \QSZK$ when $a(n)^2-b(n)\geq 1/\poly(n)$. As in~\cite{BDRV19,Liu23}, our first idea is to reduce \QSD{} to the \emph{quantum Jensen--Shannon divergence} (\QJS{}), as introduced in~\cite[Section III]{MLP05}:
\begin{equation}
    \label{eq:QJS-def}
    \QJS(\rho_0,\rho_1)\coloneqq \S\rbra[\Big]{\frac{\rho_0+\rho_1}{2}} - \frac{\S(\rho_0)+\S(\rho_1)}{2} = \frac{1}{2} \rbra[\bigg]{ \S\rbra[\Big]{\frac{\rho_0+\rho_1}{2}\otimes\frac{\rho_0+\rho_1}{2}} - \S(\rho_0\otimes\rho_1) }\enspace.
\end{equation}

Since the quantum Jensen--Shannon divergence serves as a distance version of the quantum entropy difference $\S_2(\sigma_0)-\S_2(\sigma_1)$, as indicated in \Cref{eq:QJS-def}, combining this reduction with the \QSZK{} containment of the \textsc{Quantum Entropy Difference Problem} (\QED{}) in~\cite{BASTS10} would establish \Cref{thm:QSD-in-QSZK-improved}. Here, $\QED[g]$ asks whether $\S_2(\sigma_0)-\S_2(\sigma_1)$ is at least $g$ for \emph{yes} instances or at most $-g$ for \emph{no} instances, where $\S_2(\rho)\coloneqq \S(\rho)/\ln{2}$. However, as pointed out in~\cite[Remark~4.6]{Liu23}, a direct quantum analogue of~\cite{BDRV19} does not work. Let $\QJS_2(\rho_0,\rho_1) \coloneqq \QJS(\rho_0,\rho_1)/\ln{2}$. A straightforward reduction follows from the inequalities in \Cref{eq:QJS-vs-TD}, as shown in~\cite{Holevo73,FvdG99,BH09}\footnote{See also~\cite[Section 2.2]{Liu23} for a self-contained proof.} with $\H_2(x)$ denoting the binary entropy: 
\begin{equation}
    \label{eq:QJS-vs-TD}
    1-\H_2\rbra*{ \frac{1-\TD(\rho_0,\rho_1)}{2} } \leq \QJS_2(\rho_0,\rho_1) \leq \TD(\rho_0,\rho_1)\enspace.
\end{equation}
The straightforward reduction based on~\Cref{eq:QJS-vs-TD}, as in~\cite[Theorem 4.5]{Liu23}, shows that $\QSD[a,b]\in\QSZK$ when $a(n)^2-\sqrt{2\ln{2}}\, b(n) \geq 1/\poly(n)$. 

\paragraph{Parameterized bounds for the quantum Jensen--Shannon divergence.}
To remove the factor $\sqrt{2\ln{2}}$, we use our second idea, a \emph{parameterized} version of \Cref{eq:QJS-vs-TD}. To this end, we define the state $\rho_{j,\lambda} \coloneqq \frac{1+\lambda}{2} \rho_j + \frac{1-\lambda}{2} \rho_{1-j}$ for each $j\in\cbra{0,1}$, where $\rho_{j,\lambda}$ moves continuously from $\rho_+\coloneqq (\rho_0+\rho_1)/2$ toward $\rho_j$ as $\lambda$ increases from $0$ to $1$, in analogy with~\cite[Proposition 4.8]{BDRV19}. The parameterized bounds underlying the reduction in \Cref{thm:QSD-in-QSZK-improved} are as follows: 
\begin{equation}
    \label{eq:para-QJS-vs-TD}
    \frac{\lambda^2}{2\ln2}\TD(\rho_0,\rho_1)^2
    \leq\QJS_2(\rho_{0,\lambda},\rho_{1,\lambda})
    \leq\frac{\lambda^2}{2\ln2(1-\lambda^2)} \TD(\rho_0,\rho_1)\enspace.
\end{equation}
Here, the lower bound follows directly from the quantum Pinsker inequality (see, e.g.,~\cite[Theorem 5.38]{Watrous18}), while the upper bound is more challenging: differentiate the divergence along the smoothing path (via the Daleckii--Krein formula, see, e.g.,~\cite[Section 5.3.1]{Bhatia09}), use operator monotonicity of the logarithm to bound the derivative (see, e.g.,~\cite[Section 5.3.7]{Bhatia09}), and integrate. Choosing a dyadic $\lambda^2 = \Theta(\delta)$, where $\delta\coloneqq 1/\poly(n)$ satisfies $a^2-b \geq \delta$, the \QSZK{} containment then follows as in~\cite[Theorem 4.5]{Liu23} via the reduction based on \Cref{eq:para-QJS-vs-TD}, yielding a $\QED[g]$ instance with promise gap $g(n)=\Theta(\delta^2)$. 

\paragraph{Signed combinations of \QJS{} between parameterized quantum states.}
Passing through entropy difference alone need not extend polarization beyond the regime $a^2>b$: the approach of~\cite{BDRV19} already uses this machinery and is covered by the oblivious black-box barrier of~\cite[Corollary~2]{HR05}, while our quantum adaptation in \Cref{sec:QSD-in-QSZK} also uses entropy difference, but the resulting bounds on a single Jensen--Shannon divergence still compare $a^2$ with $b$.

To overcome this barrier, our third idea is to exploit entropy additivity to encode a \emph{signed linear combination} of \QJS{} between parameterized quantum states.
Cancellation between positive and negative terms allows an additive-error approximation of the trace distance itself, making the required separation depend on $a-b$.
After subtracting a threshold and normalizing, we encode the approximation as an entropy difference, implementing negative coefficients by exchanging the corresponding components between the two constructed states.

Let $\Phi(x)\coloneqq1-\H_2((1+x)/2)$. The core inequalities underlying our quantum polarization in the natural regime are as follows: for every dyadic precision $\varepsilon\in(0,1]$, there exist a positive integer $J=O(1/\varepsilon)$ and dyadic numbers $c_j\in\bbR$ and $\lambda_j\in[0,1/2]$ for $j\in[J]$ such that
\begin{equation}
    \label{eq:TD-approx-QJScomb}
    \abs*{\sum_{j=1}^J c_j \QJS_2(\rho_{0,\lambda_j},\rho_{1,\lambda_j})-\TD(\rho_0,\rho_1)}
    \leq \sup_{\abs{x}\leq1}\abs*{\sum_{j=1}^Jc_j\Phi(\lambda_jx)-\abs{x}}
    \leq\varepsilon\enspace.
\end{equation}

Our strategy is inspired by quantum algorithms that use (uniform) polynomial approximations to estimate trace distance and its Schatten norm generalizations~\cite{WZ24,LGLW23,LW25Lalpha}. In particular, Chebyshev-based implementations~\cite{MY23,LGLW23} of quantum singular value transformation~\cite{GSLW19} using the linear-combination-of-unitaries method~\cite{CW12,BCCKS15} suggest constructing approximations from signed combinations of simpler functions. Our argument has two parts:
\begin{enumerate}[label={\upshape(\roman*)}]
    \item \emph{Scalar error bound for trace distance approximation}. The first inequality in \Cref{eq:TD-approx-QJScomb} bounds the difference between $\sum_{j=1}^J c_j \QJS_2(\rho_{0,\lambda_j},\rho_{1,\lambda_j})$ and $\TD(\rho_0,\rho_1)$ by the difference between the scalar functions $\sum_{j=1}^Jc_j\Phi(\lambda_jx)$ and $|x|$. The key ingredient is the \emph{smoothed integral representation}  of \QJS{}, implicit in~\cite{HircheTomamichel24}, expressed in terms of the symmetrized version of the \emph{quantum hockey-stick divergence} introduced in~\cite{SW13}. 
    \item \emph{Efficient signed approximation of the absolute value function}. The second inequality in \Cref{eq:TD-approx-QJScomb} is a \emph{uniform approximation} of the form $\sum_{j=1}^Jc_j\Phi(\lambda_jx)$ to $|x|$. The starting point is an even polynomial $P$ of degree $O(1/\varepsilon)$, with $P(0)=0$ and \emph{efficiently computable} coefficients that uniformly approximates $|x|$, as in~\cite[Lemma~3.1]{LW25}.\footnote{The existence of such a uniform polynomial approximation dates back more than a century to~\cite{Bernstein14}.} The coefficients $\cbra{c_j}$ are then obtained by matching the coefficients of $P$ and the truncation of $\sum_{j=1}^Jc_j\Phi(\lambda_jx)$ using the expansion $\Phi(x)=\sum_{k=1}^{\infty}\frac{x^{2k}}{(2k-1)(2k)\ln2}$. The resulting Vandermonde system determines the coefficients, which we subsequently round to dyadic numbers.
\end{enumerate}

The main quantitative issue is the $\ell_1$ norm of the coefficient vector $\norm{\bfc}_1 = \sum_{j=1}^J\abs{c_j}$. For our choice $\lambda_j=j/(16J)$, the underlying construction guarantees $\norm{\bfc}_1\leq 2^{O(1/\varepsilon)}$. Choosing $\Upsilon$ as the smallest power of two satisfying $\Upsilon\geq \norm{\bfc}_1+\tau$, where $\tau\in[0,1]$ is dyadic and satisfies $\abs*{\tau-(a+b)/2}\leq\varepsilon$, therefore gives $\Upsilon=2^{O(1/\varepsilon)}$.

Finally, it remains to construct state-preparation circuits of the quantum states $\rho'_0$ and $\rho'_1$ that enable the use of the error bound \Cref{eq:TD-approx-QJScomb} and satisfy
\begin{equation}
    \label{eq:TD-QJS-QED}
    \S_2(\rho'_0) - \S_2(\rho'_1) = \frac{1}{\Upsilon}\sum_{j=1}^J c_j \QJS_2(\rho_{0,\lambda_j}, \rho_{1,\lambda_j}) - \frac{\tau}{\Upsilon}\enspace.
\end{equation}
Consequently, for sufficiently small $\varepsilon=\Theta(\Delta)$, where the promise gap $\Delta\coloneq a-b$, the resulting $\QED[g]$ instance has promise gap $g(n) = \varepsilon(n)/\Upsilon \geq \Delta \cdot 2^{-O(1/\Delta)}$, which remains inverse-polynomial as long as $\Delta(n)\geq 1/O(\log{n})$, as desired. 

In addition, the classical result in \Cref{thm:quantum-polarization-informal} follows from the same scalar approximation in the second inequality of \Cref{eq:TD-approx-QJScomb} and the entropy encoding in \Cref{eq:TD-QJS-QED}.

\subsection{Proof techniques: Quantum public coins also make interaction useless}

Our starting point is a complete problem for $\qcQAM[\ell]$ via the \emph{steering-game value}. We first define a steering game as follows:\footnote{The word ``steering'' refers to Alice's ability to prepare different conditional states on Bob's qubits by measuring her own qubits, while Bob's average state remains maximally mixed.} Alice and Bob share $m$ EPR pairs, with Bob's halves stored in an $m$-qubit quantum register $\sfR$ and $d_\sfR \coloneqq \dim(\sfR)$. Alice performs an arbitrary POVM $\cbra{M_r}_{r\in\binset^\ell}$ on her halves of the EPR pairs and sends the outcome $r\in\binset^\ell$. 
For each label $r$, Bob then performs the binary test with acceptance operator $\Pi_r$ on $\sfR$, and Alice wins if the test accepts. 
Writing $\sigma_r\coloneqq M_r^T/d_\sfR$ for the subnormalized conditional state on $\sfR$ corresponding to outcome $r$, the feasible families $\cbra{\sigma_r}$ in \Cref{eq:MSV-value} are in one-to-one correspondence with Alice's measurements $\cbra{M_r}$. Hence, the \emph{steering-game value} $\MSV{\ell}{\Pi}$, defined as the maximum winning probability over all of Alice's measurements, can be written as
\begin{equation}
    \label{eq:MSV-value}
    \MSV{\ell}{\Pi} \coloneqq \max_{\sigma_r\succeq 0,\,
    \sum_r\sigma_r=I_{\sfR}/d_{\sfR}}
    \sum_{r\in\binset^\ell}\Tr(\Pi_r\sigma_r),
    \quad\text{where } \Pi \coloneqq \cbra{\Pi_r}_{r\in\binset^\ell}\enspace.
\end{equation}
In the $\qcQAM[\ell]$ setting, Alice and Bob correspond to the prover and verifier in the proof system, respectively, and we are interested in the regime where $\ell(n)=O(\log n)$ yet $m(n)=\poly(n)$. We thus define the promise problem $\SteerVal[\ell,a,b]$, which asks whether $\MSV{\ell}{\Pi}$ is at least $a(n)$ or at most $b(n)$, and show $\qcQAM[\ell,a,b]$-completeness in \Cref{lemma:SteerVal-qcQAM[l]-complete}.

\paragraph{Warm-up: $\qcQAM[1] = \BQP$}
We begin with the single-bit case $\ell=1$, for which we write $\BSV{\Pi_0}{\Pi_1}\coloneqq\MSV{1}{\cbra*{\Pi_0,\Pi_1}}$ for simplicity. In this case, we only need to optimize over a single subnormalized state $\sigma_0$, since $\sigma_1=I_{\sfR}/d_{\sfR}-\sigma_0$. In fact, the Holevo--Helstrom formula for \emph{binary} generalized state discrimination~\cite{AM14} gives the following closed-form expression:
\begin{equation}
    \label{eq:BSV-value}
    \BSV{\Pi_0}{\Pi_1} = \frac{\Tr(\Pi_0 + \Pi_1)+\norm{\Pi_0-\Pi_1}_1}{2d_{\sfR}}\enspace.
\end{equation}
The expression in \Cref{eq:BSV-value} eliminates the optimization over the prover's strategy, and the normalization by the dimension $d_{\sfR}$ enables efficient quantum algorithms for estimating $\Tr(\Pi_0)/d_\sfR$ and $\Tr(\Pi_1)/d_\sfR$. Since both $\Pi_0$ and $\Pi_1$ admit efficient block encodings constructed from the circuit $Q$,\footnote{See \Cref{def:block-encoding} for the definition of the block-encoding.} such algorithms follow directly from normalized trace estimation techniques originally developed in the context of \textsf{DQC1}~\cite{Shepherd06,SJ08}. This normalized trace estimation can be implemented by the block-encoded Hadamard test~\cite[Lemma~9]{GP22} (i.e., the single-bit precision phase estimation~\cite{Kitaev95}; see also~\cite{AJL09}) on the maximally mixed state. 

The more challenging part is estimating the term $\norm{\Pi_0-\Pi_1}_1/d_{\sfR}=\Tr(\abs{\Pi_0-\Pi_1})/d_{\sfR}$, which requires more sophisticated techniques. The block-encoding of $\Pi_0-\Pi_1$ follows directly from the linear-combination-of-unitaries (LCU) lemma~\cite{CW12,BCCKS15}. The block-encoding of $\abs{\Pi_0-\Pi_1}$ can then be implemented approximately using the quantum singular value transformation (QSVT)~\cite{GSLW19}, together with a uniform polynomial approximation of the absolute value function whose coefficients can be computed efficiently~\cite[Lemma 3.1]{LW25}. Combining all these ingredients with amplitude estimation~\cite{BHMT02} allows us to estimate $\BSV{\Pi_0}{\Pi_1}$ to within additive error $1/\poly(n)$, thereby establishing the \BQP{} containment.

\paragraph{The general case: $\qcQAM[\sqrt{\log n}, c, s] = \BQP$ for constant gap $c - s$.}
For $\ell > 1$, however, the optimization is over $2^\ell$ subnormalized states, and the closed-form expression for $\BSV{\Pi_0}{\Pi_1}$ does not extend directly. 
Instead, we estimate $\MSV{\ell}{\Pi}$ using the \emph{matrix multiplicative weights update} (MMWU) framework, developed by Arora and Kale~\cite{AK07}. Notably, MMWU has been used extensively in quantum complexity theory, particularly in the line of work that eventually led to $\QIP=\PSPACE$~\cite{JW09,JUW09,JJUW11,Wu10}.

We first combine the subnormalized states $\{\sigma_r\}_{r \in \binset^\ell}$ and acceptance operators $\{\Pi_r\}_{r \in \binset^\ell}$ into $\sigma \coloneq \sum_r \ketbra{r}{r}_{\sfI} \otimes \sigma_r$ and $C \coloneq \sum_r \ketbra{r}{r}_{\sfI} \otimes \Pi_r$. This standard block-diagonal representation converts the objective in \Cref{eq:MSV-value} into the linear functional $\Tr(C\sigma)$ over a single state $\sigma$ and rewrites the constraint into a constraint $\Tr_{\sfI}(\sigma)=I_{\sfR}/d_{\sfR}$ over the state $\sigma$. Formally, 
\begin{align}\label{eqn:MSV-SDP-intro}
\MSV{\ell}{Q}=\max_{\substack{\sigma \in \Dens(\sfI,\sfR)\\\Tr_{\sfI}(\sigma)=I_{\sfR}/d_{\sfR}}}\Tr(C\sigma)\enspace.
\end{align}
The rounding theorem~\cite[Theorem~3.1]{GW13} allows us to replace the constraint in \Cref{eqn:MSV-SDP-intro} with a trace-norm penalty at the cost of an additive error $\varepsilon$. We therefore focus on estimating the relaxed value $\max_{\sigma \in \Dens(\sfI,\sfR)}\calF_{\varepsilon}(\sigma)$, where $\calF_{\varepsilon}(\sigma) \coloneq \Tr(C\sigma) - \frac{1}{\varepsilon}\norm{\Tr_{\sfI}(\sigma) - I_{\sfR}/d_{\sfR}}_1$.

Following~\cite{GW13}, we use the dual characterization of the trace norm
\begin{align}\label{eqn:dual-characterization-trace-norm-intro}
    \max_{H = H^\dagger,\, \norm{H} \leq 1}\Tr\rbra*{H\rbra*{\Tr_{\sfI}(\rho) - I_{\sfR}/d_{\sfR}}} = \norm{\Tr_{\sfI}(\rho)-I_{\sfR}/d_{\sfR}}_1
\end{align}
to obtain linear feedback for MMWU on $\calF_{\varepsilon}$. In particular, the Hermitian contraction $H$ attaining the maximum in \Cref{eqn:dual-characterization-trace-norm-intro} provides a linear feedback for MMWU. For ease of implementation, an Hermitian contraction $H'$ that approximately attains the maximum in \Cref{eqn:dual-characterization-trace-norm-intro} also suffices, as in~\cite{MY23}. For constant target accuracy and appropriately chosen parameters, the standard MMWU analysis yields an estimate $\widetilde{w}$ of $\MSV{\ell}{Q}$ in $T=O(\ell+\log d_{\sfR})$ iterations.

However, computing $\widetilde{w}$ requires time exponential in $T$, and thus requires exponential time in $n$: the block-encoding in each iteration recursively invokes those from earlier iterations, causing the gate complexity to grow exponentially in $T$. We therefore use the refined analysis of MMWU in terms of quantum relative entropy \cite{WK06,ALO15}, which shows that an estimate $\widetilde{w}$ of $\MSV{\ell}{\Pi}$ can be obtained in $T$ iterations with $T = O(D(\sigma^* \| \sigma_1)) = O(\ell)$, where $\sigma^*$ is the optimal solution and $\sigma_1 \coloneq I_{\sfI\sfR}/\rbra*{2^\ell d_{\sfR}}$ is our initialization.

To obtain the estimate $\widetilde{w}$, we implement the $T$ iterations using QSVT~\cite{GSLW19}, with polynomial approximations to the exponential function~\cite[Corollary~64]{GSLW19} to construct the block-encodings of successive rescaled Gibbs states, and polynomial approximations to the sign function~\cite[Corollary~6]{LC17} to construct the block-encodings of Hermitian contractions $H'$ that approximately attain the maximum in \Cref{eqn:dual-characterization-trace-norm-intro}. By combining these techniques with LCU and normalized trace estimation and appropriately choosing the precision parameters, we obtain a quantum algorithm for estimating $\MSV{\ell}{\Pi}$ up to constant additive error with time complexity $\poly(n,\ell, 1/\Delta)\exp(O(\ell^2/\Delta^6))$, establishing the \BQP{} containment. 

\subsection{Discussion and open problems}

While classical public coins make the interaction in quantum interactive proof systems with a laconic prover useless~\cite[Section~5]{BSW11}, when the prover's response length is $\ell(n)=O(\log{n})$ and the promise gap satisfies $c(n)-s(n)\geq 1/\poly(n)$,  the \BQP{} containment of the quantum-public-coin analog holds only for $\ell(n)=O\rbra*{\sqrt{\log{n}}}$ and constant promise gap in \Cref{intro:thm:public-one-bit}, which leads to a natural question: 
\begin{enumerate}[label={\upshape(\alph*)}]
    \setcounter{enumi}{0}
    \item \emph{Are two-message quantum-public-coin interactive proof systems with a laconic prover, with the prover's response length $\ell=O(\log{n})$ and the promise gap $c(n)-s(n)\geq 1/\poly(n)$, contained in $\BQP$?}
    \label{probitem:laconic-quantum-public-coins}
\end{enumerate}

While \Cref{thm:quantum-polarization-informal} provides techniques for polarizing $\ell_1$-norm distances in the natural regime $a(n)-b(n)\geq 1/O(\log{n})$ and resolves the first open problem posed in~\cite{SV97}, the analogous question remains open for the inverse-polynomial regime: 

\begin{enumerate}[label={\upshape(\alph*)}]
    \setcounter{enumi}{1}
    \item \emph{Do $\SD[a,b]\in\SZK$ and $\QSD[a,b]\in\QSZK$
    hold for the regime $a(n)-b(n)\geq 1/\poly(n)$?}
    \label{probitem:polarization}
\end{enumerate}

A limitation of our approach is the normalization factor $\Upsilon$ in \Cref{eq:TD-QJS-QED}. Indeed, every signed linear combination $\sum_{j=1}^J c_j\Phi(\lambda_jx)$ with $\lambda_j\in[0,1/2]$ that uniformly approximates $\abs{x}$ on $[-1,1]$ to error at most $\varepsilon$ satisfies $\norm{\bfc}_1\geq 2^{\Omega(1/\varepsilon)}$ as $\varepsilon\to0$.\footnote{Truncating the Taylor series of this linear combination at degree $O(\log(\norm{\bfc}_1/\varepsilon))$ introduces additional error at most $\varepsilon$, giving a polynomial approximation of $\abs{x}$ with error at most $2\varepsilon$. Bernstein's lower bound~\cite{Bernstein14} therefore implies $\log(\norm{\bfc}_1/\varepsilon)=\Omega(1/\varepsilon)$. Since $\log(1/\varepsilon)=o(1/\varepsilon)$, this bound yields $\log\norm{\bfc}_1=\Omega(1/\varepsilon)$.}
Therefore, our coefficient bound is \emph{optimal} up to constants in the exponent \emph{within this approximation family}, and resolving Question~\ref{probitem:polarization} thus requires new ideas beyond our approach.


\section{Preliminaries}

For any positive integer $n$, we adopt the notation $\sbra{n} \coloneqq \cbra{1,2,\dots,n}$. We also use $\log(x)$ to denote the base-$2$ logarithm and $\ln(x)$ to denote the natural logarithm.

For a quantum register $\sfA$, we write $\Pos(\sfA)$ for the set of positive semidefinite operators acting on the Hilbert space associated with $\sfA$, and $\Dens(\sfA) \coloneq \set{\rho \in \Pos(\sfA)}{\Tr(\rho)=1}$ for the set of density operators on $\sfA$. We slightly abuse notation by using $\sfA$ to denote both the register and its associated Hilbert space.
In addition, we denote the maximally entangled state on two $m$-qubit registers by $\ket{\Phi_{2^m}} \coloneqq 2^{-m/2} \sum_{j\in[2^m]} \ket{j}\ket{j}$.
We also use $\ket{\bar{0}}$ to denote $\ket{0}^{\otimes a}$, where $a > 1$ is at most $\poly(n)$, with the relevant length parameter $n$ determined by the context.

\paragraph{Schatten norm and a matrix H\"older inequality.}
For $1\leq p < \infty$, the Schatten $p$-norm of a matrix $A$ is defined by $\norm{A}_{p} \coloneqq \rbra[\big]{\Tr\rbra{\abs{A}^{p}}}^{1/p}$, where $\abs{A} \coloneqq \sqrt{A^\dagger A}$.
The case $p=1$ is called the \textit{trace norm} $\norm{A}_1=\Tr|A|$, while the case $p=\infty$, obtained by taking the limit $p \to \infty$, is called the \textit{operator norm} $\norm{A} \coloneqq \norm{A}_{\infty} = \sigma_{\max}(A)$, where $\sigma_{\max}(A)$ is the largest singular value of $A$. These norms are connected by the following version of the matrix H\"older inequality:
\begin{lemma}[H\"older inequality for Schatten norms, adapted from~{\cite[Equation 1.174]{Watrous18}}]
    \label{lemma:matrix-Holder}
    For each $p\in[1,\infty]$, let $q\in[1,\infty]$ satisfy $\frac{1}{p}+\frac{1}{q}=1$. For every matrix $A$, the Schatten $p$-norm and Schatten $q$-norm are \emph{dual} to each other. Consequently, for all matrices $B$, we have
    \[ \abs*{\Tr\rbra[\big]{B^\dagger A}} \leq \norm{A}_p\norm{B}_q\enspace. \]
\end{lemma}

\subsection{Complexity classes}

We adapt the standard definition for two-message quantum interactive proof systems from \cite[Section 3.1]{JUW09}:

\begin{definition}[Two-message quantum interactive proof systems, \QIPtwo{}]\label{def:QIPtwo}
Let $c(n)$ and $s(n)$ be efficiently computable functions of the input length $n \coloneqq |x|$ such that $0 \leq s(n) < c(n) \leq 1$. 
A promise problem $\calI = (\calI_{\yes}, \calI_{\no})$ is in $\QIP\sbra*{2,c,s}$, if there exists a quantum verifier $V(x) = (V_1(x), V_2(x))$ acting on the message register $\sfM$ and its private register $\sfW$ such that:
    \begin{itemize}
        \item \textbf{\emph{Completeness}}. For any $x \in \calI_{\yes}$, there exists a prover $P(x)$, acting on its private register $\sfQ$ and the message register $\sfM$ such that
        \[\Pr\sbra*{(\protocol{P}{V})(x)\text{ accepts}} \geq c(n)\enspace.\]
        \item \textbf{\emph{Soundness}}. For any $x \in \calI_{\no}$ and any prover $P(x)$ acting on the registers $(\sfQ,\sfM)$, 
        \[\Pr\sbra*{(\protocol{P}{V})(x)\text{ accepts}} \leq s(n)\enspace.\]
    \end{itemize}

\noindent Furthermore, we define $\QIPtwo \coloneqq \QIP\sbra*{2,2/3,1/3}$, define $\QIP_{\ell\text{-}\qubit}\sbra*{2,c,s}$ to be the restriction of $\QIP\sbra*{2,c,s}$ in which the prover sends only an $\ell$-qubit quantum state, and define $\QIP_{\ell\text{-}\bit}\sbra*{2,c,s}$ to be the restriction of $\QIP\sbra*{2,c,s}$ in which the prover sends only a computational-basis state $\ket{a}$ with $a \in \binset^\ell$. We also define $\QIPbit\sbra*{c,s} \coloneqq \QIP_{1\text{-}\bit}\sbra*{2,c,s}$.
\end{definition}

In this work, we also study two-message quantum interactive proof systems with quantum public coins (the EPR pairs). We adapt the definition of generalized Arthur--Merlin proof systems from \cite[Section 3]{KLGN19}:
\begin{definition}[\qcQAM{}]
\label{def:qcQAM}
Let $c(n)$ and $s(n)$ be efficiently computable functions of the input length $n \coloneqq |x|$ such that $0 \leq s(n) < c(n) \leq 1$. Let $k(n)$ be an efficiently computable function of the input length $n$ such that $k(n) = \poly(n)$.
A promise problem $\calI = (\calI_{\yes}, \calI_{\no})$ is in $\qcQAM\sbra*{k,c,s}$, if there exists a quantum verifier $V(x)=(V_1(x),V_2(x))$ and a polynomial $p$, where $V_1$ prepares $p(n)$ EPR pairs and sends the second halves of them to the prover, and stores the first halves of them in $\sfV$, and $V_2$ acts on the message register $\sfM$ and $\sfV$, such that\emph{:}
    \begin{itemize}
        \item \textbf{\emph{Completeness}}. For any $x \in \calI_{\yes}$, there exists a prover $P(x)$ such that, upon receiving the halves of the EPR pairs, $P(x)$ returns a $k(n)$-bit string $a \in \binset^{k(n)}$ and
        \[\Pr\sbra*{(\protocol{P}{V})(x)\text{ accepts}} \geq c(n)\enspace.\]
        \item \textbf{\emph{Soundness}}. For any $x \in \calI_{\no}$ and any prover $P(x)$ that returns a $k(n)$-bit string, 
        \[\Pr\sbra*{(\protocol{P}{V})(x)\text{ accepts}} \leq s(n)\enspace.\]
    \end{itemize}
\noindent For convenience, we adopt the notation $\qcQAM[k] \coloneqq \bigcup_{c(n)-s(n)\geq 1/\poly(n)} \qcQAM[k,c(n),s(n)]$.
\end{definition}

Analogously to \Cref{def:qcQAM}, one can define the class \qqQAM{}, where the only difference is that the prover's response is a $k_1(n)$-qubit \emph{quantum state} rather than a $k_2(n)$-bit string. While $\qqQAM=\qcQAM$ is known from~\cite[Theorem 1.7(ii)]{KLGN19}, the quantitative subtlety is that the equivalence holds only for $k_2=2k_1$ (see \Cref{cor:qqQAMell=qcQAM2ell}).

\subsection{Quantum state distinguishability and mixed-state preparation}

We begin by defining the \textsc{Quantum State Distinguishability Problem} (\QSD{}): 

\begin{definition}[Quantum State Distinguishability, \QSD{}, adapted from~{\cite[Section 3.3]{Watrous02}}]
    \label{def:QSD}
    Let $(Q_0, Q_1)$ be a pair of polynomial-size quantum circuits with total description length $n$ such that $Q_0$ and $Q_1$ act on $m$ qubits with $k$ specified output qubits, where $m(n)$ and $k(n)$ are polynomial in $n$. Let $\rho_0$ and $\rho_1$ be the output states obtained by running $Q_0$ and $Q_1$ on $\ket{0}^{\otimes m}$, respectively, and tracing out all non-output qubits. Let $a(n)$ and $b(n)$ be efficiently computable functions such that $0 \leq b(n) < a(n) \leq 1$. The promise problem $\QSD[a(n),b(n)]$ asks to distinguish between the following two cases: 
    \begin{itemize}
        \item \emph{Yes:} The pair of quantum circuits $(Q_0,Q_1)$ satisfies $\TD(\rho_0,\rho_1) \geq a(n)$;
        \item \emph{No:} The pair of quantum circuits $(Q_0,Q_1)$ satisfies $\TD(\rho_0,\rho_1) \leq b(n)$. 
    \end{itemize}
\end{definition}

We also need a \emph{mixed-state} analogue of the linear-combination-of-unitaries
lemma~\cite{CW12,BCCKS15} with nonnegative dyadic coefficients:

\begin{lemma}[Preparing dyadic convex combinations of quantum states]
    \label{lemma:dyadic-linear-combi-states}
    Let $J$ be a positive integer. For each $j\in[J]$, let $Q_j$ be a quantum circuit with $r$ specified output qubits. Let $\rho_j$ denote the state obtained by applying $Q_j$ to the input state $\ket{\bar{0}}$ and tracing out all non-output qubits. Let $p_1,\ldots,p_J$ be nonnegative dyadic numbers satisfying $\sum_{j=1}^Jp_j=1$, and write $p_j=t_j/2^L$ for nonnegative integers $t_j$ and a common positive integer $L$. One can construct quantum circuits $Q$ and $\widehat Q$ whose output states are, respectively,
    \[
        \rho=\sum_{j=1}^Jp_j\rho_j
        \quad\text{and}\quad
        \widehat\rho
        =\sum_{j=1}^Jp_j\ketbra{j-1}{j-1}_{\sfF}\otimes\rho_j\enspace.
    \]
    Here, the label register $\sfF$, appearing in $\widehat{Q}$ and $\widehat{\rho}$, uses $\ceil*{\log J}$ qubits. Each circuit queries each $Q_j$ once and uses $O\rbra*{J(L+r+\log J)}$ additional one- and two-qubit quantum gates. The circuit descriptions can be computed in time polynomial in the total input description length, $r$, and $L$.
\end{lemma}

\begin{proof}
We construct $\widehat Q$ as specified in \circuitref{circuit:dyadic-convex-combi-states}.

\begingroup
\LinesNumbered
\begin{algorithm}[!ht]
    \SetAlgorithmName{Circuit}{Algorithm}{List of Algorithms}
    \caption{Preparing a dyadic convex combination of quantum states.}
    \label{circuit:dyadic-convex-combi-states}
    \SetEndCharOfAlgoLine{.}
    \setlength{\parskip}{5pt}
    \SetKwFor{For}{For}{:}{}
    \SetKwComment{Comment}{// }{}
    \SetKwInOut{Input}{Input}
    \SetKwInOut{Output}{Output}
    \SetKwInOut{Registers}{Registers}

    \Input{Quantum circuits $Q_j$, each with $r$ specified
    output qubits, and nonnegative integers $t_j$ for $j\in[J]$,
    where $\sum_{j=1}^Jt_j=2^L$.}
    \Output{The state $\widehat\rho$ on $(\sfF,\sfA)$,
    where $\sfF$ and $\sfA$ use $\ceil*{\log J}$ and $r$
    qubits, respectively.}
    \Registers{All registers are initialized to $\ket{\bar{0}}$, including the label register $\sfF$, an $L$-qubit index register $\sfI$, and a $J$-qubit control register $\sfC$.}

    Apply $H^{\otimes L}$ to the register $\sfI$, preparing $2^{-L/2}\sum_{x=0}^{2^L-1}\ket{x}_{\sfI}$\;

    Controlled on the register $\sfI$, flip the corresponding qubits of the register $\sfC$ to implement
    \[
        \ket{x}_{\sfI}\ket{\bar{0}}_{\sfC}
        \mapsto
        \ket{x}_{\sfI}
        \ket{c_1(x),\ldots,c_J(x)}_{\sfC},
        \quad\text{where }
        c_j(x)={\bf1}_{T_{j-1}\leq x<T_j}\enspace,
    \]
    for every $x\in\{0,\ldots,2^L-1\}$. Here, $T_0\coloneqq 0$ and $T_j \coloneqq \sum_{k=1}^jt_k$ for each $j\in[J]$\;

    \For{each $j\in[J]$}{
        Apply $Q_j$ to the registers $(\sfA_j,\sfR_j)$, where $\sfA_j$ contains the specified output qubits\;

        Controlled on the $j$-th qubit of the register $\sfC$ being $1$, apply $\ket{f}_{\sfF}\mapsto\ket{f\oplus(j-1)}_{\sfF}$, and swap the registers $(\sfA_j, \sfA)$\;
    }

    Output $(\sfF,\sfA)$ and trace out all other registers\;
\end{algorithm}
\endgroup

We now show the correctness of \circuitref{circuit:dyadic-convex-combi-states}. 
For each $j\in[J]$, write $\ket{\psi_j} \coloneqq Q_j\ket{\bar{0}}$. 
For each integer $x\in[T_{j-1}, T_j)$, we have $c_j(x)=1$ and $c_{j'}(x)=0$ for every $j' \neq j$. 
Thus, the joint state immediately before Line 6 in \circuitref{circuit:dyadic-convex-combi-states} is
\[ \ket{\Psi} = 2^{-L/2}
    \sum_{j=1}^J\sum_{x=T_{j-1}}^{T_j-1} \ket{x}_{\sfI} \ket{c_1(x),\ldots,c_J(x)}_{\sfC} \ket{j-1}_{\sfF}
    \otimes\ket{\psi_j}_{\sfA\sfR_j}\otimes\ket{\bar{0}}_{\sfA_j}
    \otimes\bigotimes_{\substack{j'\in[J]\\j'\ne j}}\ket{\psi_{j'}}_{\sfA_{j'}\sfR_{j'}}\enspace. \]
Consequently, a direct calculation gives
\begin{align*}
    \Tr_{(\sfI,\sfC,\sfA_1,\sfR_1,\ldots,\sfA_J,\sfR_J)} \rbra*{ \ketbra{\Psi}{\Psi} }
    &= 2^{-L}\sum_{j=1}^J\sum_{x=T_{j-1}}^{T_j-1} \ketbra{j-1}{j-1}_{\sfF}
        \otimes \Tr_{\sfR_j} \rbra*{ \ketbra{\psi_j}{\psi_j}_{\sfA\sfR_j} }\\
    &= \sum_{j=1}^J\frac{T_j-T_{j-1}}{2^L} \ketbra{j-1}{j-1}_{\sfF}\otimes\rho_j\\
    &= \sum_{j=1}^J p_j \ketbra{j-1}{j-1}_{\sfF}\otimes\rho_j
    =\widehat{\rho}\enspace.
\end{align*}
Here, the first line uses $\Tr(\ketbra{x}{y})=\delta_{x,y}$ when tracing out $\sfI$ and $\Tr(\ketbra{\psi_{j'}}{\psi_{j'}})=1$ when tracing out $(\sfA_{j'},\sfR_{j'})$ for each $j' \neq j$, the second line follows from $\Tr_{\sfR_j}(\ketbra{\psi_j}{\psi_j}) = \rho_j$, and the last line follows from $(T_j-T_{j-1})/2^L = t_j/2^L = p_j$. Further tracing out $\sfF$ gives $\Tr_{\sfF}(\widehat{\rho}) = \sum_{j=1}^J p_j \rho_j =\rho.$

Each circuit applies every $Q_j$ once. Using reversible interval tests, decompositions of Toffoli gates into one- and two-qubit gates, and controlled swaps, the additional gate count is
\[
    \underbrace{L}_{\text{Line 1}} + \underbrace{O(JL)}_{\text{Line 2: interval test}} + \sum_{j=1}^J \rbra*{
        \underbrace{O(\log J)}_{\text{Line 5: label}}
        +\underbrace{O(r)}_{\text{Line 5: swap}}
    }
    =O\rbra*{J(L+r+\log J)}\enspace.
\]
Computing the endpoints $T_j$ and generating the two circuit descriptions of $Q$ and $\widehat Q$ take time polynomial in the total input description length, $r$, and $L$.
\end{proof}

\subsection{Information-theoretic measures for distributions and states}

We begin by introducing the classical and quantum $\ell_1$-norm distances: 

\begin{definition}[Total variation distance]
	\label{def:statDist}
    Let $p_0$ and $p_1$ be two probability distributions over a finite set $S$. The total variation distance between $p_0$ and $p_1$ is defined by
	\[\TV(p_0,p_1)  \coloneqq  \frac{1}{2}\|p_0-p_1\|_1 = \frac{1}{2}\sum_{x\in S} |p_0(x)-p_1(x)|\enspace.\]
\end{definition}

\begin{definition}[Trace distance]
    Let $\rho_0$ and $\rho_1$ be quantum states of the same dimension. 
    The trace distance between $\rho_0$ and $\rho_1$ is defined by 
     \[\TD(\rho_0,\rho_1) \coloneqq \frac{1}{2} \norm{\rho_0-\rho_1}_1 = \frac{1}{2}\Tr\rbra{|\rho_0-\rho_1|}\enspace.\]
\end{definition}

We then introduce classical and quantum entropies, together with their symmetrized distance versions and some basic properties:

\begin{definition}[Shannon and binary entropies]
    Let $p$ be a probability distribution over a finite set $S$. The Shannon entropy of $p$ is defined by 
    \[ \H(p) \coloneqq -\sum_{x\in S} p(x) \ln p(x) \enspace. \]
    For a binary distribution $p=(x,1-x)$, $\H(p)$ becomes the (Shannon) binary entropy: 
    \[\H(x) \coloneqq -x \ln{x} -(1-x) \ln(1-x) \enspace.\]
    For convenience, we use $\H_2(p) \coloneqq \H(p)/\ln{2}$ and $\H_2(x) \coloneqq \H(x)/\ln{2}$ to denote the Shannon and binary entropies using base-$2$ logarithms.
\end{definition}

\begin{definition}[Quantum Jensen--Shannon divergence, adapted from~{\cite[Section~III]{MLP05}}]
    Let $\rho_0$ and $\rho_1$ be quantum states of the same dimension. 
    The quantum Jensen--Shannon divergence between $\rho_0$ and $\rho_1$ is defined by 
    \[\QJS(\rho_0,\rho_1)  \coloneqq  \S\rbra[\Big]{\frac{\rho_0+\rho_1}{2}} -\frac{\S(\rho_0)+\S(\rho_1)}{2}=\frac{1}{2}\rbra*{ \D\rbra[\Big]{ \rho_0\Big\| \frac{\rho_0+\rho_1}{2} } + \D\rbra[\Big]{ \rho_1\Big\| \frac{\rho_0+\rho_1}{2} } }\enspace.\]
	Here, the von Neumann entropy $\S(\rho)$ and the quantum relative entropy $\D(\rho_0 \| \rho_1)$ are defined by
    \[ \S(\rho) \coloneqq -{\rm Tr}(\rho \ln \rho) \quad\text{and}\quad  \D(\rho_0 \| \rho_1) \coloneqq \Tr\rbra*{ \rho_0 \rbra*{\ln \rho_0 - \ln \rho_1} }\enspace.\]    
\end{definition}

For convenience, we use $\QJS_2(\rho_0,\rho_1) \coloneqq \QJS(\rho_0,\rho_1)/\ln{2}$ to denote the quantum Jensen--Shannon divergence defined using base-$2$ logarithms, and likewise use $\S_2(\rho)\coloneqq\S(\rho)/\ln{2}$ for the von Neumann entropy.
Throughout this work, we use $\rho_\pm \coloneqq (\rho_0\pm\rho_1)/2$, and define
\begin{equation}
    \label{eq:state-interpolate}
    \forall z\in\binset, \quad \rho_{z,\lambda}
    \coloneqq \rho_+ + (-1)^z \lambda \ \rho_-
    = \frac{1+\lambda}{2}\rho_z+\frac{1-\lambda}{2}\rho_{1-z},
    \quad\text{where } \lambda\in\sbra*{0,1}\enspace.
\end{equation}

\begin{lemma}[Joint entropy theorem, adapted from~{\cite[Theorem~11.8]{NC10}}]
    \label{lemma:joint-entropy}
    Let $\cbra*{\rho_j}_{j\in [k]}$ be a set of quantum states where $k\in\bbN$. Let $k$-tuple $\mu\coloneqq (\mu_1,\ldots,\mu_k)$ be a probability distribution. Then, we have
    \[ \S_2\rbra[\Bigg]{\sum_{j\in[k]} \mu_j\ketbra{j}{j}\otimes\rho_j}
        =\H_2(\mu)+\sum_{j\in[k]} \mu_j\S_2(\rho_j)\enspace. \]
    In particular, for probability distributions $X_j$ on a common finite set $S$,
    \[ \S_2\rbra[\Bigg]{\sum_{j\in[k]} \sum_{x\in S} \mu_j X_j(x) \ketbra{j}{j}\otimes\ketbra{x}{x}}
        =\H_2(\mu)+\sum_{j\in[k]}\mu_j\H_2(X_j)\enspace. \]
\end{lemma}

Next, we define the \emph{quantum hockey-stick divergence} and provide the properties needed for the smoothed integral representation of $\QJS$.
For an Hermitian matrix $X$, let ${\bf1}_{X\succ 0}$ denote the projector onto the eigenspaces of $X$ corresponding to its positive eigenvalues. 
\begin{definition}[Quantum hockey-stick divergence, adapted from~\cite{SW13}]
    Let $\rho_0$ and $\rho_1$ be quantum states of the same dimension. 
    For any real $\gamma\geq 1$, the quantum hockey-stick divergence between $\rho_0$ and $\rho_1$ is defined by 
    \[ E_\gamma(\rho_0\|\rho_1) \coloneqq \Tr(\rho_0-\gamma\rho_1)_+, \quad\text{where } A_+ \coloneqq A{\bf1}_{A\succ0}
     \enspace. \]
\end{definition}

\begin{lemma}[Properties of the quantum hockey-stick divergence, adapted from~{\cite[Lemmas~2.2 and~2.5]{LiuHircheCheng25}}]
    \label{lem:hockey-derivative}
    Let $\rho_0$ and $\rho_1$ be quantum states of the same dimension. Then, we have
    \begin{enumerate}[label={\upshape(\arabic*)}]
        \item\label{thmitem:hockey-stick-convexity} \textbf{\emph{Convexity}}. The function $\gamma\mapsto E_\gamma(\rho_0\|\rho_1)$ is convex and nonincreasing on $[1,\infty)$. 
        \item\label{thmitem:hockey-stick-right-derivative} \textbf{\emph{Right derivative}}. $\frac{\dd^+}{\dd \gamma}E_\gamma(\rho_0\|\rho_1) =-\Tr \rbra*{ \rho_1\mathbf1_{\rho_0-\gamma\rho_1>0} }$, where $\frac{\dd^+}{\dd x}h(x)  \coloneqq\lim\limits_{\delta\to 0^+}\frac{h(x+\delta)-h(x)}{\delta}$.
    \end{enumerate}
\end{lemma}

\begin{lemma}[Smoothed integral representation of $\QJS$, implicit in~\cite{HircheTomamichel24}]
    \label{lem:smoothed-integral}
    Let $\rho_0$ and $\rho_1$ be quantum states of the same dimension. For any $\lambda \in [0,1/2]$, the quantum Jensen--Shannon divergence between $\rho_{0,\lambda}$ and $\rho_{1,\lambda}$ admits the following integral representation:
    \[ \QJS(\rho_{0,\lambda},\rho_{1,\lambda})
        =\lambda^2 \int_0^1
        \frac{K(t;\rho_0,\rho_1)}{1-\lambda^2t^2}\,\dd t\enspace. \]
    Here, $K(t;\rho_0,\rho_1)$ denotes a \emph{symmetrized} version of the quantum hockey-stick divergence:
    \[ \forall t\geq 0, \quad K(t;\rho_0,\rho_1) \coloneqq E_{1+t}(\rho_0\|\rho_+)+E_{1+t}(\rho_1\|\rho_+)
    =\Tr(\rho_--t\rho_+)_++\Tr(-\rho_--t\rho_+)_+\enspace. \]
\end{lemma}

\begin{proof}
    As pointed out in~\cite[Equation~(2.105)]{HircheTomamichel24}, applying~\cite[Proposition~2.9]{HircheTomamichel24} with $f(x)=x\ln{x}$ and $\lambda=\mu=1/2$ yields $\QJS(\rho_0,\rho_1)$, which can be viewed as the quantum $f$-divergence associated with $F(x)=\frac{1}{2}\rbra*{ x\ln x-(1+x)\ln\frac{1+x}{2} }$. Following~\cite[Definition~2.4]{HircheTomamichel24} and noting that $F''(\gamma)    =\gamma^{-3}F''(\gamma^{-1})=\frac{1}{2\gamma(1+\gamma)}$, we obtain the integral representation
    \begin{equation}
        \label{eq:QJS-integral-rep}
        \QJS(\rho_0,\rho_1)=\frac{1}{2}\int_1^\infty \frac{E_\gamma(\rho_0\|\rho_1)+E_\gamma(\rho_1\|\rho_0)}{\gamma(1+\gamma)}\,\dd \gamma\enspace.
    \end{equation}

    Since $-\rho_+\preceq\rho_-\preceq\rho_+$, we have $(\rho_--t\rho_+)_+=0$ and $(-\rho_--t\rho_+)_+=0$ for $t\geq 1$. For $\lambda>0$, setting $t=(\gamma-1)/(\lambda(\gamma+1))$, it follows that
    \[ \rho_{j,\lambda}-\gamma\rho_{1-j,\lambda}
        =\lambda(1+\gamma)\rbra*{ (-1)^j\rho_--t\rho_+ }, \quad\text{where } \gamma=\frac{1+\lambda t}{1-\lambda t}\enspace.\]

    Applying \Cref{eq:QJS-integral-rep} to $\rho_{0,\lambda}$ and $\rho_{1,\lambda}$, we see that the integrand vanishes for $\gamma\geq(1+\lambda)/(1-\lambda)$ under this substitution. Consequently, changing variables from $\gamma$ to $t$ yields
    \[
        \QJS(\rho_{0,\lambda},\rho_{1,\lambda})
         =\frac{1}{2}\int_0^1 \frac{\lambda K(t;\rho_0,\rho_1)}{\gamma} \frac{\dd \gamma}{\dd t}\,\dd t
        =\lambda^2\int_0^1\frac{K(t;\rho_0,\rho_1)}{1-\lambda^2t^2}\,\dd t\enspace.
    \]
    Here, the second equality follows from $\frac{\dd \gamma}{\dd t}=\frac{2\lambda}{(1-\lambda t)^2}$, while both sides vanish at $\lambda=0$. 
\end{proof}

\subsection{Quantum algorithmic toolkit}

\paragraph{Quantum singular value transformation.}
\begin{definition}[Block-encoding, adapted from~\cite{GSLW19}]
\label{def:block-encoding}
Let $A$ be a linear operator on a quantum register $\sfR$. A unitary $U$ acting on the registers $(\sfR,\sfA)$, where $\sfA$ is initialized to $\ket{0}^{\otimes a}$, is an \emph{$(\alpha,a,\epsilon)$-block-encoding} of $A$ if
\[ \norm*{ A - \alpha\rbra*{\bra{0}^{\otimes a} \otimes I_{\sfR}} U \rbra*{\ket{0}^{\otimes a} \otimes I_{\sfR}} } \leq \epsilon\enspace. \]
When $\alpha=1$ and $\epsilon=0$, we call $U$ an \emph{exact block-encoding} of $A$ \enspace.
\end{definition}

We use the following version of quantum singular value transformation (QSVT):

\begin{theorem}[QSVT for Hermitian operators, adapted from~{\cite[Theorem~31]{GSLW19}}]
\label{thm:qsvt-hermitian}
Let $U_A$ be a unitary operator that is an $(\alpha,a,\epsilon)$-block-encoding of a Hermitian operator $A$, where $\alpha>0$, $\epsilon\geq0$, and $\norm{A}\leq\alpha$. Let $P\in\bbR[x]$ be a degree-$d$ polynomial, where $d\geq1$, and suppose that $\abs{P(x)}\leq1/2$ for all $x\in[-1,1]$.
Then, for every $\delta\in(0,1)$, there exists an explicit quantum circuit $U$ that is a $\rbra[\big]{1,a+2,4d\sqrt{\epsilon/\alpha}+\delta}$-block-encoding of $P(A/\alpha)$.
This circuit can be implemented using $O(d)$ applications of $U_A$ and $U_A^\dagger$, one application of controlled-$U_A$, and $O((a+1)d)$ additional one- and two-qubit quantum gates. Moreover, given the coefficients of $P$, a classical description of $U$ can be computed in deterministic time $\poly(d,\log(1/\delta))$.
\end{theorem}

The following corollary follows from \Cref{thm:qsvt-hermitian} by a slight rescaling and a linear combination of a QSVT circuit and its adjoint:\footnote{
Applying \Cref{thm:qsvt-hermitian} to polynomial $Q \coloneq (1-\delta/2)P/2$ with error
$\delta/4$ gives a unitary $V$ whose zero block $N$ satisfies $\norm{N-Q(A/\alpha)}\leq\delta/4$ and hence $\norm{N}\leq(1-\delta/2)/2+\delta/4=1/2$.
Thus, $A_P\coloneqq N+N^\dagger$ is Hermitian, $\norm{A_P}\leq1$, and $\norm{A_P-P(A/\alpha)}\leq \delta$. Taking an equal-weight linear combination of $V$ and $V^\dagger$ gives a $(2,a+4,0)$-block-encoding of $A_P$:
in addition to the $a+2$ ancilla qubits of $V$, we use one qubit to select between $V$ and $V^\dagger$ and one qubit to implement the resulting doubly controlled oracle calls using singly controlled calls.}

\begin{corollary}[QSVT with an exactly encoded Hermitian contraction]
\label{corr:qsvt-hermitian-contraction}
Let $U_A$ be a unitary operator that is an $(\alpha,a,0)$-block-encoding of a Hermitian operator $A$, where $\alpha>0$ and $\norm{A}\leq\alpha$. Let $P\in\bbR[x]$ be a degree-$d$ polynomial, where $d\geq1$, and suppose that $\abs{P(x)}\leq1$ for all $x\in[-1,1]$.
Then, for every $\delta\in(0,1)$, there exists an explicit quantum circuit $U$ that is a $(2,a+4,0)$-block-encoding of a Hermitian operator $A_P$ satisfying
\[
    \norm{A_P}\leq1
    \quad\text{and}\quad
    \norm*{A_P-P(A/\alpha)}\leq\delta\enspace.
\]
This circuit can be implemented using $O(d)$ applications of $U_A$, $U_A^\dagger$, and their controlled versions, and $O((a+1)d)$ additional one- and two-qubit quantum gates. Moreover, given the coefficients of $P$, a classical description of $U$ can be computed in deterministic time $\poly(d,\log(1/\delta))$.
\end{corollary}

Using~\cite[Lemma 3.1]{LW25} with the parameters $r=1$ and $\alpha=0$, we obtain:\footnote{Let $P_{\eta_0}$ be the polynomial given by~\cite[Lemma~3.1]{LW25} with $r=1$, $\alpha=0$, and $\epsilon=\eta_0$, for a fixed sufficiently small $\eta_0\in(0,1/2)$. We then enlarge $\eta_0$ to $\eta\in(\eta_0,1/2)$ by setting $P_\eta=P_{\eta_0}$. We may thus enlarge $\beta$ to an integer greater than $\max\cbra*{1,\deg(P_{\eta_0})/2}$, yielding the desired degree bound  $\deg(P_\eta)=\deg(P_{\eta_0})<2\beta<\beta/\eta$.} 
\begin{lemma}[Efficient uniform polynomial approximation of the absolute value function, adapted from~{\cite[Lemma 3.1]{LW25}}]
    \label{lemma:abs-polynomial-approx}
    For every $\eta \in (0,1/2)$, there is an efficiently computable even polynomial $P_{\eta}\in\bbR[x]$ of degree $d \leq \ceil*{\beta/\eta}$, where $\beta>1$ is a constant, such that 
    \[ \max_{x\in[-1,1]} \abs*{ P_{\eta}(x)-\frac{|x|}{2} } \leq \eta \quad\text{and}\quad \max_{x\in[-1,1]} \abs*{ P_{\eta}(x) } \leq 1\enspace. \]
\end{lemma}

Rescaling the polynomial approximation in~\cite[Theorem~4.1]{SV14} gives \Cref{lemma:exp-polynomial-approx} with the additional guarantee $\max_{x\in[-1,1]}\abs{P_{\eta,\kappa}(x)}\leq1$:\footnote{For $\kappa=0$, take $P_{\eta,0}=1$. For $\kappa>0$, apply~\cite[Theorem~4.1]{SV14} on $[0,2\kappa]$ with error $\eta/2$ and substitute $y=\kappa(1-x)$ to obtain a polynomial $Q$ satisfying $\max_{x\in[-1,1]}\abs{Q(x)-e^{\kappa(x-1)}}\leq\eta/2$.
Set $P_{\eta,\kappa}=Q/(1+\eta/2)$. Since $0<e^{\kappa(x-1)}\leq1$ on $[-1,1]$, we have $\max_{x\in[-1,1]}\abs{P_{\eta,\kappa}(x)}\leq1$ and $\max_{x\in[-1,1]}\abs{P_{\eta,\kappa}(x)-e^{\kappa(x-1)}} \leq\eta/(1+\eta/2)\leq\eta$. The rescaling preserves efficient computability, and replacing $\eta$ by $\eta/2$ preserves the degree bound.}
\begin{lemma}[Efficient uniform polynomial approximation of the exponential function, adapted from~{\cite[Corollary~64 of the full version]{GSLW19}}]
    \label{lemma:exp-polynomial-approx}
    For every $\kappa\geq0$ and $\eta\in(0,1/2)$, there is an efficiently computable polynomial $P_{\eta, \kappa}\in\bbR[x]$ of degree $d=O\rbra*{\sqrt{\kappa\ln(1/\eta)}+\ln(1/\eta)}$ such that
    \[
        \max_{x\in[-1,1]} \abs*{P_{\eta, \kappa}(x)-e^{\kappa(x-1)}}\leq\eta
        \quad\text{and}\quad
        \max_{x\in[-1,1]}\abs*{P_{\eta, \kappa}(x)}\leq1\enspace.
    \]
\end{lemma}

\begin{lemma}[Efficient polynomial approximation of the sign function, adapted from~{\cite[Corollary~6]{LC17} and~\cite[Lemma 14]{GSLW19}}]
    \label{lemma:sign-polynomial-approx}
    For every $\delta>0$ and $\eta \in (0,1/2)$, there exists an efficiently computable odd polynomial $P_{\eta, \delta}\in\bbR[x]$ of degree $d=O\rbra*{\ln(1/\eta)/\delta}$ such that
    \[ \max_{x\in[-2,2]\setminus(-\delta,\delta)} \abs{P_{\eta, \delta}(x)-\sign(x)}\leq \eta \quad\text{and}\quad \max_{x\in[-2,2]} \abs{P_{\eta, \delta}(x)}\leq 1  \enspace. \]
\end{lemma}

The linear combination of unitaries (LCU) method was developed for Hamiltonian simulation in~\cite{CW12,BCCKS15}. We use the following formulation in terms of block-encodings: 
\begin{lemma}[Linear combination of block-encoded matrices,
    adapted from~{\cite[Lemma~29]{GSLW19}}]
    \label{lemma:lcu}
    Let $A_0,\ldots,A_{m-1}$ be operators on the same $s$-qubit register, and let $U_j$ be an $(\alpha,a,\epsilon)$-block-encoding of $A_j$ for every $0\leq j<m$, where $\alpha>0$ and $\epsilon\geq0$. Let $w_0,\ldots,w_{m-1}\in\bbR$. Let $W$ be a $b$-qubit unitary such that $W\ket{0}^{\otimes b} =\sum_{j=0}^{m-1}\sqrt{\frac{\abs{w_j}}{w}}\ket{j}$, where $w=\sum_{j=0}^{m-1}\abs{w_j}>0$ and $b=\ceil*{\log_2 m}$. 
    Set $w_j=0$ and $U_j=I$ for $m\leq j<2^b$. Then the unitary
    \[
        U=(W^\dagger\otimes I) U_{\rm select} (W\otimes I),
        \quad\text{where } U_{\rm select}\coloneqq \sum_{j=0}^{2^b-1} (-1)^{{\bf1}_{\cbra{w_j<0}}}\ketbra{j}{j}\otimes U_j\enspace,
    \]
    is an $(\alpha',a+b,\epsilon')$-block-encoding of $A=\sum_{j=0}^{m-1}w_jA_j$ with $\alpha'=\alpha w$ and $\epsilon'=\epsilon w$. Here, the indicator ${\bf1}_{\cbra{w_j<0}}$ equals $1$ if $w_j<0$ and $0$ otherwise.
Suppose that $W$ and $U_j$ can be implemented using
$S_W$ and $S_j$ one- and two-qubit gates, respectively.
Then $U$ can be implemented using
    $O\rbra[\big]{S_W+(b+1)^2\sum_{j=0}^{m-1}(S_j+1)}$
one- and two-qubit gates, without additional ancillary
qubits. When the weights $w_0,\ldots,w_{m-1}$ are explicitly given,
$W$ can be chosen to use $O(m(b+1)^2)$ gates.

\noindent In particular, if $m=\poly(n)$ and every $U_j$ can be implemented using $\poly(n)$ one- and two-qubit gates, then $U$ can also be implemented using $\poly(n)$ such gates.
\end{lemma}

\paragraph{Normalized trace estimation.} To estimate the normalized trace of an exact block-encoding, we need the next two subroutines. The first is quantum amplitude estimation: 
\begin{lemma} [Quantum amplitude estimation, {\cite[Theorem 12]{BHMT02}}] \label{lemma:amplitude-estimation}
    Let $U$ be a unitary operator such that $U \ket{0} \ket{0} = \sqrt{p} \ket{0} \ket{\phi_0} + \sqrt{1-p} \ket{1} \ket{\phi_1}$, where $\ket{\phi_0}$ and $\ket{\phi_1}$ are normalized pure quantum states and $p \in \sbra{0, 1}$. 
    Then, there is a quantum query algorithm using $O\rbra{M}$ queries to $U$ that outputs $\tilde p$ such that 
    \[
    \Pr\sbra*{ \abs*{\tilde p - p} \leq \frac{2\pi\sqrt{p\rbra{1-p}}}{M} + \frac{\pi^2}{M^2} } \geq \frac{8}{\pi^2}\enspace.
    \]
    Furthermore, if $U$ acts on $n$ qubits, the quantum query algorithm can be implemented by using $O\rbra{Mn}$ one- and two-qubit quantum gates.
\end{lemma}

The second subroutine is a tailored version of one-bit precision phase estimation~\cite{Kitaev95}, often referred to as the Hadamard test~\cite{AJL09}, as stated in~\cite{GP22}:

\begin{lemma} [Hadamard test for block-encodings, adapted from {\cite[Lemma 9]{GP22}}] \label{lemma:hadamard-test}
For a unitary $U$ over $m + a$ qubits, one can implement a quantum circuit that, given an input $m$-qubit quantum state $\rho$, outputs $0$ with probability $\frac{1}{2}+\frac{1}{2}\Real\sbra{\Tr\rbra{\bra{0}^{\otimes a}U\ket{0}^{\otimes a}\rho}}$, using $1$ query to controlled-$U$ and $O\rbra{1}$ one- and two-qubit quantum gates.

In particular, for every $\delta \geq 0$, if $U$ is a $\rbra{\alpha, a, \delta}$-block-encoding of an $m$-qubit linear operator $A$, then the above quantum circuit outputs $0$ with probability $\frac{1}{2}+\frac{1}{2\alpha}p$ where $\abs{\Real\sbra{\Tr\rbra{A\rho}} - p} \leq \delta$. 
\end{lemma}

Let $\rho_\Phi$ be the $m$-qubit maximally mixed state, with purification $\ket{\Phi_{2^m}}$, which can be prepared by a $2m$-qubit unitary operator $U_{\Phi,2^m}$. One can estimate the normalized trace $\Tr(A)/2^m$ by combining \Cref{lemma:hadamard-test} with \Cref{lemma:amplitude-estimation}:
\begin{corollary}[Normalized trace estimation for a normalized block-encoding]
    \label{corr:normalized-trace-estimation}    
    Let $A$ be a linear operator on an $m$-qubit register, and let $U$ be a $(\alpha,a,\delta)$-block-encoding of $A$ where $\delta \geq 0$. Then
    \[\abs*{ \alpha\bra{\Phi_{2^m}} \rbra[\big]{ \bra{0}^{\otimes a} U \ket{0}^{\otimes a} \otimes I } \ket{\Phi_{2^m}} - \frac{\Real \sbra*{\Tr(A)}}{2^m}} \leq \delta\enspace. \]
    Hence, by applying the Hadamard test (\Cref{lemma:hadamard-test}) to $\rho_\Phi$ and $U$ and combining it with amplitude estimation (\Cref{lemma:amplitude-estimation}), one can estimate $\Real \sbra*{\Tr(A)}/2^m$ to within additive error $\epsilon + \delta$ with success probability at least $8/\pi^2$, using $O\rbra{\alpha/\epsilon}$ queries to each of $U$ and $U_{\Phi,2^m}$ and $O\rbra{\alpha\rbra{m+a}/\epsilon}$ additional one- and two-qubit quantum gates. 

    In particular, if $A$ is Hermitian, then the algorithm provides an estimate of $\Tr(A)/2^m$.
\end{corollary}

It is noteworthy that this Hadamard-test-based observation appears in~\cite[Lemma 1]{Shepherd06} in the context of \textsf{DQC1} and also in~\cite{SJ08}, which gives a self-contained proof. 

\section{\QIPtwo{} with a laconic prover: the landscape ``just above'' \QSZK{}}

We begin by defining the promise problem $\MultiQSD{}$ in \Cref{def:MultiQSD}, which formalizes \emph{finite-ensemble quantum state discrimination}, as considered in, e.g.,~\cite[Scenario~3.8]{Watrous18}:  
Given a state $\rho_s$ from an ensemble $\cbra{\rho_s}_{s\in\binset^\ell}$, where the label $s$ is chosen uniformly at random, the receiver aims to identify the hidden label $s$. This task admits an SDP formulation~\cite[Section~II]{EMV03} (see also~\cite[Section~3.1.2]{Watrous18} and~\cite[Section~4.2]{vAG19}), and the resulting promise problem naturally generalizes \QSD{} introduced in~\cite{Watrous02} (see \Cref{def:QSD}). 

\begin{definition}[\textsc{Multi-State Distinguishability}, \MultiQSD{}]
    \label{def:MultiQSD}
    Let $\ell$ be an efficiently computable function such that $\ell = O(\log n)$. Let $Q$ be a polynomial-size quantum circuit with description length $n$ such that $Q$ acts on $\ell + m$ qubits with $k$ specified output qubits, where $m(n)$ and $k(n)$ are polynomial in $n$. Let $\rho_s$ be the output state obtained by running $Q$ on $\ket{s} \ket{0}^{\otimes m}$ for an $\ell$-bit string $s \in \binset^{\ell}$. Define \[\MSucc{Q} \coloneq \max_{\{\Pi_s\}_{s \in \binset^\ell}}\frac{1}{2^\ell} \sum_{s \in \binset^\ell}\Tr(\Pi_s \rho_s)\enspace,\]
    where the maximum is over all POVMs $\{\Pi_s\}_{s \in \binset^\ell}$.
    
    \noindent Let $a(n)$ and $b(n)$ be efficiently computable functions such that $0 \leq b(n) < a(n) \leq 1$. The promise problem $\MultiQSD\sbra*{\ell(n), a(n),b(n)}$ asks to distinguish between the following two cases:
    \begin{itemize}
        \item \emph{Yes:} The quantum circuit $Q$ satisfies $\MSucc{Q} \geq a(n)$;
        \item \emph{No:} The quantum circuit $Q$ satisfies $\MSucc{Q} \leq b(n)$. 
    \end{itemize}

\noindent Furthermore, we use the notation $\MultiQSD_{2^\ell}$ when emphasizing that the number of states to be discriminated is $2^\ell$.
\end{definition}

Our main result in this section is a natural complete characterization of two-message quantum interactive proof systems with a laconic prover, who sends only $\ell(n)$ bits, where $\ell(n) = O(\log n)$: 

\begin{theorem}[\MultiQSD{} is $\QIP_{\log\text{-}\bit}$-complete]
\label{thm:MultiQSD-is-QIP-ell-bit-Complete}
For every positive integer-valued function $\ell(n) = O(\log n)$,  $\MultiQSD_{2^\ell}$ is $\QIP_{\ell\text{-}\bit}$-complete. 
In particular, let $a(n)$ and $b(n)$ be efficiently computable functions such that $0 \leq b(n) < a(n) \leq 1$, the following statements hold:
\begin{enumerate}[label={\upshape(\arabic*)}]
    \item\label{thmitem:MultiQSD-containment} \textbf{\emph{Containment}.} $\displaystyle \MultiQSD{\sbra*{\ell, a, b}} \in \QIP_{\ell\text{-}\bit}{\sbra*{2, a, b}}\enspace.$
    \item\label{thmitem:MultiQSD-hardness} \textbf{\emph{Hardness}.} $\displaystyle \MultiQSD{\sbra*{\ell, \frac{1}{2^{\ell}} + a \cdot \frac{2^{\ell} - 1}{2^{2\ell}}, \frac{1}{2^{\ell}} + b \cdot \frac{2^{\ell} - 1}{2^{2\ell}}}}$ is $\displaystyle \QIP_{\ell\text{-}\bit}{[2, a, b]}$-hard. 
\end{enumerate}
\end{theorem}

\begin{remark}[Another $\QIP_{\log\text{-}\bit}$ complete problem via conditional min-entropy]
    \Cref{thm:MultiQSD-is-QIP-ell-bit-Complete} also gives a complete problem formulated in terms of conditional min-entropy.
    Given a circuit $Q$ as in \Cref{def:MultiQSD}, consider the classical-quantum state
    $\rho_{\sfJ\sfM}\coloneq 2^{-\ell}\sum_{s\in\binset^\ell} \ketbra{s}{s}_{\sfJ}\otimes(\rho_s)_{\sfM}$.
    The conditional min-entropy of $\sfJ$ conditioned on $\sfM$ for the quantum state $\rho$ is defined by
    \[
        \HminCond{\sfJ}{\sfM}_{\rho}
        \coloneq
        -\log\inf\set*{\Tr(Y)}
        {Y\in\Pos(\sfM),\,I_{\sfJ}\otimes Y\succeq\rho_{\sfJ\sfM}}
        \enspace.
    \]
    By the guessing-probability interpretation of conditional min-entropy~\cite[Theorem~1]{KRS09}, $2^{-\HminCond{\sfJ}{\sfM}_{\rho}}$ equals  the optimal probability of guessing the classical label $s$ from the quantum state $\rho_s$, which is exactly $\MSucc{Q}$.
    Consequently, distinguishing low from high conditional min-entropy for these classical-quantum states, with suitable promise thresholds, gives another $\QIP_{\ell\text{-}\bit}$-complete problem for $\ell(n)=O(\log n)$.
\end{remark}

\Cref{thm:MultiQSD-is-QIP-ell-bit-Complete} immediately implies that \QSD{} is \QIPbit{}-complete, which complements the observation that \QSD{} (and thus \QSZK{}) is in \QIPbit{}~\cite[Section 5]{BSW11}. Therefore, our $\QIP_{\log\text{-}\bit}$-complete problem in \Cref{thm:MultiQSD-is-QIP-ell-bit-Complete} \emph{places $\QIP_{\ell\text{-}\bit}$ for $\ell\geq 2$ in the landscape ``just above'' $\QSZK$.}

\vspace{1em}
However, \Cref{thm:MultiQSD-is-QIP-ell-bit-Complete}\ref{thmitem:MultiQSD-hardness} yields only that $\QSD{[a/2,b/2]}$ is $\QIPbit[a,b]$-hard. To further improve the $\QIPbit[a,b]$-hard regime, we need a somewhat different construction: 
\begin{theorem}[$\QSD{}$ is $\QIPbit{}$-complete]
\label{thm:QSD-is-QIPbit-Complete}
Let $a(n)$ and $b(n)$ be efficiently computable functions such that $0 \leq b(n) < a(n) \leq 1$. Then, the following statements hold: 
\begin{enumerate}[label={\upshape(\arabic*)}]
    \item\label{thmitem:QSD-containment} \textbf{\emph{Containment}.} $\QSD{[a, b]} \in \QIPbit\sbra*{\frac{a + 1}{2}, \frac{b + 1}{2}}$.
    \item\label{thmitem:QSD-hardness} \textbf{\emph{Hardness}.} $\QSD{[a, b]}$ is $\QIPbit{[a, b]}$-hard.
\end{enumerate}
\end{theorem}

\vspace{1em}
In the remainder of this section, we first prove \Cref{thm:QSD-is-QIPbit-Complete} as a warm-up in \Cref{sec:QSD-is-QIPbit-complete}, and then establish \Cref{thm:MultiQSD-is-QIP-ell-bit-Complete} in \Cref{sec:MultiQSD-is-QIPellbit-complete}. 

To this end, we characterize the maximum acceptance probability of $\QIP_{\ell\text{-}\bit}(2)$ proof systems, specializing the formulation in~\cite[Section~4]{JUW09} by restricting the prover's action to a \textit{quantum-to-classical channel} (see, e.g.,~\cite[Definition~2.35]{Watrous18}): 

\begin{lemma}[Maximum acceptance probability of $\QIP_{\ell\text{-}\bit}(2)$ proof systems]
\label{lemma:optimal-acceptance-classical-response}
Let $V(x)=(V_1(x),V_2(x))$ be the verifier's actions in a $\QIP_{\ell\text{-}\bit}(2)$ proof system $\protocol{P}{V}$.
Let $\ket{\psi_x}\coloneq V_1(x)\ket{\bar{0}}_{\sfM\sfW}$ be the state after the verifier's first action, where $\sfM$ is sent to the prover and $\sfW$ is the verifier's private register.
Let $\sfM'$ denote the register containing the prover's \emph{classical} response.
The verifier applies $V_2(x)$ to $(\sfM', \sfW)$ and accepts if measuring its output qubit $\sfO$ in the computational basis gives $1$.
For each $s\in\binset^\ell$, let the subnormalized state on $\sfM$ be
\begin{align}
\label{eqn:subnormalized-state-X_xs}
    X_{x,s}\coloneq\Tr_{\sfM'\sfW}\rbra*{\ketbra{1}{1}_{\sfO}V_2(x)_{\sfM'\sfW}\rbra*{\ketbra{s}{s}_{\sfM'}\otimes \ketbra{\psi_x}{\psi_x}_{\sfM\sfW}}V_2(x)^\dagger_{\sfM'\sfW}}\enspace.
\end{align}

Then the maximum acceptance probability over every cheating prover $P^*$ satisfies
\begin{align}
\label{eqn:winning-prob-MultiQSD-hardness-2}
    \max_{P^*}\prob{(\protocol{P^*}{V})(x)\text{ accepts}} =\max_{\textnormal{POVM }\{M_s\}_{s\in\binset^\ell}}\sum_{s\in\binset^\ell}\Tr\rbra*{M_sX_{x,s}}\enspace.
\end{align}
Here, the maximum on the right-hand side is over all POVMs on $\sfM$.
\end{lemma}

\begin{proof}
By~\cite[Theorem~2.37]{Watrous18}, the prover's effective quantum-to-classical channel corresponds to a POVM $\cbra{M_s}_{s\in\binset^\ell}$ on $\sfM$, and every such POVM defines a valid prover strategy.
Using~\cite[Equation~(2.224)]{Watrous18} and the definition of $X_{x,s}$ in \Cref{eqn:subnormalized-state-X_xs}, the acceptance probability of the corresponding strategy is given by
\[ \prob{(\protocol{P^*}{V})(x)\text{ accepts}} = \sum_{s\in\binset^\ell}\Tr(M_sX_{x,s}). \] 
Consequently, maximizing over all POVMs establishes \Cref{eqn:winning-prob-MultiQSD-hardness-2}. 
Finally, $X_{x,s}$ is the subnormalized reduced state on $\sfM$ associated with acceptance when the verifier's second action is applied to $\ket{s}_{\sfM'}\otimes\ket{\psi_x}_{\sfM\sfW}$.
Thus, $X_{x,s}\succeq0$ and $\Tr(X_{x,s})\leq 1$.
\end{proof}

\subsection{Warm-up: \QSD{} is \texorpdfstring{\QIPbit{}}{}-complete}
\label{sec:QSD-is-QIPbit-complete}

To prove \Cref{thm:QSD-is-QIPbit-Complete}, the \QIPbit{} containment in \Cref{thm:QSD-is-QIPbit-Complete}\ref{thmitem:QSD-containment} follows directly from~\cite[Section~4.2 and Figure~2]{Watrous02} (see also~\cite[Theorem~17]{LGLW26} for a detailed analysis), with the underlying proof system recovered as the special case $\ell=1$ of \Cref{protocol:MultiQSD-standard-protocol}:

\begin{lemma}[\QSD{} is in \QIPbit{}, adapted from~{\cite[Theorem 10]{Watrous02}}]
\label{lemma:QSD-is-in-QIPbit}
Let $a(n)$ and $b(n)$ be efficiently computable functions such that $0 \leq b(n) < a(n) \leq 1$ for every integer $n$. 
Then 
\[ \QSD{[a, b]} \in \QIPbit\sbra*{\frac{a + 1}{2}, \frac{b + 1}{2}}\enspace. \]
\end{lemma}

The challenging direction is the \QIPbit{} hardness result in \Cref{thm:QSD-is-QIPbit-Complete}\ref{thmitem:QSD-hardness}:\footnote{An instance $x$ of a $\QIPbit{}$ proof system has length $n=|x|$, while as shown in the proof of \Cref{thm:QSD-is-QIPbit-hard}, the reduction maps $x$ to a $\QSD{}$ instance consisting of two circuits with description length $n'=n+c$, for some constant $c$. The shifts in $a$ and $b$ account for this constant difference in input lengths between the two problems.\label{footnote:shifted-number}} 

\begin{theorem}[\QSD{} is \QIPbit{}-hard]\label{thm:QSD-is-QIPbit-hard}
Let $a, b \colon \Naturals \to [0, 1]$ be efficiently computable functions such that $0 \leq b(n) < a(n) \leq 1$ for every $n \in \Naturals$.
Then for every promise problem $\PromiseProblemI \in \QIPbit{[a, b]}$, there exists a constant $c$, such that $\PromiseProblemI$ can be polynomial-time many-one reduced to
$\QSD{[a', b']}$, where $a'$ and $b'$ are efficiently computable functions such that $a'(n) = a(n - c)$ and $b'(n) = b(n - c)$.
\end{theorem}

\begin{proof}
Consider a promise problem $\calI = (\calI_{\yes}, \calI_{\no}) \in \QIPbit{[a, b]}$. Let $V(x) = (V_1(x), V_2(x))$ be the quantum verifier for the $\QIPbit{[a, b]}$ protocol. 
By \Cref{lemma:optimal-acceptance-classical-response}, we obtain
\begin{align}\label{eqn:optimal-acceptance-prob-binary}
    \max_{P^*}\prob{(\protocol{P^*}{V})(x)\text{ accepts}} =\max_{\textnormal{POVM }\{M_r\}_{r\in\binset}}\sum_{r\in\binset}\Tr\rbra*{M_rX_{x,r}}\enspace,
\end{align}
where $X_{x,r}\coloneq\Tr_{\sfM'\sfW}\rbra*{\ketbra{1}{1}_{\sfO}V_2(x)_{\sfM'\sfW}\rbra*{\ketbra{r}{r}_{\sfM'}\otimes \ketbra{\psi_x}{\psi_x}_{\sfM\sfW}}V_2(x)^\dagger_{\sfM'\sfW}}$ is a subnormalized state on $\sfM$ as defined in \Cref{eqn:subnormalized-state-X_xs} for $r \in \binset$ and the state $\ket{\psi_x} \coloneq V_1(x)\ket{\bar{0}}_{\sfM\sfW}$ is the joint state of $V$'s private register and what $V$ sends.

We need the following closed-form expression for \Cref{eqn:optimal-acceptance-prob-binary}, which is a direct consequence of the Holevo--Helstrom formula for \emph{binary} generalized state discrimination~\cite[Section~II.C]{AM14}:\footnote{See the equality $P^*_s(A_1,A_2) = \frac{1}{2} \Tr\rbra*{A_1+A_2} + \frac{1}{2} \norm{A_1-A_2}_1$ in~\cite[Page~7]{AM14}.\label{footnote:binary-generalized-state-discrimination}}
\begin{align}\label{eqn:optimal-acceptance-prob-binary-closed-form}
    \max_{P^*}\prob{(\protocol{P^*}{V})(x)\text{ accepts}} =\frac{1}{2}\Tr\rbra*{X_{x, 0} + X_{x, 1}} + \frac{1}{2}\norm{X_{x, 0} - X_{x, 1}}_1\enspace,
\end{align}

Next we show \Cref{eqn:optimal-acceptance-prob-binary-closed-form} can also be written as the trace distance between two states that can be prepared by quantum polynomial-size circuits. Let's fix an instance $x$ of size $n$, and then define state-preparation circuits $Q_0$ and $Q_1$:
\begin{description}
    \item[Circuit $Q_0$.] The circuit $Q_0$ runs $V_1(x)$ to obtain $\ket{\psi_x}_{\sfM\sfW}$, prepares $\ket{0}$ on the register $\sfM'$, and then runs $V_2(x)$ with the prover's message in $\sfM'$. If $V_2$ accepts, the circuit $Q_0$ prepares $\ket{0}$ on a fresh register $\sfF$, and otherwise prepares $\ket{1}$ on $\sfF$. The circuit $Q_0$ outputs the state on registers $(\sfF, \sfM)$.
    \item[Circuit $Q_1$.] The circuit $Q_1$ runs $V_1(x)$ to obtain $\ket{\psi_x}_{\sfM\sfW}$, prepares $\ket{1}$ on the register $\sfM'$, and then runs $V_2(x)$ with the prover's message in $\sfM'$. If $V_2$ accepts, the circuit $Q_1$ prepares $\ket{2}$ on a fresh register $\sfF$, and otherwise prepares $\ket{1}$ on $\sfF$. The circuit outputs the state on registers $(\sfF, \sfM)$.
\end{description}

By construction, the total description length of the pair $(Q_0, Q_1)$ is $n + O(1)$ because we consider uniform verifiers of constant description length, and both $Q_0$ and $Q_1$ are polynomial-size quantum circuits because $V$ runs in polynomial time.

Notice that $X_{x, r} \preceq \Tr_{\sfM'\sfW}\rbra*{\ketbra{1}{1}_{\sfO}V_2(x)_{\sfM'\sfW}\rbra*{I_{\sfM'}\otimes \ketbra{\psi_x}{\psi_x}_{\sfM\sfW}}V_2(x)^\dagger_{\sfM'\sfW}} \preceq \Tr_{\sfW}(\ketbra{\psi_x}{\psi_x})$. Thus the following $\rho_0$ and $\rho_1$ are positive semidefinite:
\begin{align*}
    \rho_0 &\coloneq \ketbra{0}{0}_{\sfF} \otimes X_{x, 0} + \ketbra{1}{1}_{\sfF} \otimes \left(\Tr_{\sfW}(\ketbra{\psi_x}{\psi_x}) - X_{x, 0}\right)\enspace,\\
    \rho_1 &\coloneq \ketbra{2}{2}_{\sfF} \otimes X_{x, 1} + \ketbra{1}{1}_{\sfF} \otimes \left(\Tr_{\sfW}(\ketbra{\psi_x}{\psi_x}) - X_{x, 1}\right)\enspace.
\end{align*}
Moreover, by direct calculation, the circuit $Q_0$ outputs 
\begin{align*}
    \ketbra{0}{0}_{\sfF} \otimes X_{x, 0} + \ketbra{1}{1}_{\sfF} \otimes \left(\Tr_{\sfM'\sfW}(V_2(x)_{\sfM'\sfW}\ketbra{\psi_x}{\psi_x} \otimes \ketbra{0}{0}_{\sfM'}V_2(x)^\dagger_{\sfM'\sfW}) - X_{x, 0}\right) = \rho_0
\end{align*}
and similarly, the circuit $Q_1$ outputs the state $\rho_1$.

The trace distance between $\rho_0$ and $\rho_1$ exactly captures the maximal winning probability of a cheating prover:
\begin{align*}
    \TD(\rho_0, \rho_1) &= \frac{1}{2}\|\rho_0 - \rho_1\|_1\\
    &= \frac{1}{2}\left(\|X_{x, 0}\|_1 + \|X_{x, 1}\|_1 + \|X_{x, 0} - X_{x, 1}\|_1\right)\\
    &= \max_{P^*} \prob{(\protocol{P^*}{V})(x)\text{ accepts}}\enspace. \tag{By \Cref{eqn:optimal-acceptance-prob-binary-closed-form}}
\end{align*}

Note that the pair of quantum circuits that generates $\rho_0$ and $\rho_1$ has description length $\abs{x} + c$ for some constant $c$. As a result, a \emph{yes} instance $x \in \calI_{\yes}$ is mapped to a \emph{yes} instance $(Q_0, Q_1)$ for $\QSD{[a', b']}$, and a \emph{no} instance $x \in \calI_{\no}$ is mapped to a \emph{no} instance $(Q_0, Q_1)$ for $\QSD{[a', b']}$, which concludes the proof.
\end{proof}

\subsection{A complete problem for \texorpdfstring{$\QIP_{\ell\text{-}\bit}$ via distinguishing multiple states}{}}
\label{sec:MultiQSD-is-QIPellbit-complete}

To establish \Cref{thm:MultiQSD-is-QIP-ell-bit-Complete}, we start with the $\QIP_{\ell\text{-}\bit}(2)$ containment in \Cref{thm:MultiQSD-is-QIP-ell-bit-Complete}\ref{thmitem:MultiQSD-containment} for $\ell(n)=O(\log{n})$, which follows from a straightforward generalization of the proof system for Graph Non-isomorphism in~\cite{GMW91} (and thus for \SD{}~\cite{SV97} and \QSD{}~\cite{Watrous02}):
\begin{lemma}[$\MultiQSD_{2^\ell}$ is in $\QIP_{\ell\text{-}\bit}$]\label{lemma:MultiQSD-is-in-QIP-ell-bit}
Let $a(n)$ and $b(n)$ be efficiently computable functions such that $0 \leq b(n) < a(n) \leq 1$ for every integer $n$. Let $\ell(n)$ be an efficiently computable function such that $\ell(n) = O(\log n)$.
Then $\MultiQSD{\sbra*{\ell, a, b}} \in \QIP_{\ell\text{-}\bit}{\sbra*{2, a, b}}$.
\end{lemma}

\begin{proof}
We start by presenting a protocol for \MultiQSD{}, as presented in \Cref{protocol:MultiQSD-standard-protocol}.

\begingroup
\LinesNotNumbered
\begin{algorithm}[!ht]
    \SetAlgorithmName{Protocol}{protocol}{List of Protocols}
    \caption{Standard protocol for $\MultiQSD{[\ell, a, b]}$.}
	\label{protocol:MultiQSD-standard-protocol}
    \SetEndCharOfAlgoLine{.}
    \SetKwFor{While}{}{:}{}
    \SetKwFor{For}{For}{:}{}
    \SetKwIF{If}{ElseIf}{Else}{If}{:}{elif}{Else:}{}%
    \SetKwComment{Comment}{// }{}
    \SetKwInOut{Input}{Input}
    \SetKwInOut{Output}{Output}

    \Input{Quantum circuit $Q$ of description size $n$.}
    \Output{ACCEPT or REJECT.}
    \medskip

    \textbf{1.} $V$ samples $s \gets \{0,1\}^{\ell(n)}$, runs $Q$ on $\ket{s}_{\sfI}\ket{0}^{\otimes m}_{\sfA}$ to get an output state $\rho_s$, and sends $\rho_s$ to $P$, where \[\rho_s \coloneq \Tr_{\setcomplement{\sfO}} (Q \ketbra{s}{s}_{\sfI} \ketbra{\bar{0}}{\bar{0}}_{\sfA} Q^\dagger)\enspace,\]
    and $\sfO$ is the output register. 
    \medskip
    
    \textbf{2.} $V$ receives a string $s' \in \binset^{\ell(n)}$ from $P$, where $s'$ is supposed to be the outcome obtained by measuring $V$'s message using the optimal POVM $\{M_s\}_{s \in \binset^\ell}$ for distinguishing $\{\rho_s\}_{s \in \binset^\ell}$.
    \medskip
    
    \textbf{3.} $V$ accepts if and only if $s = s'$.
\end{algorithm}
\endgroup

By construction, the honest prover in the protocol of \Cref{protocol:MultiQSD-standard-protocol} sends only an $\ell(n)$-bit string, and the verifier is uniform and runs in quantum polynomial time since $\ell$ is already included in the description of $Q$. It remains to show that \Cref{protocol:MultiQSD-standard-protocol} has the desired completeness and soundness.

\paragraph{Completeness.}
For a \emph{yes} instance $Q$ and the honest prover $P$,
\begin{align*}
&\prob{(\protocol{P}{V})(Q)\text{ accepts}}
= \frac{1}{2^\ell}\sum_{s \in \binset^\ell}\Tr(M_s \rho_s)\enspace.
\end{align*}

By the optimality of the POVM $\{M_s\}_{s \in \binset^\ell}$, $\frac{1}{2^\ell}\sum_{s \in \binset^\ell}\Tr(M_s \rho_s) = \MSucc{Q}$, which is at least $a(n)$ for the \emph{yes} instance $Q$.
That is to say, $\prob{(\protocol{P}{V})(Q)\text{ accepts}} \geq a(n)$.

\paragraph{Soundness.}
Because the prover sends only one $\ell$-bit classical message, any prover strategy induces a POVM $\{M_{s}^*\}_{s \in \binset^\ell}$ on its received state.
Then for a cheating prover $P^*$ with strategy $\{M_{s}^*\}_{s \in \binset^\ell}$, $V$ is convinced with probability
\begin{align*}
&\prob{(\protocol{P^*}{V})(Q)\text{ accepts}}
= \frac{1}{2^\ell}\sum_{s \in \binset^\ell}\Tr(M^*_s \rho_s)\enspace.
\end{align*}

For a \emph{no} instance with $\MSucc{Q} \leq b(n)$, by the definition of $\MSucc{Q}$, for every POVM $\{M_{s}^*\}_{s \in \binset^\ell}$, we have that $\prob{(\protocol{P^*}{V})(Q)\text{ accepts}} \leq b(n)$.
\end{proof}

The challenging direction is also the $\QIP_{\ell\text{-}\bit}(2)$-hardness result in \Cref{thm:MultiQSD-is-QIP-ell-bit-Complete}\ref{thmitem:MultiQSD-hardness}, where $\ell(n)=O(\log{n})$, for which we reduce the promise problem $\PromiseProblemI \in \QIP_{\ell\text{-}\bit}{[2, a, b]}$ to $\MultiQSD$ with inverse-polynomial gap:\footnote{$a$, $b$, and $\ell$ are shifted for the same reason explained in \Cref{footnote:shifted-number}.}

\begin{theorem}[$\MultiQSD_{2^\ell}$ is $\QIP_{\ell\text{-}\bit}$-hard]\label{thm:MultiQSD-is-QIP-ell-bit-hard}
Let $a, b \colon \Naturals \to [0, 1]$ be efficiently computable functions such that $0 \leq b(n) < a(n) \leq 1$ for every $n \in \Naturals$.
Let $\ell(n)$ be an efficiently computable function such that $\ell(n) = O(\log n)$.
Then for every promise problem $\PromiseProblemI \in \QIP_{\ell\text{-}\bit}{[2, a, b]}$, there exists a constant $c$, such that $\PromiseProblemI$ can be polynomial-time many-one reduced to 
\[ \MultiQSD{\sbra*{\ell', \frac{1}{2^{\ell'}} + a' \cdot \frac{2^{\ell'} - 1}{2^{2\ell'}}, \frac{1}{2^{\ell'}} + b' \cdot \frac{2^{\ell'} - 1}{2^{2\ell'}}}}\enspace, \] 
where $a', b'$ and $\ell'$ are efficiently computable functions such that $a'(n) = a(n - c)$, $b'(n) = b(n - c)$, and $\ell'(n) = \ell(n - c)$.
\end{theorem}

\begin{proof}
Consider a promise problem $\calI = (\calI_{\yes}, \calI_{\no}) \in \QIP_{\ell\text{-}\bit}{[2, a, b]}$. Let $V(x) = (V_1(x), V_2(x))$ be the verifier actions in the $\QIP_{\ell\text{-}\bit}{[2, a, b]}$ protocol.

By \Cref{lemma:optimal-acceptance-classical-response}, we obtain
\begin{align}\label{eqn:optimal-acceptance-prob-multiple}
    \max_{P^*}\prob{(\protocol{P^*}{V})(x)\text{ accepts}} =\max_{\textnormal{POVM }\{M_s\}_{s\in\binset^\ell}}\sum_{s\in\binset^\ell}\Tr\rbra*{M_sX_{x,s}}\enspace,
\end{align}
where $X_{x,s}\coloneq\Tr_{\sfM'\sfW}\rbra*{\ketbra{1}{1}_{\sfO}V_2(x)_{\sfM'\sfW}\rbra*{\ketbra{s}{s}_{\sfM'}\otimes \ketbra{\psi_x}{\psi_x}_{\sfM\sfW}}V_2(x)^\dagger_{\sfM'\sfW}}$ is a subnormalized state on $\sfM$ as defined in \Cref{eqn:subnormalized-state-X_xs} for $s \in \binset^\ell$ and the state $\ket{\psi_x} \coloneq V_1(x)\ket{\bar{0}}_{\sfM\sfW}$ is the joint state of $V$'s private register and what $V$ sends. We set $p_{x, s} \coloneq \Tr(X_{x, s})$.

\paragraph{Constructing the state-preparation circuit $Q$.}
Next we show the right-hand side of \Cref{eqn:optimal-acceptance-prob-multiple} is related to $\MSucc{Q}$ for a circuit $Q$.
We first define a family of states $\{\rho_s\}_{s \in \binset^\ell}$ before showing they can be generated by $Q$.
For each $s \in \binset^\ell$, define
\begin{align*}
    \rho_s &\coloneq \frac{2^\ell - 1}{2^\ell} (X_{x, s})_{\sfA} \otimes \ketbra{0}{0}_{\sfB} \otimes \ketbra{\bar{0}}{\bar{0}}_{\sfC}\\ 
    &\qquad + \ketbra{\bar{0}}{\bar{0}}_{\sfA} \otimes \ketbra{1}{1}_{\sfB} \otimes \rbra*{\frac{1}{2^\ell} (1 - p_{x, s}) \sum_{s' \in \binset^\ell \text{ s.t. } s'\ne s}\ketbra{s'}{s'}_{\sfC} + \frac{1}{2^\ell}\ketbra{s}{s}_{\sfC}}\enspace.
\end{align*}
Recall that $p_{x, s} = \Tr(X_{x, s}) \in [0, 1]$ and $0 \preceq X_{x, s}$. We obtain that \[\Tr(\rho_s) = \frac{2^\ell - 1}{2^\ell} p_{x, s} + \frac{2^\ell - 1}{2^\ell} (1 - p_{x, s}) + \frac{1}{2^\ell} = 1\] and $\rho_s \succeq 0$ and thus $\rho_s$ is indeed a normalized quantum state.

\begingroup
\LinesNumbered
\begin{algorithm}[!ht]
    \SetAlgorithmName{Circuit}{Algorithm}{List of Algorithms}
    \caption{Circuit $Q$ for preparing $\{\rho_s\}_{s \in \binset^\ell}$}
	\label{circuit:hardness-of-MultiQSD}
    \SetEndCharOfAlgoLine{.}
    \setlength{\parskip}{5pt}
    \SetKwFor{While}{}{:}{}
    \SetKwFor{For}{For}{:}{}
    \SetKwIF{If}{ElseIf}{Else}{If}{:}{elif}{Else:}{}%
    \SetKwComment{Comment}{// }{}
    \SetKwInOut{Input}{Input}
    \SetKwInOut{Output}{Output}

    \Input{A string $s \in \binset^\ell$.}
    \Output{The quantum state $\rho_s$.}

    Prepare $\frac{1}{\sqrt{2^\ell}}\sum_{i \in \binset^\ell}\ket{i}$ on register $\sfF$.

    \medskip
    {\color{gray}\Comment{Generating the branch $\ketbra{\bar{0}}{\bar{0}}_{\sfA} \otimes \ketbra{1}{1}_{\sfB} \otimes \frac{1}{2^\ell}\ketbra{s}{s}_{\sfC}$}}

    Controlled on $\sfF = 0^\ell$, prepare $\ket{\bar{0}, 1, s}$ on the registers $(\sfA, \sfB, \sfC)$, and output the state on the registers $(\sfA, \sfB, \sfC)$.

    \medskip
    {\color{gray}\Comment{Deciding whether to generate the branch $\frac{2^\ell - 1}{2^\ell} (X_{x, s})_{\sfA} \otimes \ketbra{0}{0}_{\sfB} \otimes \ketbra{\bar{0}}{\bar{0}}_{\sfC}$ or the branch $\frac{1}{2^\ell} (1 - p_{x, s}) \ketbra{\bar{0}}{\bar{0}}_{\sfA} \otimes \ketbra{1}{1}_{\sfB} \otimes \sum_{s' \ne s}\ketbra{s'}{s'}_{\sfC}$}}

    Controlled on $\sfF \ne 0^\ell$, run $V_1(x)$ on the ancillary qubits to produce $\ket{\psi_x}$, where the verifier's message is stored in register $\sfM$, and then run $V_2(x)$ where the prover's message is set to be $s$. Let the output register of $V_2(x)$ be $\sfO$.

    \medskip
    {\color{gray}\Comment{Generating the branch $\frac{1}{2^\ell} (1 - p_{x, s}) \ketbra{\bar{0}}{\bar{0}}_{\sfA} \otimes \ketbra{1}{1}_{\sfB} \otimes \sum_{s' \ne s}\ketbra{s'}{s'}_{\sfC}$}}

    Controlled on $\sfO = 0$, prepare $\ket{\bar{0}, 1}$ on the registers $(\sfA, \sfB)$, prepare the state $\frac{1}{2^\ell - 1}\sum_{s' \in \binset^\ell \text{ s.t. } s' \ne s}\ketbra{s'}{s'}$ on the register $\sfC$, and output the state on the registers $(\sfA, \sfB, \sfC)$.\footnotemark

    \medskip
    {\color{gray}\Comment{Generating the branch $\frac{2^\ell - 1}{2^\ell} (X_{x, s})_{\sfA} \otimes \ketbra{0}{0}_{\sfB} \otimes \ketbra{\bar{0}}{\bar{0}}_{\sfC}$}}

    Controlled on $\sfO = 1$, apply a swap unitary on the registers $\sfA$ and $\sfM$, prepare $\ket{0, \bar{0}}$ on the registers $(\sfB, \sfC)$, and output the state on the registers $(\sfA, \sfB, \sfC)$.
\end{algorithm}
\footnotetext{We can prepare the state $\frac{1}{2^\ell - 1}\sum_{s' \in \binset^\ell \text{ s.t. } s' \ne s}\ketbra{s'}{s'}$ on the register $\sfC$ by bitwise adding $s$ to the register $\sfF$ and then swapping the register $\sfF$ and the register $\sfC$.}
\endgroup

By direct calculation, the circuit $Q$ defined in \circuitref{circuit:hardness-of-MultiQSD} outputs $\rho_s$ given input $s$, and it has polynomially many gates because $V$ runs in polynomial time.

\paragraph{Correctness of the construction.}
It remains to show that
\begin{align}\label{eqn:relation-between-N_s-M_s}
    \max_{\{N_s\}_{s \in \binset^\ell}} \sum_{s \in \binset^\ell}\Tr\rbra*{N_s \rho_s} = \frac{2^\ell - 1}{2^\ell}\max_{\{M_s\}_{s \in \binset^\ell}} \sum_{s \in \binset^\ell}\Tr\rbra*{M_s X_{x, s}} + 1 \enspace,
\end{align}
where the first maximum is over all POVMs $\{N_s\}_{s \in \binset^\ell}$ on registers $(\sfA, \sfB, \sfC)$.

On one hand, the left-hand side of \Cref{eqn:relation-between-N_s-M_s} is upper bounded by maximizing the distinguishing probability for two orthogonal blocks of $\rho_s$ separately, 
\begin{align*}
&\max_{\{N_s\}_{s \in \binset^\ell}} \sum_{s \in \binset^\ell}\Tr\rbra*{N_s \rho_s}\\
=&\max_{\{N_s\}_{s \in \binset^\ell}} \sum_{s \in \binset^\ell}\Tr\rbra*{N_s \rbra*{\frac{2^\ell - 1}{2^\ell}X_{x, s} \otimes \ketbra{0}{0} \otimes \ketbra{\bar{0}}{\bar{0}} + \ketbra{\bar{0}}{\bar{0}} \otimes \ketbra{1}{1} \otimes \mu_{x, s}}}\\
\leq &\max_{\{N_s\}_{s \in \binset^\ell}} \sum_{s \in \binset^\ell}\Tr\rbra*{N_s \rbra*{\frac{2^\ell - 1}{2^\ell}X_{x, s} \otimes \ketbra{0}{0} \otimes \ketbra{\bar{0}}{\bar{0}} + \frac{1}{2^\ell}\ketbra{\bar{0}}{\bar{0}} \otimes \ketbra{1}{1} \otimes I}}\\
= & \frac{2^\ell - 1}{2^\ell}\max_{\{N_s\}_{s \in \binset^\ell}} \sum_{s \in \binset^\ell}\Tr\rbra*{\bra{0, \bar{0}}_{\sfB\sfC}N_s \ket{0, \bar{0}}_{\sfB\sfC} X_{x, s}} + \frac{1}{2^\ell} \cdot 2^{\ell}\\
= & \frac{2^\ell - 1}{2^\ell}\max_{\{M_s\}_{s \in \binset^\ell}} \sum_{s \in \binset^\ell}\Tr\rbra*{M_s X_{x, s}} + 1\enspace. 
\end{align*}
Here, we define the $\ell$-qubit subnormalized state $\mu_{x, s} \coloneq \frac{1}{2^\ell} (1 - p_{x, s}) \sum_{s' \in \binset^\ell \text{ s.t. } s'\ne s}\ketbra{s'}{s'} + \frac{1}{2^\ell}\ketbra{s}{s}$ in the second line, the third line uses $0 \leq p_{x,s} \leq 1$ and thus $\mu_{x, s} \preceq I/2^\ell$, and the last line follows from the fact that $\{\bra{0, \bar{0}}_{\sfB\sfC}N_s \ket{0, \bar{0}}_{\sfB\sfC}\}_{s \in \binset^\ell}$ ranges over all POVMs on $\sfA$ when $\{N_s\}_{s \in \binset^\ell}$ ranges over all POVMs on $(\sfA,\sfB,\sfC)$. 

On the other hand, let $\{M_s^*\}_{s \in \binset^\ell}$ be the POVM over register $\sfA$ that maximizes the right-hand side of \Cref{eqn:relation-between-N_s-M_s}. We show there exists $\{N_s^*\}_{s \in \binset^\ell}$ such that 
\begin{align*}
   \sum_{s \in \binset^\ell}\Tr\rbra*{N^*_s \rho_s} = \frac{2^\ell - 1}{2^\ell}\sum_{s \in \binset^\ell}\Tr\rbra*{M^*_s X_{x, s}} + 1 \enspace.
\end{align*}

Consider the following $\{N^*_s\}_{s \in \binset^\ell}$, which forms a POVM over $(\sfA, \sfB, \sfC)$:
\begin{itemize}
    \item For $s \ne 0^\ell$, $N^*_s \coloneq M_s^* \otimes \ketbra{0}{0}_{\sfB} \otimes \ketbra{\bar{0}}{\bar{0}}_{\sfC} + I_{\sfA}\otimes\ketbra{1}{1}_{\sfB} \otimes \ketbra{s}{s}_{\sfC}\enspace$;
    \item For $s = 0^\ell$, $N^*_s \coloneq M_s^* \otimes \ketbra{0}{0}_{\sfB} \otimes \ketbra{\bar{0}}{\bar{0}}_{\sfC} + I_{\sfA}\otimes\ketbra{1}{1}_{\sfB} \otimes \ketbra{s}{s}_{\sfC} + I_{\sfA}\otimes\ketbra{0}{0}_{\sfB} \otimes (I - \ketbra{\bar{0}}{\bar{0}})_{\sfC}\enspace$.
\end{itemize}

By direct calculation, for this POVM, 
\begin{align*}
&\sum_{s \in \binset^\ell}\Tr\rbra*{N^*_s \rho_s}\\
=&\sum_{s \in \binset^\ell}\Tr\rbra*{N^*_s \rbra*{\frac{2^\ell - 1}{2^\ell}X_{x, s} \otimes \ketbra{0}{0} \otimes \ketbra{\bar{0}}{\bar{0}} + \ketbra{\bar{0}}{\bar{0}} \otimes \ketbra{1}{1} \otimes \mu_{x, s}}}\\
=&\frac{2^\ell - 1}{2^\ell} \sum_{s \in \binset^\ell}\Tr\rbra*{M^*_s X_{x, s}} + \sum_{s \in \binset^\ell}\Tr\rbra*{N^*_s \ketbra{\bar{0}}{\bar{0}} \otimes \ketbra{1}{1} \otimes \mu_{x, s}}\\
=&\frac{2^\ell - 1}{2^\ell} \sum_{s \in \binset^\ell}\Tr\rbra*{M^*_s X_{x, s}} + 1\enspace,
\end{align*}
where in the last line, we notice that the only term of $N^*_s$ that contributes in $\Tr\rbra*{N^*_s \ketbra{\bar{0}}{\bar{0}} \otimes \ketbra{1}{1} \otimes \mu_{x, s}}$ is the term $I_{\sfA}\otimes\ketbra{1}{1}_{\sfB} \otimes \ketbra{s}{s}_{\sfC}$, and we use the fact that $\Tr\rbra*{\ketbra{s}{s}\mu_{x, s}} = \frac{1}{2^\ell}$.
Therefore, this argument establishes \Cref{eqn:relation-between-N_s-M_s}. Combining \Cref{eqn:relation-between-N_s-M_s,eqn:optimal-acceptance-prob-multiple} yields
\begin{align*}
    \max_{\{N_s\}_{s \in \binset^\ell}} \sum_{s \in \binset^\ell}\Tr\rbra*{N_s \rho_s} = \frac{2^\ell - 1}{2^\ell}\max_{P^*}\prob{(\protocol{P^*}{V})(x)\text{ accepts}} + 1 \enspace,
\end{align*}
which implies
\begin{align*}
    \MSucc{Q} = \frac{2^\ell - 1}{2^{2\ell}}\max_{P^*}\prob{(\protocol{P^*}{V})(x)\text{ accepts}} + \frac{1}{2^\ell}\enspace.
\end{align*}

Notice that the description length of $Q$ is $\abs{x} + c$ for some constant $c$, because we consider uniform verifiers for the $\QIP_{\ell\text{-}\bit}$ protocol and the prover's message length $\ell$ can be easily derived from the description of the verifier and $\abs{x}$. As a result, a \emph{yes} instance $x \in \calI_{\yes}$ is mapped to a \emph{yes} instance $Q$ of size $\abs{x} + c$ for $\MultiQSD{\sbra*{\ell', \frac{1}{2^{\ell'}} + a' \cdot \frac{2^{\ell'} - 1}{2^{2\ell'}}, \frac{1}{2^{\ell'}} + b' \cdot \frac{2^{\ell'} - 1}{2^{2\ell'}}}}$, and a \emph{no} instance $x \in \calI_{\no}$ is mapped to a \emph{no} instance $Q$ of size $\abs{x} + c$ for $\MultiQSD{\sbra*{\ell', \frac{1}{2^{\ell'}} + a' \cdot \frac{2^{\ell'} - 1}{2^{2\ell'}}, \frac{1}{2^{\ell'}} + b' \cdot \frac{2^{\ell'} - 1}{2^{2\ell'}}}}$,  where $a'(n) = a(n - c)$, $b'(n) = b(n - c)$, and $\ell'(n) = \ell(n - c)$.
\end{proof}

\section{Easy regimes for \texorpdfstring{$\QIP_{\ell\text{-}\bit}(2)$}{} that collapse to \texorpdfstring{$\QSZK$}{}}

In this section, we identify two easy regimes for \texorpdfstring{$\QIP_{\ell\text{-}\bit}(2)$}{} that collapse to \texorpdfstring{$\QSZK$}{}. Our first main result follows from combining the quantum polarization in the natural regime (\Cref{thm:QSD-in-QSZK-natural-regime}) with the \QIPbit{}-hardness of \QSD{} (\Cref{thm:QSD-is-QIPbit-Complete}\ref{thmitem:QSD-hardness}):

\begin{theorem}[\QIPbit{} is in \QSZK{}]
\label{thm:QIPbit-in-QSZK}
Let $c(n)$ and $s(n)$ be efficiently computable functions such that $0 \leq s(n) < c(n) \leq 1$. For all sufficiently large $n$, the following statement holds:
\[ \text{If } c(n)-s(n) \geq 1/O(\log{n}), \quad \QIPbit[c,s] \subseteq \QSZK. \]
\end{theorem}

Consequently, since \QSD{} is \QSZK{}-hard, \QSD{} is in \QIPbit{}, and $\QSZK$ is closed under complement, all as established in~\cite{Watrous02,Wat09}, we obtain:

\begin{corollary}
    $\QIPbit = \mathsf{co}\text{-}\QIPbit$. 
\end{corollary}

Our second main result identifies an easy regime for $\QIP_{\ell\text{-}\bit}$ by compressing the prover's response from $\ell$ bits to a single bit whenever the completeness $c$ and soundness $s$ are sufficiently separated; a related answer-compression result for classical interactive proof systems with a laconic prover was proved in~\cite[Theorem 3.4]{GVW02}:

\begin{theorem}[Answer compression for $\QIP_{\ell\text{-}\bit}$]
\label{thm:QIPlbit-answer-compression}
Let $\ell=O(\log{n})$ be a positive integer-valued function. 
Let $c(n)$ and $s(n)$ be efficiently computable functions such that $0 \leq s(n) < c(n) \leq 1$. 
\[ \text{If } c>\frac{\rbra*{ 1+2^{\ell/2} }}{2}\cdot s, \quad 
\QIP_{\ell\text{-}\bit}[2,c,s] \subseteq \QIPbit\sbra[\bigg]{ \frac{1}{2}+\frac{c}{2^{\ell+1}}, \frac{1}{2}+\frac{(1+2^{\ell/2})s}{2^{\ell+2}}} \enspace. \]
\end{theorem}

In the remainder of this section, we establish the improved quantum polarization in \Cref{subsec:quantum-polarization}. Next, we provide answer compression for $\QIP_{\ell\text{-}\bit}$ via \QSD{} in \Cref{subsec:answer-compression-QIPlbit}. 

\subsection{Polarizing the trace distance in the natural regime}
\label{subsec:quantum-polarization}
We begin by strengthening the known \QSZK{} containment of \QSD{} to cover the natural regime: 

\begin{theorem}[\QSD{} is in \QSZK{} with the natural regime]
    \label{thm:QSD-in-QSZK-natural-regime}
    Let $a, b \colon \Naturals \to [0, 1]$ be efficiently computable functions such that $0 \leq b(n) < a(n) \leq 1$ for every $n\in\Naturals$. If \[a(n)-b(n) \geq 1/O(\log{n})\] for all sufficiently large $n$, then $\QSD[a,b] \in \QSZK$. 
\end{theorem}

To establish \Cref{thm:QSD-in-QSZK-natural-regime}, the main idea is to reduce $\QSD[a,b]$  to \textsc{Quantum Entropy Difference} (\QED{}) using a signed linear combination of quantum Jensen--Shannon divergences between $\rho_{0,\lambda_i}$ and $\rho_{1,\lambda_i}$ that closely approximates the trace distance between $\rho_0$ and $\rho_1$, where the coefficients $\cbra{c_i}$ are obtained from an \emph{efficiently computable} uniform polynomial approximation of the absolute value function in~\cite[Lemma 3.1]{LW25}. Our construction consists of two steps:
\begin{enumerate}[label={\upshape(\arabic*)}]
    \item \label{algoitem:td-approx}Bound the approximation error between $\sum_i c_i \QJS_2(\rho_{0,\lambda_i},\rho_{1,\lambda_i})$ and $\TD(\rho_0,\rho_1)$, while also establishing an upper bound on $\sum_i \abs{c_i}$.  
    \item \label{algoitem:td-implement}Efficiently implement a \QED{} instance encoding a suitably shifted and normalized version of $\sum_i c_i \QJS_2(\rho_{0,\lambda_i},\rho_{1,\lambda_i})$.
\end{enumerate}

\vspace{1em}
Remarkably, the proof strategy underlying \Cref{thm:QSD-in-QSZK-natural-regime} also applies to polarizing the total variation distance, thereby resolving the first open problem posed in~\cite[Section 6]{SV97}:

\begin{restatable}[\SD{} is in \SZK{} in the natural regime]{theorem}{SDinSZKnaturalRegime}
    \label{thm:SD-in-SZK-natural-regime}
    Let $a,b\colon\Naturals\to[0,1]$ be efficiently computable functions such that $0\leq b(n)<a(n)\leq1$ for every $n\in\Naturals$. If \[a(n)-b(n)\geq 1/O(\log n)\] for all sufficiently large $n$, then $\SD[a,b]\in\SZK$.
\end{restatable}

The proof of \Cref{thm:SD-in-SZK-natural-regime} is deferred to \Cref{sec:SD-in-SZK} and also uses \Cref{lemma:td-efficient-signed-approx-dyadic}. 

\subsubsection{Approximating the trace distance}
We begin by reducing the approximation error between $\sum_i c_i \QJS_2(\rho_{0,\lambda_i},\rho_{1,\lambda_i})$ and $\TD(\rho_0,\rho_1)$ to a scalar upper bound: 
\begin{lemma}[Scalar error bound for trace distance approximation]
    \label{lemma:td-scalar-error-bound}
    Let $\rho_0$ and $\rho_1$ be quantum states of the same dimension. For every positive integer $J$, let $\rho_{0,\lambda_j}$ and $\rho_{1,\lambda_j}$ be defined as in \Cref{eq:state-interpolate}, and let $c_j$ be real coefficients with $\lambda_j\in[0,1/2]$ for each $j\in [J]$. Then,
    \[ \abs*{\sum_{j\in[J]} c_j\QJS_2(\rho_{0,\lambda_j},\rho_{1,\lambda_j})-\TD(\rho_0,\rho_1)}
    \leq\sup_{|t|\leq1}\abs*{\sum_{j\in[J]} c_j\Phi(\lambda_j t)-\abs{t}} \enspace. \]
    Here, the function $\Phi(x)\coloneqq 1-\H_2\rbra*{(1+x)/2}$ for $\abs{x}<1$. 
\end{lemma}

Our proof of \Cref{lemma:td-scalar-error-bound} uses a standard result in real analysis:
\begin{lemma}[Bonnet's second mean-value theorem, adapted from~{\cite[Theorem~7.37]{Apostol74}}]
\label{lemma:bonnet-mean-value}
Let $g$ be a nonnegative nonincreasing function on $[0,1]$, and let $h$ be a continuous real-valued function on $[0,1]$. Then, 
\[
    \abs*{\int_0^1 g(t)h(t)\,\dd t}
    \leq g(0)\sup_{0\leq u\leq 1}\abs*{\int_0^u h(t)\,\dd t} \enspace.
\]
\end{lemma}

\begin{proof}[Proof of \Cref{lemma:td-scalar-error-bound}]
We begin by expressing both $\sum_{j\in[J]} c_j\QJS_2(\rho_{0,\lambda_j},\rho_{1,\lambda_j})$ and $\TD(\rho_0,\rho_1)$ as integral representations that share the same nonnegative, nonincreasing multiplicative factor in the integrand, which enable us to apply \Cref{lemma:bonnet-mean-value} and obtain the desired scalar upper bound. 
We consider the following function:
\[ \forall t\geq 0,\quad G(t;\rho_0,\rho_1)\coloneqq -\frac{\dd^+}{\dd t}K(t;\rho_0,\rho_1)\enspace. \]
Here, the symmetrized divergence $K(t;\rho_0,\rho_1)$ is defined in \Cref{lem:smoothed-integral} and satisfies 
\begin{equation}
    \label{eq:symmetrized-hockey-boundaries}
    K(0;\rho_0,\rho_1)=\TD(\rho_0,\rho_1) \quad\text{and}\quad K(1;\rho_0,\rho_1)=0 \enspace.
\end{equation}

Since $t\mapsto K(t;\rho_0,\rho_1)$ is convex and nonincreasing, as guaranteed by \Cref{lem:hockey-derivative}\ref{thmitem:hockey-stick-convexity}, it follows that the function $t\mapsto G(t;\rho_0,\rho_1)$ is nonnegative and nonincreasing. Using the right derivative formula in \Cref{lem:hockey-derivative}\ref{thmitem:hockey-stick-right-derivative}, we obtain
\begin{equation}
    \label{eq:right-derivative-bound}
    G(0;\rho_0,\rho_1)
    =-\frac{\dd^+}{\dt} \rbra*{E_{1+t}(\rho_0\|\rho_+)+E_{1+t}(\rho_1\|\rho_+)}\Big|_{t=0}
    =\Tr\rbra*{\rho_+\rbra*{{\bf1}_{\rho_-\succ 0}+{\bf1}_{-\rho_-\succ 0}}}\leq 1 \enspace.
\end{equation}
Hence, $0\leq G(t;\rho_0,\rho_1)\leq1$ for $t\in[0,1]$, and
$t\mapsto G(t;\rho_0,\rho_1)$ is Riemann integrable on $[0,1]$.
Applying the standard integral formula for a convex function~\cite[Corollary~24.2.1]{Rockafellar70} gives
\begin{equation}
    \label{eq:profile-primitive}
    \forall t\in[0,1], \quad K(t;\rho_0,\rho_1)=\int_t^1G(s;\rho_0,\rho_1)\,\dd s \enspace.
\end{equation}
Combining \Cref{eq:symmetrized-hockey-boundaries,eq:profile-primitive} yields the desired integral representation of $\TD(\rho_0,\rho_1)$. 

\vspace{1em}
We now turn to $\sum_{j\in[J]} c_j\QJS_2(\rho_{0,\lambda_j},\rho_{1,\lambda_j})$.  Let $f(t)\coloneqq\sum_{j\in[J]}c_j\Phi(\lambda_jt)$ for $t\in[-1,1]$.
Using the smoothed integral representation of \QJS{} (\Cref{lem:smoothed-integral}), we obtain
\begin{subequations}
\label{eq:QJS-dyadic-sum}
\begin{align}
    \sum_{j\in[J]}c_j\QJS_2(\rho_{0,\lambda_j},\rho_{1,\lambda_j})
    &= \sum_{j\in[J]} \frac{c_j\lambda_j^2}{\ln2}\int_0^1 \frac{K(t;\rho_0,\rho_1)}{1-\lambda_j^2t^2}\,\dd t\\
    &= \sum_{j\in[J]} c_j\lambda_j^2 \int_0^1 K(t;\rho_0,\rho_1) \Phi''(\lambda_j t) \dd t\\
    &=\int_0^1K(t;\rho_0,\rho_1) \, \dd f'(t)\\
    &=\rbra[\big]{K(t;\rho_0,\rho_1)f'(t)}\Big|_0^1
      +\int_0^1G(t;\rho_0,\rho_1)f'(t)\,\dd t\\
    &=\int_0^1G(t;\rho_0,\rho_1)f'(t)\,\dd t\enspace.
\end{align}
\end{subequations}
Here, the second line uses $\Phi''(x)=\frac1{(1-x^2)\ln2}$ for all $x\in(-1,1)$, the third line follows from $f''(t)=\frac1{\ln2}\sum_{j\in[J]} \frac{c_j\lambda_j^2}{1-\lambda_j^2t^2}$, the fourth line applies integration by parts using \Cref{eq:profile-primitive}, and the last line follows from $f'(0)=\Phi'(0)=0$ and $K(1;\rho_0,\rho_1)=0$. 

Putting \Cref{eq:symmetrized-hockey-boundaries,eq:profile-primitive,eq:QJS-dyadic-sum} together, we obtain
\begin{align*}
    \abs*{\sum_{j\in[J]}c_j\QJS_2(\rho_{0,\lambda_j},\rho_{1,\lambda_j})-\TD(\rho_0,\rho_1)}
    &=\abs*{\int_0^1G(t;\rho_0,\rho_1)(f'(t)-1)\,\dd t}\\
    &\leq G(0;\rho_0,\rho_1)\sup_{0\leq u\leq1}\abs*{\int_0^u(f'(t)-1)\,\dd t}\\
    &=G(0;\rho_0,\rho_1)\sup_{0\leq u\leq1}\abs*{f(u)-u}\\
    &\leq\sup_{\abs{t}\leq1}\abs*{\sum_{j\in[J]}c_j\Phi(\lambda_jt)-\abs{t}}\enspace.
\end{align*}
Here, the second line follows from \Cref{lemma:bonnet-mean-value}, the third line uses $f(0)=0$, and the last line follows from $G(0;\rho_0,\rho_1) \leq 1$ by \Cref{eq:right-derivative-bound} and the fact that $f$ is even. 
\end{proof}

\begin{lemma}[Efficient signed approximation of the absolute value function]
    \label{lemma:td-efficient-signed-approx-dyadic}
    For every dyadic $\varepsilon \in (0,1]$, one can compute a positive integer $J\leq\ceil*{64\beta/\varepsilon}$ and dyadic numbers $c_j$ and $\lambda_j\in[0,1/2]$ for $j\in[J]$, each represented using $O(1/\varepsilon)$ bits, such that
    \begin{equation}
        \sup_{|z|\leq1}\abs*{\sum_{j=1}^J c_j\Phi(\lambda_j z)-|z|}\leq\varepsilon
        \quad\text{and}\quad
        \sum_{j=1}^J|c_j|\leq2^{\ceil*{496\beta/\varepsilon}}+1\enspace.
    \end{equation}
    The running time is polynomial in $1/\varepsilon$ and the binary encoding length of $\varepsilon$.
\end{lemma}

\begin{proof}
Our proof strategy begins with an even polynomial approximation $P$ of the absolute value function, with a zero constant term and efficiently computable coefficients. We then match $P$ to a truncation of $\sum_j c_j\Phi(\lambda_j z)$ and bound the higher-order terms. Finally, we round the coefficients and compute the dyadic truncation errors. 

\paragraph{Constructing the signed approximation.}
We construct the signed approximation by matching the coefficients of an even polynomial approximation of $\abs{z}$ with those of a linear combination of truncated power series of $\Phi$.
Let $P_\eta$ be the polynomial approximation of $\abs{x}/2$ given by \Cref{lemma:abs-polynomial-approx} with $\eta = \beta/(2m)$, where $m=2^{\ceil*{\log(4\beta/\varepsilon)}} \geq 4$. Consequently, $4\beta/\varepsilon \leq  m < 8\beta/\varepsilon$, so $m\geq 4$ and $0 < \eta \leq \varepsilon/8<1/2$. 
Since $P_\eta$ is even and $\deg P_\eta\leq\ceil*{\beta/\eta}=2m$, the polynomial
\[ P(z) \coloneqq 2(P_\eta(z)-P_\eta(0)) = \sum_{k\in[m]} d_kz^{2k} \]
has zero constant term and efficiently computable coefficients. By the triangle inequality and the approximation guarantee in \Cref{lemma:abs-polynomial-approx},
\begin{equation}
    \label{eq:td-polynomial-error}
    \sup_{|z|\leq 1}\abs{P(z)-|z|} 
    \leq 2\sup_{|z|\leq 1} \abs{P_\eta(z)-|z|/2} + 2\abs{P_\eta(0)}
    \leq 4\eta =\frac{2\beta}{m}\enspace.
\end{equation}

We now turn to the infinite power series of $\Phi$ for $\abs{x}<1$, 
\begin{equation}
    \label{eq:td-series}
    \Phi(x)=\sum_{k=1}^{\infty} a_k x^{2k}, \quad\text{where } a_k=\frac{1}{(2k-1)(2k)\ln2}\in (0,1]\enspace.
\end{equation}

Set $J=8m$, so $J<64\beta/\varepsilon\leq\ceil*{64\beta/\varepsilon}$.
Since $J$ is a power of two, the parameters $\lambda_j=j/(16J)$ for $0\leq j\leq J$ are dyadic and lie in $[0,1/16]$.
We then aim to find real coefficients $c_0,\ldots,c_J$ such that the truncation at degree $2J$ agrees with $P$:
\begin{equation}
    \label{eq:td-truncation}
    \sum_{j=0}^J c_j\sum_{k=1}^J a_k(\lambda_jz)^{2k}
    =\sum_{k=1}^J \frac{a_k}{256^k} \sum_{j=0}^J c_j\rbra*{\frac{j}{J}}^{2k} z^{2k}
    =\sum_{k=1}^m d_kz^{2k}\enspace.
\end{equation}
Since $\lambda_0=0$ and $\Phi(0)=0$, the coefficient $c_0$ does not affect the signed approximation or its truncation. Matching coefficients in \Cref{eq:td-truncation} and choosing $c_0$ so that $\sum_{j=0}^J c_j=0$ give 
\begin{equation}
    \label{eq:moments}
    \sum_{j=0}^J c_j(j/J)^{2k}
    =\begin{cases}
        256^kd_k/a_k,&1\leq k\leq m,\\
        0,&k=0\text{ or }m<k\leq J\enspace.
    \end{cases}
\end{equation}
Setting $x_j\coloneqq(j/J)^2$, the coefficient matrix becomes the
Vandermonde matrix $V_{k,j}=x_j^k$ for $0\leq k,j\leq J$, so \Cref{eq:moments} corresponds to a linear system with variable vector $(c_0,\ldots,c_J)^T$. Since the values $\cbra{x_j}$ are distinct, $V$ is invertible and the coefficients $c_j$ are uniquely determined. 

\paragraph{Bounding the coefficients and the tail.}
We first bound the $\ell_1$-norm of the coefficients $\sum_{j=0}^J \abs{c_j}$ and then use this bound to control the higher-order terms. 
For each $j\in\cbra{0}\cup[J]$, by solving the linear system in \Cref{eq:moments} using $V^{-1}$ and the triangle inequality, we obtain
\begin{equation}
    \label{eq:td-linear-system}
    \abs{c_j}=\abs*{\sum_{k=1}^m(V^{-1})_{j,k}\frac{256^kd_k}{a_k}}
    \leq \sum_{k=1}^m 256^k \, \abs*{(V^{-1})_{j,k}\frac{d_k}{a_k}}
    \leq 2^{10m} \sum_{k=1}^m \abs*{(V^{-1})_{j,k}}\abs{d_k}\enspace.
\end{equation}
Here, the last inequality uses \Cref{eq:td-series} and the bound $4k^2 \leq 4m^2 \leq 4^m$ for any $k\in[m]$.

Let $\norm{P}_1\coloneqq\sum_{k\in[m]}\abs{d_k}$ denote the
$\ell_1$-norm of the coefficient vector of $P$,
and similarly for other polynomials.
Since $P=2(P_\eta-P_\eta(0))$ and $P_\eta$ is even, for each $k\in[m]$, we have the bound
\begin{equation}
    \label{eq:td-poly-coefficient}
    \abs{d_k} \leq \norm{P}_1
    \leq 4\sum_{\ell\in[m]}\norm{T_{2\ell} - T_{2\ell}(0)}_1
    \leq 4\sum_{\ell\in[m]}\norm{T_{2\ell}}_1
    \leq 4\sum_{\ell\in[m]}3^{2\ell}
    = \frac{9}{2}(9^m-1)
    < 2^{4m}\enspace.
\end{equation}
Here, the second inequality uses $\norm{P_\eta-P_\eta(0)}_1 \leq 2\sum_{\ell\in[m]} \norm{T_{2\ell}-T_{2\ell}(0)}_1$ since each nonconstant Chebyshev coefficient of $P_\eta$ has absolute value at most $2$, the third inequality follows because removing the constant term cannot increase the $\ell_1$ norm, the fourth inequality uses $\norm{T_r}_1\leq3^r$,\footnote{This bound follows by induction from $T_0(z)=1$, $T_1(z)=z$, and the recurrence $T_{r+1}(z)=2zT_r(z)-T_{r-1}(z)$.} and the final inequality holds for $m\geq 4$.

For every $j\in\cbra{0}\cup[J]$, using the equality obtained from~\cite[Equation~(22.1)]{Higham02} by evaluating the Lagrange basis polynomial at $x=-1$, we obtain\footnote{This equality holds because all nodes $x_h$ are nonnegative and the nonzero monomial
coefficients alternate in sign, so the absolute value at $x=-1$
equals the sum of the absolute values of the coefficients.}
\begin{equation}
    \label{eq:td-inverse-row-sum}
    \sum_{k=0}^J\abs{(V^{-1})_{j,k}}
    = \prod_{\substack{0\leq h\leq J\\h\ne j}}\frac{1+x_h}{\abs{x_j-x_h}}
    \leq 2^J \prod_{\substack{0\leq h\leq J\\h\ne j}}\frac{1}{\abs{x_j-x_h}}
    \leq 2^{J+1}\frac{J^{2J}}{(J!)^2}
    \leq 2(2e^2)^J
    \leq 2^{5J}\enspace.
\end{equation}
Here, the first inequality follows from $x_h=(h/J)^2\in[0,1]$, the second inequality uses the bound $\prod_{\substack{0\leq h\leq J\\h\ne j}}\abs{x_j-x_h} \geq \frac{(J!)^2}{2J^{2J}}$,\footnote{$\prod\limits_{0\leq h\leq J, h\ne j}\abs{x_j-x_h}=J^{-2J}
    \begin{cases}
        (J!)^2,&j=0,\\
        (J-j)!(J+j)!/2,&1\leq j\leq J,
    \end{cases}
    \geq\frac{(J!)^2}{2J^{2J}}$, where $(J-j)!(J+j)!\geq(J!)^2$.}
the third inequality uses $J!\geq(J/e)^J$, and the last inequality follows from $2e^2<16$ and $J \geq 1$. 

\vspace{1em}
By combining \Cref{eq:td-linear-system,eq:td-poly-coefficient}, a direct calculation gives
\begin{equation}
    \label{eq:coefficient-mass}
    \sum_{j=0}^J\abs{c_j}
    \leq 2^{14m}\sum_{j=0}^J\sum_{k=1}^m\abs*{\rbra*{V^{-1}}_{j,k}}
    \leq 2^{14m}\sum_{j=0}^J\sum_{k=0}^J\abs*{\rbra*{V^{-1}}_{j,k}}
    \leq (J+1)2^{14m+5J}
    \leq 2^{62m}\enspace.
\end{equation}
Here, the second inequality enlarges the inner summation range using $m\leq J$ and the nonnegativity of every summand, the third inequality applies \Cref{eq:td-inverse-row-sum} to each of the $J+1$ rows, and the last inequality uses the facts that $J+1\leq 2^J$ for $J\geq 1$ and $J=8m$.

We now bound the higher-order terms using the coefficient bound in \Cref{eq:coefficient-mass}. By \Cref{eq:td-truncation}, subtracting $P(z)$ from $\sum_{j=1}^J c_j\Phi(\lambda_jz)$ cancels all terms of degree at most $2J$, leaving only the higher-order terms:
\begin{subequations}
\label{eq:synthesis-error}
\begin{align}
    \sup_{|z|\leq 1}\abs*{\sum_{j=1}^Jc_j\Phi(\lambda_jz)-P(z)}
    =\sup_{|z|\leq 1}\abs*{\sum_{j=1}^Jc_j\sum_{k>J}a_k(\lambda_jz)^{2k}}
    &\leq\sum_{j=1}^J\abs{c_j}\sum_{k>J}256^{-k}\\
    &\leq\frac{2^{62m}}{255\cdot256^J}
    \leq\frac{1}{m}\enspace.
\end{align}
\end{subequations}
Here, the first inequality uses the triangle inequality, $0<a_k\leq1$,
and $\abs{\lambda_jz}\leq1/16$; the second inequality uses \Cref{eq:coefficient-mass} and the geometric-series identity
$\sum_{k>J}256^{-k}=1/(255\cdot 256^J)$; the last inequality follows from $J=8m$ and $m\leq2^{2m}$.

\paragraph{Making coefficients dyadic.}
For each $1\leq k\leq m$, compute a dyadic approximation to $256^kd_k/a_k$ with absolute error at most $2^{-5J}/(2J^2)$. Substitute these approximations into the right-hand side of \Cref{eq:moments}, leaving all zero entries unchanged.
Solve the resulting rational linear system exactly to obtain
rational numbers $\widetilde{c}_0,\ldots,\widetilde{c}_J$.
For each $j\in[J]$, define the truncation $\bar{c}_j \coloneqq \floor*{2J^2\widetilde c_j}/(2J^2)$, which is dyadic since $J$ is a power of two. 

By the triangle inequality, a direct calculation gives
\begin{equation}
    \label{eq:td-trunc-coefficients}
    \abs{\bar{c}_j-c_j} 
    \leq \abs{\bar{c}_j-\widetilde{c}_j}+\abs{\widetilde{c}_j-c_j}
    \leq\frac{1}{2J^2}+\frac{2^{-5J}}{2J^2} \sum_{k=0}^J\abs{(V^{-1})_{j,k}}
    \leq\frac{1}{2J^2}+\frac{2^{-5J}}{2J^2}\,2^{5J}
    =\frac1{J^2}\enspace.
\end{equation}
Here, the second inequality uses $0 \leq \widetilde{c}_j-\bar{c}_j \leq 1/(2J^2)$ for the first term and the fact that the approximation error of $256^kd_k/a_k$ is at most $2^{-5J}/(2J^2)$ for the second term; the third inequality follows from \Cref{eq:td-inverse-row-sum}. 

Applying the triangle inequality again, we obtain
\begin{align*}
    \sup_{|z|\leq 1}\abs*{\sum_{j=1}^J\bar c_j\Phi(\lambda_j z)-|z|}
    &\leq \sup_{|z|\leq 1} \abs*{\sum_{j=1}^J(\bar{c}_j-c_j)\Phi(\lambda_j z)}
      +\sup_{|z|\leq 1}\abs*{\sum_{j=1}^J c_j\Phi(\lambda_j z)-P(z)}
      +\sup_{|z|\leq 1}\abs{P(z)-|z|}\\
    &\leq \sum_{j=1}^J\abs{\bar{c}_j-c_j} + \frac{1}{m}+\frac{2\beta}{m}\\
    &\leq \frac{1}{J}+\frac{1}{m}+\frac{2\beta}{m}
      \leq\frac{4\beta}{m}
      \leq\varepsilon\enspace.
\end{align*}
Here, the second line uses $0\leq\Phi(x)\leq1 $ for $|x|\leq 1/2$ for the first term, \Cref{eq:synthesis-error} for the second term, and \Cref{eq:td-polynomial-error} for the third term; the last line uses \Cref{eq:td-trunc-coefficients}, $J=8m$ and $\beta>1$, and $m\geq4\beta/\varepsilon$ for the three inequalities, respectively. 

Following \Cref{eq:td-trunc-coefficients}, we obtain the magnitude bound
\[
    \sum_{j=1}^J\abs{\bar c_j}
    \leq\sum_{j=1}^J\abs{c_j}+\frac1J
    \leq2^{62m}+1
    \leq2^{\ceil*{496\beta/\varepsilon}}+1\enspace.
\]
Here, the second inequality uses \Cref{eq:coefficient-mass} and the last inequality uses $m<8\beta/\varepsilon$. The coefficient bound and the denominator $2J^2$ imply that each $\bar{c}_j$ has an $O(m)$-bit representation, while each $\lambda_j=j/(16J)$ has an $O(\log J)=O(\log m)$-bit representation.

Since the coefficients $d_k$ are efficiently computable, \Cref{eq:td-poly-coefficient} and the bound $256^k(2k-1)(2k)\leq2^{10m}$ in \Cref{eq:td-linear-system} show that the required dyadic approximations to $256^kd_k/a_k=256^k(2k-1)(2k)d_k\ln2$ can be computed in polynomial time by computing $d_k$ and $\ln2$ to absolute error at most $2^{-O(m)}$, requiring $O(m)$ fractional bits.
The entries of $V$ and the dyadic approximations used in \Cref{eq:moments} have polynomial bit length, so exact rational elimination and dyadic truncation take time polynomial in $m=O(1/\varepsilon)$ and the binary encoding length of $\varepsilon$. We complete the proof by relabeling $\bar{c}_j$ as $c_j$.
\end{proof}

\subsubsection{Implementing the QED instance}
Our construction underlying \Cref{thm:QSD-in-QSZK-natural-regime} is a reduction from a \QSD{} instance to an instance of \textsc{Quantum Entropy Difference} ($\QED[g]$), where $\QED[g]$ is defined similarly to \Cref{def:QSD}, with the following promise:
\begin{itemize}
    \item \emph{Yes:} The pair of quantum circuits $(Q_0,Q_1)$ satisfies $\S_2(\rho_0)-\S_2(\rho_1)\geq g(n)$;
    \item \emph{No:} The pair of quantum circuits $(Q_0,Q_1)$ satisfies $\S_2(\rho_0)-\S_2(\rho_1)\leq-g(n)$. 
\end{itemize}
Following~\cite[Section~5]{BASTS10}, $\QED[1/2]$ is in \QSZK{}, and the \QSZK{} containment extends directly to $g(n)\geq 1/\poly(n)$, as stated explicitly in~\cite[Theorem~2.20]{Liu23}:\footnote{While the notion of \emph{input length} differs between our setting and~\cite{Liu23}, we can simply pad the output with pure qubits, allowing the input description length to replace the number of output qubits used in~\cite{Liu23}.}
\begin{lemma}[\QED{} is in \QSZK{}, adapted from~{\cite[Section~5]{BASTS10}}]
    \label{lemma:QED-in-QSZK}
    For every efficiently computable function $g\colon\Naturals\to\mathbb{R}_{>0}$ such that $g(n)\geq 1/\poly(n)$, we have $\QED[g]\in\QSZK$.
\end{lemma}

\begin{proof}[Proof of \Cref{thm:QSD-in-QSZK-natural-regime}]
By the \QSZK{} containment in \Cref{lemma:QED-in-QSZK}, it suffices to establish a Karp reduction from $\QSD[a,b]$ with $\Delta(n) \coloneqq a(n)-b(n) \geq 1/O(\log{n})$ to $\QED[g]$ with $g(n) \geq 1/q(n)$ for some polynomial $q$. 
Choose dyadic numbers $\varepsilon=2^{-t}$ and $\tau\in 2^{-t}\mathbb{Z} \cap [0,1]$ satisfying
\[
    6\varepsilon < \Delta \leq 16\varepsilon
    \quad\text{and}\quad
    \abs*{\tau(n)-\frac{a(n)+b(n)}2}\leq\varepsilon\enspace.
\]
These parameters can be computed using $O(1+\log(1/\Delta))$ bits of precision.

\paragraph{Constructing the $\QED$ instance.}
Apply \Cref{lemma:td-efficient-signed-approx-dyadic} with accuracy $\varepsilon$ to obtain $J$ and dyadic numbers $c_j$ and $\lambda_j$ for $j\in[J]$.
Using the fixed constant $\beta>1$ in \Cref{lemma:abs-polynomial-approx}, the coefficient bound in \Cref{lemma:td-efficient-signed-approx-dyadic} gives $\sum_{j=1}^J \abs{c_j}+\tau \leq 2^{496\beta/\varepsilon}+2 \leq 2^{1+\ceil*{496\beta/\varepsilon}} \eqqcolon \Upsilon$. 

Let the quantum circuit pair $(Q_0,Q_1)$ be an instance of \QSD{}, where each $Q_j$ prepares a purification of the state $\rho_j$ for $j\in\binset$. Recall that $\rho_{j,\lambda}= \frac{1+\lambda}{2}\rho_j + \frac{1-\lambda}{2}\rho_{1-j}$ for each $j\in\binset$, as defined in \Cref{eq:state-interpolate}. 
Combining \Cref{lemma:td-scalar-error-bound,lemma:td-efficient-signed-approx-dyadic} gives the error bound
\begin{equation}
    \label{eq:shared-approximation}
    \abs*{\sum_{j=1}^J c_j\QJS_2(\rho_{0,\lambda_j},\rho_{1,\lambda_j}) - \TD(\rho_0,\rho_1)}\leq\varepsilon\enspace.
\end{equation}

We obtain quantum circuits $Q'_0$ and $Q'_1$, preparing $\rho'_0$ and $\rho'_1$, respectively, by applying the dyadic convex combination construction (\Cref{lemma:dyadic-linear-combi-states}) to the circuit-coefficient pairs in \Cref{table:polarization-newQ0,table:polarization-newQ1}, with $J+2$ terms and $L=\log(2\Upsilon J^2)$. Here, \Cref{table:polarization-newQ0,table:polarization-newQ1} list the ingredients, including the circuit-coefficient pairs, used to construct $Q'_0$ and $Q'_1$, respectively.

\begin{table}[ht]
    \centering
    \begin{tabular}{c c c c}
        \toprule
        Term & Circuit & State & Coefficient \\
        \midrule
        $j\in[J]$ with $c_j\geq0$
        & $Q_{\geq0}$
        & $\varrho\coloneqq I_2/2 \otimes\rho_+$
        & $\pi_j \coloneqq c_j/\Upsilon$ \\
        \midrule
        $j\in[J]$ with $c_j<0$
        & $Q_{<0,j}$
        & $\varsigma_j \coloneqq \frac{1}{2}\sum_{z\in\binset} \ketbra{z}{z}\otimes\rho_{z,\lambda_j}$
        & $\pi_j \coloneqq -c_j/\Upsilon$ \\
        \midrule
        Threshold offset
        & $Q_{\ketbra{\bar{0}}{\bar{0}}}$
        & $\ketbra{0}{0}^{\otimes(r+1)}$
        & $\pi_{J+1} \coloneqq \tau/\Upsilon$ \\
        \midrule
        Remainder
        & $Q_{\ketbra{\bar{0}}{\bar{0}}}$
        & $\ketbra{0}{0}^{\otimes(r+1)}$
        & $\pi_{J+2} \coloneqq 1-\rbra[\big]{ \sum_{j=1}^J\abs{c_j}+\tau }/\Upsilon$ \\
        \bottomrule
    \end{tabular}
    \caption{Ingredients for constructing $Q'_0$.}
    \label{table:polarization-newQ0}
\end{table}

\begin{table}[ht]
    \centering
    \begin{tabular}{c c c c}
        \toprule
        Term & Circuit & State & Coefficient \\
        \midrule
        $j\in[J]$ with $c_j\geq0$
        & $Q_{<0,j}$
        & $\varsigma_j = \frac{1}{2}\sum_{z\in\binset}
            \ketbra{z}{z}\otimes\rho_{z,\lambda_j}$
        & $\pi_j \coloneqq c_j/\Upsilon$ \\
        \midrule
        $j\in[J]$ with $c_j<0$
        & $Q_{\geq0}$
        & $\varrho = I_2/2 \otimes\rho_+$
        & $\pi_j \coloneqq -c_j/\Upsilon$ \\
        \midrule
        Threshold offset
        & $Q_{I_2/2}$
        & $I_2/2 \otimes\ketbra{0}{0}^{\otimes r}$
        & $\pi_{J+1} \coloneqq \tau/\Upsilon$ \\
        \midrule
        Remainder
        & $Q_{\ketbra{\bar{0}}{\bar{0}}}$
        & $\ketbra{0}{0}^{\otimes(r+1)}$
        & $\pi_{J+2} \coloneqq 1-\rbra[\big]{\sum_{j=1}^J\abs{c_j}+\tau}/\Upsilon$ \\
        \bottomrule
    \end{tabular}
    \caption{Ingredients for constructing $Q'_1$.}
    \label{table:polarization-newQ1}
\end{table}

The resulting states are given by
\begin{subequations}
\label{eq:implementation}
\begin{align}
    \rho'_0 &=
    \sum_{\substack{j\in[J]\colon c_j\geq0}} \pi_j\ketbra{j-1}{j-1}\otimes\varrho
    +\sum_{\substack{j\in[J]\colon c_j<0}} \pi_j\ketbra{j-1}{j-1}\otimes\varsigma_j \\
    &\qquad+\pi_{J+1}\ketbra{J}{J}\otimes\ketbra{0}{0}^{\otimes(r+1)}
    +\pi_{J+2}\ketbra{J+1}{J+1}\otimes\ketbra{0}{0}^{\otimes(r+1)}\enspace,\\
    \rho'_1 &=
    \sum_{\substack{j\in[J]\colon c_j\geq0}} \pi_j\ketbra{j-1}{j-1}\otimes\varsigma_j
    +\sum_{\substack{j\in[J]\colon c_j<0}} \pi_j\ketbra{j-1}{j-1}\otimes\varrho \\
    &\qquad+\pi_{J+1}\ketbra{J}{J} \otimes\frac{I_2}{2}\otimes\ketbra{0}{0}^{\otimes r}
    +\pi_{J+2}\ketbra{J+1}{J+1} \otimes\ketbra{0}{0}^{\otimes(r+1)}\enspace.
\end{align}
\end{subequations}
Next, we construct the circuit ingredients in \Cref{table:polarization-newQ0,table:polarization-newQ1}:\footnote{$Q_{I_2/2}$ is obtained by first preparing the EPR state $(\ket{00}+\ket{11})/\sqrt{2}$ and then outputting the first qubit together with $r$ all-zero qubits and tracing out the second qubit.}
\begin{itemize}
    \item $Q_{\geq0}$: Applying \Cref{lemma:dyadic-linear-combi-states} to the circuit-coefficient pairs $(I_2\otimes Q_0,1/4)$, $(I_2\otimes Q_1,1/4)$, $(\PauliX\otimes Q_0,1/4)$, and $(\PauliX\otimes Q_1,1/4)$, with four terms and $L=2$ gives a circuit preparing $\varrho$. Here, for each $z\in\binset$, the circuits $I_2\otimes Q_z$ and $\PauliX\otimes Q_z$ use an additional output qubit initialized to $\ket{0}$ and prepare $\ketbra{0}{0}\otimes\rho_z$ and $\ketbra{1}{1}\otimes\rho_z$, respectively.
    
    \item $Q_{<0,j}$: For each $j\in[J]$, applying \Cref{lemma:dyadic-linear-combi-states} to the circuit-coefficient pairs $\rbra[\big]{I_2\otimes Q_0,\frac{1+\lambda_j}{4}}$, $\rbra[\big]{I_2\otimes Q_1,\frac{1-\lambda_j}{4}}$, $\rbra[\big]{\PauliX\otimes Q_0,\frac{1-\lambda_j}{4}}$, and $\rbra[\big]{\PauliX\otimes Q_1,\frac{1+\lambda_j}{4}}$, with four terms and $L=\log(64J)$ gives a circuit preparing $\varsigma_j$. Here, the choice $L=\log(64J)$ is valid because $(1\pm\lambda_j)/4=(16J\pm j)/(64J)$.
\end{itemize}

\paragraph{Analysis.} Since the label distributions in $\rho'_0$ and $\rho'_1$ coincide and equal $(\pi_1,\ldots,\pi_{J+2})$, applying the joint entropy theorem (\Cref{lemma:joint-entropy}) to \Cref{eq:implementation} yields
\begin{subequations}
\label{eq:encoding}
\begin{align}
    \S_2(\rho'_0)-\S_2(\rho'_1)
    &=\frac{1}{\Upsilon}\sum_{j=1}^J c_j\rbra*{\S_2(\varrho)-\S_2(\varsigma_j)}-\frac{\tau}{\Upsilon}\\
    &=\frac{1}{\Upsilon}\sum_{j=1}^J c_j\rbra*{ \S_2(\rho_+)- \frac{\S_2(\rho_{0,\lambda_j}) + \S_2(\rho_{1,\lambda_j})}{2} } -\frac{\tau}{\Upsilon}\\
    &=\frac{1}{\Upsilon}\sum_{j=1}^J c_j\QJS_2(\rho_{0,\lambda_j},\rho_{1,\lambda_j})-\frac{\tau}{\Upsilon}\enspace.
\end{align}
\end{subequations}
Here, exchanging $\varrho$ and $\varsigma_j$ when $c_j<0$ gives the signed coefficient $c_j/\Upsilon$, while the threshold branches contribute $-\tau/\Upsilon$ and the remainder branches contribute zero in the first line. The second line again uses \Cref{lemma:joint-entropy} and $(\rho_{0,\lambda_j}+\rho_{1,\lambda_j})/2=\rho_+$.

Following \Cref{eq:encoding}, to analyze the entropy difference $\S_2(\rho'_0)-\S_2(\rho'_1)$, it suffices to analyze the quantity $\sum_{j=1}^J c_j\QJS_2(\rho_{0,\lambda_j},\rho_{1,\lambda_j})-\tau$. 
To this end, using the error bound in \Cref{eq:shared-approximation}, the choice of $\tau$, and $6\varepsilon<\Delta$, we obtain
\begin{itemize}
    \item For \emph{yes} instances, $\TD(\rho_0,\rho_1) \geq a(n)$ implies that 
    \[ \frac{1}{\Upsilon}\rbra[\Bigg]{ \sum_{j=1}^J c_j\QJS_2(\rho_{0,\lambda_j},\rho_{1,\lambda_j})-\tau }
    \geq\frac{\TD(\rho_0,\rho_1)-\varepsilon-\tau}{\Upsilon}
    \geq\frac{\Delta/2-2\varepsilon}{\Upsilon} > \frac{\varepsilon}{\Upsilon} \eqqcolon g\enspace.
    \]
    Here, the second inequality follows from $\TD(\rho_0,\rho_1)-\varepsilon-\tau \geq a-\varepsilon-\rbra[\big]{\frac{a+b}{2}+\varepsilon} = \frac{\Delta}{2}-2\varepsilon$. 

    \item For \emph{no} instances, $\TD(\rho_0,\rho_1) \leq b(n)$ implies that 
    \[ \frac{1}{\Upsilon}\rbra[\Bigg]{ \sum_{j=1}^J c_j\QJS_2(\rho_{0,\lambda_j},\rho_{1,\lambda_j})-\tau } 
    \leq\frac{\TD(\rho_0,\rho_1)+\varepsilon-\tau}{\Upsilon}
    \leq\frac{-\Delta/2+2\varepsilon}{\Upsilon} < -\frac{\varepsilon}{\Upsilon}\enspace.
    \]
    Here, the second inequality uses $\TD(\rho_0,\rho_1)+\varepsilon-\tau \leq b+\varepsilon-\rbra[\big]{\frac{a+b}{2}-\varepsilon} = -\frac{\Delta}{2}+2\varepsilon$. 
\end{itemize}
Consequently, we have $\log(1/g) = 1+\ceil*{496\beta/\varepsilon}+\log\frac{1}\varepsilon=O(\log n)$, where the last equality follows from $1/\varepsilon\leq16/\Delta=O(\log n)$. Therefore, we obtain the required quantum entropy gap bound $g(n)\geq n^{-O(1)}$, enabling the use of \Cref{lemma:QED-in-QSZK}.

\vspace{1em}
It remains to check the computational efficiency of the reduction. Each term circuit in \Cref{table:polarization-newQ0,table:polarization-newQ1} uses $O(1)$ copies of $Q_0,Q_1$ and $O(r+\log J)$ additional gates. The final application of \Cref{lemma:dyadic-linear-combi-states} combines $J+2=O(1/\varepsilon)$ terms with $L=\log(2\Upsilon J^2)=O(1/\varepsilon)$.
Since the resulting dyadic coefficients are computable in polynomial time in $n$ and $1/\varepsilon$ and have $O(1/\varepsilon)$-bit representations, both the time required to compute the circuit descriptions and the resulting circuit sizes are $\poly(n,1/\varepsilon)=\poly(n,1/\Delta)$, which completes the proof. 
\end{proof}


\subsection{Answer compression from \texorpdfstring{$\QIP_{\log\text{-}\bit}(2)$}{} to \texorpdfstring{\QIPbit{}}{} via \QSD{}}
\label{subsec:answer-compression-QIPlbit}

To establish \Cref{thm:QIPlbit-answer-compression}, since $\QSD$ is in $\QIPbit$ (cf.~\cite[Section 4.2]{Watrous02}, see also \Cref{thm:QSD-is-QIPbit-Complete}\ref{thmitem:QSD-containment}), it suffices to prove that $\QIP_{\log\text{-}\bit}(2)$ can be reduced to $\QSD$ when the completeness and soundness parameters are sufficiently separated:\footnote{$a$, $b$, and $\ell$ are shifted for the same reason explained in \Cref{footnote:shifted-number}.}

\begin{theorem}[Answer compression for $\QIP_{\ell\text{-}\bit}$ via \QSD{}]\label{thm:answer-compression-QSD}
Let $\ell(n)$ be an efficiently computable positive integer-valued function such that $\ell(n) = O(\log n)$. Let $a, b \colon \Naturals \to [0, 1]$ be efficiently computable functions such that $0 \leq b(n) < a(n) \leq 1$ and $a > \rbra*{1 + 2^{\ell/2}}b/2$ for every $n \in \Naturals$. Then there exists a constant $c$ such that 
\[ \QSD\sbra*{\frac{a'}{2^{\ell'}}, \frac{1 + 2^{\ell'/2}}{2} \cdot \frac{b'}{2^{\ell'}}} \text{ is } \QIP_{\ell\text{-}\bit}\sbra*{2, a, b}\text{-hard}\enspace,\] 
where $a'$, $b'$, and $\ell'$ are efficiently computable functions satisfying $a'(n) = a(n - c)$, $b'(n) = b(n - c)$, and $\ell'(n) = \ell(n - c)$ for every $n \geq c$.
\end{theorem}

\begin{proof}
Consider a promise problem $\calI = (\calI_{\yes}, \calI_{\no}) \in \QIP_{\ell\text{-}\bit}{[2, a, b]}$. Let $V(x) = (V_1(x), V_2(x))$ be the verifier actions in the $\QIP_{\ell\text{-}\bit}{[2, a, b]}$ protocol.
By \Cref{lemma:optimal-acceptance-classical-response}, we obtain
\begin{align}\label{eqn:optimal-acceptance-prob-multiple-2}
    \omega(V(x)) \coloneq \max_{P^*}\prob{(\protocol{P^*}{V})(x)\text{ accepts}} =\max_{\textnormal{POVM }\{\Pi_s\}_{s\in\binset^\ell}}\sum_{s\in\binset^\ell}\Tr\rbra*{\Pi_sX_{x,s}}\enspace,
\end{align}
where $X_{x,s}\coloneq\Tr_{\sfM'\sfW}\rbra*{\ketbra{1}{1}_{\sfO}V_2(x)_{\sfM'\sfW}\rbra*{\ketbra{s}{s}_{\sfM'}\otimes \ketbra{\psi_x}{\psi_x}_{\sfM\sfW}}V_2(x)^\dagger_{\sfM'\sfW}}$ is a subnormalized state on $\sfM$ as defined in \Cref{eqn:subnormalized-state-X_xs} for $s \in \binset^\ell$ and the state $\ket{\psi_x} \coloneq V_1(x)\ket{\bar{0}}_{\sfM\sfW}$ is the joint state of $V$'s private register and what $V$ sends. 
We set 
\[ p_{x, s} \coloneq \Tr(X_{x, s}) \quad\text{and}\quad p_x \coloneq \frac{1}{2^\ell}\sum_{s \in \binset^\ell}p_{x, s}\enspace. \]

To construct the $\QSD$ instance, we use hash functions $\{h_{u, v}\}_{u\in \binset^\ell, v \in \binset}$ where for each $x \in \binset^\ell$, $h_{u, v}(x) = \innerprodF{u}{x} \oplus v$.\footnote{This is the standard affine inner-product pairwise independent hash family; see, e.g., \cite[Claim~2.4]{LW05}.} We define two states $\sigma_0$ and $\sigma_1$ before showing that they can be efficiently generated by circuits $(Q_0, Q_1)$, which will be our \QSD{} instance: for $b \in \binset$,
\begin{align*}
\sigma_{x, b} &\coloneq \frac{1}{2^{2\ell}}\sum_{u \in \binset^\ell, v \in \binset}\ketbra{0, u, v}{0, u, v} \otimes \sum_{s \in \binset^\ell \text{ s.t. }h_{u, v}(s) = b}X_{x, s} + (1 - p_x)\ketbra{1, 0^\ell, 0, \bar{0}}{1, 0^\ell, 0, \bar{0}}\enspace,
\end{align*}
which can be rewritten as
\begin{align*}
\sigma_{x, b} &= \frac{1}{2^{2\ell}}\sum_{u \in \binset^\ell, s \in \binset^\ell}\ketbra{0, u, \innerprodF{u}{s} \oplus b}{0, u, \innerprodF{u}{s} \oplus b} \otimes X_{x, s} + (1 - p_x)\ketbra{1, 0^\ell, 0, \bar{0}}{1, 0^\ell, 0, \bar{0}}\enspace.
\end{align*}

Then as $0 \preceq X_{x, s}$ and $p_{x, s} \leq 1$, we can get that $\sigma_{x, b} \succeq 0$ and $\Tr\rbra*{\sigma_{x, b}} = 1$. Moreover, $\sigma_{x, b}$ can be generated efficiently by the following circuit $Q_{x, b}$:

\begingroup
\LinesNumbered
\begin{algorithm}[H]
    \SetAlgorithmName{Circuit}{Algorithm}{List of Algorithms}
    \caption{Quantum circuit $Q_{x, b}$ for preparing $\sigma_{x, b}$.}
	\label{circuit:easy-regime-for-QIPtwol}
    \SetEndCharOfAlgoLine{.}
    \setlength{\parskip}{5pt}
    \SetKwFor{While}{}{:}{}
    \SetKwFor{For}{For}{:}{}
    \SetKwIF{If}{ElseIf}{Else}{If}{:}{elif}{Else:}{}%
    \SetKwComment{Comment}{// }{}
    \SetKwInOut{Input}{Input}
    \SetKwInOut{Output}{Output}

    Prepare the mixed state $\frac{1}{2^{2\ell}}\sum_{u \in \binset^\ell, s \in \binset^\ell}\ketbra{0, u, s}{0, u, s}$ on registers $(\sfA, \sfB, \sfM')$.

    Compute $\innerprodF{u}{s} \oplus b$ coherently in the register $\sfC$.

    Prepare the state $\ket{\psi_x}$ by applying $V_1(x)$ on $\ket{\bar{0}}_{\sfM\sfW}$.

    Apply $V_2(x)$ on the registers $(\sfM', \sfW)$ to obtain the result on a register $\sfO$.

    \medskip
    {\color{gray}\Comment{Preparing the branch $\frac{1}{2^{2\ell}}\sum_{u \in \binset^\ell, s \in \binset^\ell}\ketbra{0, u, \innerprodF{u}{s} \oplus b}{0, u, \innerprodF{u}{s} \oplus b} \otimes X_{x, s}$}}

    Controlled on $\sfO = 1$, return the registers $(\sfA, \sfB, \sfC, \sfM)$.

    \medskip
    {\color{gray}\Comment{Preparing the branch $(1 - p_x)\ketbra{1, 0^\ell, 0, \bar{0}}{1, 0^\ell, 0, \bar{0}}$}}

    Controlled on $\sfO = 0$, prepare the state $\ket{1, 0^\ell, 0, \bar{0}}$ on the registers $(\sfA_1, \sfB_1, \sfC_1, \sfM_1)$.
    
    Swap the registers $(\sfA_1, \sfB_1, \sfC_1, \sfM_1)$ and the registers $(\sfA, \sfB, \sfC, \sfM)$.
    
    Return the registers $(\sfA, \sfB, \sfC, \sfM)$.

\end{algorithm}
\endgroup

Notice that with the description of $V$, the length of the prover message $\ell$ and the length of the verifier message are efficiently computable given the instance $x$. As we only consider a uniform verifier $V$, there exists a constant $c$ such that the pair of circuits $(Q_{x, 0}, Q_{x, 1})$ has description length $n' = n + c$ where $n \coloneq \abs{x}$ is the input length. Let $a'$, $b'$, and $\ell'$ be efficiently computable functions such that $a'(n) = a(n - c)$, $b'(n) = b(n - c)$, and $\ell'(n) = \ell(n - c)$ as described in the theorem statement. Moreover, both $Q_{x, 0}$ and $Q_{x, 1}$ are polynomial-size quantum circuits because $V$ runs in polynomial time.

It remains to show that the map from $x$ to $(Q_{x, 0}, Q_{x, 1})$ reduces the problem $\QIP_{\ell\text{-}\bit}\sbra*{2, a, b}$ to the problem $\QSD\sbra*{\frac{a'}{2^{\ell'}}, \frac{1 + 2^{\ell'/2}}{2} \cdot \frac{b'}{2^{\ell'}}}$.

\paragraph{Completeness.}
Intuitively, a \emph{yes} instance $x$ is mapped to a \emph{yes} instance $(Q_{x, 0}, Q_{x, 1})$ because $\omega(V(x)) \geq a$ implies that there is a way to distinguish $\{X_{x, s}\}_{s \in \binset^\ell}$, and thus we can guess $\innerprodF{u}{s}$ with high probability given a string $u \in \binset^\ell$, which enables us to distinguish between $\sigma_{x, 0}$ and $\sigma_{x, 1}$.

Formally, by \Cref{eqn:optimal-acceptance-prob-multiple-2}, for a \emph{yes} instance $x \in \GetYes{\PromiseProblemI}$, there exists a POVM $\{\Pi^*_s\}_{s \in \binset^\ell}$ such that 
\begin{align}\label{eqn:condition-for-yes-instance-QIP-logbit}
    \sum_{s \in \binset^\ell}\Tr(\Pi^*_s X_{x, s}) \geq a(n)\enspace.
\end{align}

We define a POVM $\{M_0, M_1\}$ where $M_0 \coloneq \sum_{u, s \in \binset^\ell}\ketbra{0, u, \innerprodF{u}{s}}{0, u, \innerprodF{u}{s}}_{\sfA\sfB\sfC} \otimes (\Pi^*_s)_{\sfM}$ and $M_1 \coloneq \sum_{u, s \in \binset^\ell}\ketbra{0, u, \innerprodF{u}{s} \oplus 1}{0, u, \innerprodF{u}{s} \oplus 1}_{\sfA\sfB\sfC} \otimes (\Pi^*_s)_{\sfM} + \ketbra{1}{1}_{\sfA} \otimes I_{\sfB\sfC\sfM}$. We show that $\{M_0, M_1\}$ can distinguish between $\sigma_{x, 0}$ and $\sigma_{x, 1}$ with high probability:
\begin{align*}
&\Tr\rbra*{\sigma_{x, 0} M_0} + \Tr\rbra*{\sigma_{x, 1} M_1}\\
=& \frac{1}{2^{2\ell}}\sum_{u, u', s, s' \in \binset^\ell}\Tr\rbra*{\rbra*{\ketbra{0, u, \innerprodF{u}{s}}{0, u, \innerprodF{u}{s}} \otimes X_{x, s}}\rbra*{ \ketbra{0, u', \innerprodF{u'}{s'}}{0, u', \innerprodF{u'}{s'}} \otimes \Pi^*_{s'}}} + 1 - p_x\\
&+ \frac{1}{2^{2\ell}}\sum_{u, u', s, s' \in \binset^\ell}\Tr\rbra*{\rbra*{\ketbra{0, u, \innerprodF{u}{s} \oplus 1}{0, u, \innerprodF{u}{s} \oplus 1}\otimes X_{x, s}} \rbra*{ \ketbra{0, u', \innerprodF{u'}{s'} \oplus 1}{0, u', \innerprodF{u'}{s'} \oplus 1} \otimes \Pi^*_{s'}}}\\
=& \frac{1}{2^{2\ell}}\sum_{\substack{u \in \binset^\ell, s, u', s' \in \binset^\ell, b \in \binset \\ \text{s.t. } u = u' \land \innerprodF{u}{s} \oplus b =  \innerprodF{u'}{s'} \oplus b}}\Tr\rbra*{\Pi^*_{s'}X_{x, s}} + 1 - p_x\\
=& \frac{1}{2^{2\ell}} \cdot 2^\ell \cdot 2\sum_{s \in \binset^\ell}\Tr\rbra*{\Pi^*_{s}X_{x, s}} + \frac{1}{2^{2\ell}}\sum_{\substack{u \in \binset^\ell, s, u', s' \in \binset^\ell, b \in \binset \\ \text{s.t. } s \neq s' \land u = u' \land \innerprodF{u}{s} \oplus b =  \innerprodF{u'}{s'} \oplus b}}\Tr\rbra*{\Pi^*_{s'}X_{x, s}} + 1 - p_x\\
=& \frac{1}{2^{\ell - 1}}\sum_{s \in \binset^\ell}\Tr\rbra*{\Pi^*_{s}X_{x, s}} + \frac{1}{2^{2\ell}} \cdot 2^{\ell - 1} \cdot 2\sum_{\substack{s, s' \in \binset^\ell\\ \text{s.t. } s \ne s'}}\Tr\rbra*{\Pi^*_{s'}X_{x, s}} + 1 - p_x\\
=& \frac{1}{2^{\ell - 1}}\sum_{s \in \binset^\ell}\Tr\rbra*{\Pi^*_{s}X_{x, s}} + \frac{1}{2^{\ell}} \sum_{s \in \binset^\ell}\Tr\rbra*{(I - \Pi^*_{s})X_{x, s}} + 1 - p_x\\
=& \frac{1}{2^{\ell}}\sum_{s \in \binset^\ell}\Tr\rbra*{\Pi^*_{s}X_{x, s}} + \frac{1}{2^{\ell}}\sum_{s \in \binset^\ell}\Tr\rbra*{X_{x, s}} + 1 - p_x\\
\geq & \frac{1}{2^{\ell}}a(n) + 1\enspace,
\end{align*}
where we split the summation according to whether $s$ equals $s'$ in the fifth line, and the sixth line follows from the fact that there are exactly $2^{\ell - 1}$ strings $u \in \binset^\ell$ such that $\innerprodF{u}{s} \oplus b =  \innerprodF{u}{s'} \oplus b$ for $s \ne s'$, and in the last line, we plug in the definition of $p_x$ and use \Cref{eqn:condition-for-yes-instance-QIP-logbit}.

By the operational interpretation of trace distance (i.e., the Holevo--Helstrom bound~\cite{Holevo73TraceDist,Helstrom69}), it follows that
\begin{align*}
    \TD\rbra*{\sigma_{x, 0}, \sigma_{x, 1}} \geq& \Tr\rbra*{\sigma_{x, 0} M_0} - \Tr\rbra*{\sigma_{x, 1} M_0}\\
    =& \Tr\rbra*{\sigma_{x, 0} M_0} + \Tr\rbra*{\sigma_{x, 1} M_1} - 1\\
    \geq& \frac{a(n)}{2^{\ell(n)}} = \frac{a'(n')}{2^{\ell'(n')}}\enspace,
\end{align*}
which implies that $(Q_{x, 0}, Q_{x, 1})$ is a \emph{yes} instance for $\QSD\sbra*{\frac{a'}{2^{\ell'}}, \frac{1 + 2^{\ell'/2}}{2} \cdot \frac{b'}{2^{\ell'}}}$.

\paragraph{Soundness.} 
Intuitively, the hash family $\{h_{u, v}\}_{u\in \binset^\ell, v \in \binset}$ is a standard pairwise-independent hash family, and thus it is a strong quantum randomness extractor. Informally, with strong quantum randomness extractors, one can extract nearly uniform randomness $h_{u,v}(s)$ from an imperfect randomness source $s$, even when $X_{x, s}$ and the hash seed $(u,v)$ are leaked to the adversary, as long as it is impossible to guess $s$ from $X_{x, s}$ with high probability. This intuition can be formalized by comparing the two states $\varsigma$ and $\varsigma'$ in \Cref{lemma:extractor-works}, where $\varsigma$ is the joint state of a uniformly random bit and the adversary's state, $\varsigma'$ is the joint state of the hash result $h_{u, v}(s)$ and the adversary's state, and both states are normalized with a dummy state $\ket{1, 0, 0^\ell, 0, \bar{0}}$. 

\begin{restatable}[The hashed state index is almost uniformly random]{lemma}{extractorWorks}
\label{lemma:extractor-works}
Let $\{h_{u, v}\}_{u \in \binset^\ell, v \in \binset}$, $p_x$, and $\{X_{x, s}\}_{s \in \binset^\ell}$ be the functions, the number, and the subnormalized quantum states defined above in the proof of \Cref{thm:answer-compression-QSD}.
Define quantum states $\varsigma$ and $\varsigma'$ as below:
\begin{align*}
    \varsigma &\coloneq \frac{1}{2^{2\ell + 2}}\sum_{u, s \in \binset^\ell, v, b \in \binset} \ketbra{0}{0} \otimes \ketbra{b}{b} \otimes \ketbra{u, v}{u, v} \otimes X_{x, s} + (1 - p_x) \ketbra{1, 0, 0^\ell, 0, \bar{0}}{1, 0, 0^\ell, 0, \bar{0}}\enspace,\\
    \varsigma' &\coloneq \frac{1}{2^{2\ell + 1}}\sum_{u, s \in \binset^\ell, v \in \binset}\ketbra{0}{0} \otimes \ketbra{h_{u, v}(s)}{h_{u, v}(s)} \otimes \ketbra{u, v}{u, v} \otimes X_{x, s}  + (1 - p_x) \ketbra{1, 0, 0^\ell, 0, \bar{0}}{1, 0, 0^\ell, 0, \bar{0}}\enspace.  
\end{align*} 
If $\max\limits_{\{M^*_s\}_{s \in \binset^\ell}} \sum_{s \in \binset^\ell}\Tr\rbra*{M^*_s X_{x, s}} \leq b$, then the trace distance between $\varsigma$ and $\varsigma'$ is bounded by
\[\TD\rbra*{\varsigma, \varsigma'} \leq \frac{1 + 2^{\ell/2}}{2^{\ell + 2}} \cdot b\enspace.\]
\end{restatable}

The proof of \Cref{lemma:extractor-works} is deferred to \Cref{subsection:proof-of-extractor}. It remains to show the soundness using \Cref{lemma:extractor-works}.

We relate the difference between $\varsigma$ and $\varsigma'$ to the difference between $\sigma_0$ and $\sigma_1$:\footnote{Strictly speaking, we identify tensor products that differ only by the canonical permutation of their classical registers. Equivalently, the equality below holds up to conjugation by the unitary that reorders these registers.}
\begin{align*}
    \varsigma - \varsigma' =& \frac{1}{2^{2\ell + 2}}\sum_{u, s \in \binset^\ell, v, b \in \binset} (-1)^{h_{u, v}(s) + b + 1} \ketbra{0, b}{0, b} \otimes \ketbra{u, v}{u, v} \otimes X_{x, s}\\
    =&\ketbra{0}{0} \otimes \rbra*{\ketbra{1}{1} - \ketbra{0}{0}} \otimes \frac{1}{2^{2\ell + 2}}\sum_{u \in \binset^\ell, v \in \binset} \ketbra{u, v}{u, v} \otimes \sum_{s \in \binset^\ell}(-1)^{h_{u, v}(s)}X_{x, s}\\
    =& \frac{1}{4}\rbra*{\ketbra{1}{1} - \ketbra{0}{0}}\otimes (\sigma_{x, 0} - \sigma_{x, 1})\enspace,
\end{align*}
which implies that
\begin{align}\label{eqn:TD-relation}
    \TD\rbra*{\varsigma, \varsigma'} = \frac{1}{4} \cdot 2 \cdot \TD\rbra*{\sigma_{x, 0}, \sigma_{x, 1}} = \frac{1}{2}\TD\rbra*{\sigma_{x, 0}, \sigma_{x, 1}}\enspace.
\end{align}

For a \emph{no} instance $x \in \GetNo{\PromiseProblemI}$, $\omega(V(x)) \leq b(n)$. By \Cref{eqn:TD-relation,eqn:optimal-acceptance-prob-multiple-2,lemma:extractor-works},
\begin{align*}
    \TD\rbra*{\sigma_{x, 0}, \sigma_{x, 1}} = 2 \cdot \TD\rbra*{\varsigma, \varsigma'} \leq \frac{1 + 2^{\ell(n)/2}}{2^{\ell(n) + 1}} \cdot b(n) = \frac{1 + 2^{\ell'(n')/2}}{2^{\ell'(n') + 1}} \cdot b'(n')\enspace,
\end{align*}
which implies that $(Q_{x, 0}, Q_{x, 1})$ is a \emph{no} instance for $\QSD\sbra*{\frac{a'}{2^{\ell'}}, \frac{1 + 2^{\ell'/2}}{2} \cdot \frac{b'}{2^{\ell'}}}$.
\end{proof}

\section{Public coins weaken quantum proof systems with a laconic prover}

In this section, we extend the phenomenon that \emph{classical public coins} make interaction useless in laconic interactive proof systems, where the prover sends $O(\log{n})$ bits in the classical setting~\cite{GH98} or $O(\log{n})$ qubits in the quantum setting, to \emph{quantum public coins}: 
\begin{itemize}
    \item In \Cref{subsec:qcQAM-complete}, we show that estimating the $2^\ell$-ary steering-game value, which is inspired by ``steering'' in quantum information theory and naturally aligns with two-message quantum public-coin proof systems with an $\ell$-bit prover response, is $\qcQAM[\ell]$-complete. 
    \item In \Cref{subsec:qcQAM[1]=BQP}, since the \emph{binary} steering-game value admits a closed-form expression~\cite[Section~II.C]{AM14}, we prove that $\qcQAM[1]$ with \emph{inverse-polynomial} gap, where the prover sends a single bit, collapses to \BQP{} (\Cref{thm:qcQAM[1]=BQP}).
    \item In \Cref{subsec:qcQAM[sqrt(logn)]=BQP}, we establish that $\qcQAM[O(\sqrt{\log{n}})]$ with \emph{constant} gap, where the prover sends $O(\sqrt{\log{n}})$ bits, collapses to \BQP{} (\Cref{thm:qcQAM[sqrt log n]=BQP}). 
\end{itemize}


\subsection{A complete problem for \texorpdfstring{$\qcQAM[\ell]$}{} via \texorpdfstring{$2^\ell$}{}-ary steering-game values}
\label{subsec:qcQAM-complete}

The steering-game value defined in \Cref{def:steering-prob} is a restricted and rescaled version of the optimal success quantity in \emph{generalized
state discrimination}~\cite[Equation~(14)]{AM14},
\begin{equation}
    \label{eq:generalized-state-discrimination}
    \max_{\text{POVM} \cbra{E_r}} \sum_{r\in\binset^\ell} \Tr(A_r E_r)\enspace,
\end{equation}
where the operators $\cbra{A_r}$ are allowed to be arbitrary positive semidefinite operators. 
One can recover the optimal success probability of the quantum state discrimination used in \MultiQSD{} (\Cref{def:MultiQSD}) by setting $A_r = \rho_r/2^{\ell}$. To obtain the steering-game value from \Cref{eq:generalized-state-discrimination}, one restricts $A_r$ from PSD operators to $0 \preceq A_r \preceq I$ and sets $E_r=\sigma_r d_\sfR$ for each $r\in\binset^\ell$.

\begin{definition}[Steering-game value problem, \SteerVal{}]
    \label{def:steering-prob}
    In a steering game, Alice and Bob share $m$ EPR pairs, with Bob's halves stored in the register $\sfR$. Alice performs an arbitrary POVM  $\cbra{M_r}_{r\in\binset^\ell}$ on her halves and sends the outcome $r$, stored in the register $\sfI$, to Bob. Bob applies a quantum circuit $Q$ of description length $n$ to $(\sfI,\sfA,\sfR)$, with the register $\sfA$ initialized to $\ket{\bar{0}}$, and measures the output qubit in the computational basis. Alice wins if the outcome is $1$.
    
    \noindent Let $Q_r$ denote $Q$ with $\sfI$ initialized to $\ket{r}$,
    and let $\Pi_r \coloneqq \rbra{\bra{\bar{0}}_{\sfA}\otimes I_\sfR} Q_r^\dagger \ketbra{1}{1}_\Out Q_r \rbra{\ket{\bar{0}}_\sfA\otimes I_\sfR}$ be the induced acceptance effect on $\sfR$. The \emph{$2^\ell$-ary steering-game value} $\MSV{\ell}{Q}$ is the maximum winning probability over Alice's measurements. Equivalently,
    \[ \MSV{\ell}{Q} \coloneqq \max_{\substack{\sigma_0,\cdots,\sigma_{2^\ell-1}\succeq 0\\ \sum_{r\in\binset^\ell} \sigma_r = I_\sfR/d_\sfR}} \sum_{r\in\binset^\ell} \Tr(\Pi_r\sigma_r), \quad\text{where } d_\sfR\coloneqq\dim(\sfR)\enspace.\]

    \noindent Let $a(n)$ and $b(n)$ be efficiently computable functions such that $0 \leq b(n) < a(n) \leq 1$. The promise problem $\SteerVal[\ell(n),a(n),b(n)]$ asks to distinguish between the following two cases: 
    \begin{itemize}
        \item \emph{Yes}: The quantum circuit $Q$ satisfies $\MSV{\ell}{Q}\geq a(n)$;
        \item \emph{No}: The quantum circuit $Q$ satisfies $\MSV{\ell}{Q}\leq b(n)$.
    \end{itemize}
    When $\ell=1$, we denote the corresponding \emph{binary} version by $\BinSteerVal[a(n),b(n)]$, and write $\BSV{Q_0}{Q_1}$ for the binary steering-game value. 
\end{definition}

For convenience, from an information-theoretic perspective, we also adopt the notation $\MSV{\ell}{\Pi}$, where $\Pi \coloneqq \cbra{\Pi_r}_{r\in\binset^\ell}$, and $\BSV{\Pi_0}{\Pi_1}$ for the steering-game value and its binary version, respectively.

\begin{lemma}[$\SteerVal_{2^\ell}$ is {$\qcQAM[\ell]$}-complete]
    \label{lemma:SteerVal-qcQAM[l]-complete}
    For any efficiently computable functions $c(n)$ and $s(n)$ such that $0 \leq s(n) < c(n) \leq 1$,
    \[ \SteerVal[\ell(n),c(n),s(n)] \text{ is } \qcQAM[\ell(n),c(n),s(n)]\text{-complete}\enspace.\]
\end{lemma}

\begin{proof}
    We first prove that $\SteerVal[\ell(n),c(n),s(n)] \in \qcQAM[\ell(n),c(n),s(n)]$, and then show that $\SteerVal[\ell(n),c(n),s(n)]$ is $\qcQAM[\ell(n),c(n),s(n)]$-hard.

    \parheading{The $\qcQAM[\ell]$ containment.}
    We start by presenting a simple $\qcQAM[\ell(n),c(n),s(n)]$ protocol for $\SteerVal[\ell(n),c(n),s(n)]$, as described in \Cref{protocol:SV-in-qcQAM[ell]}.

\begingroup
\LinesNotNumbered
\begin{algorithm}[!ht]
    \SetAlgorithmName{Protocol}{protocol}{List of Protocols}
    \caption{A $\qcQAM[\ell]$ protocol for $\SteerVal_{2^\ell}$.}
	\label{protocol:SV-in-qcQAM[ell]}
    \SetEndCharOfAlgoLine{.}
    \SetKwFor{While}{}{:}{}
    \SetKwFor{For}{For}{:}{}
    \SetKwIF{If}{ElseIf}{Else}{If}{:}{elif}{Else:}{}%
    \SetKwComment{Comment}{// }{}
    \SetKwInOut{Input}{Input}
    \SetKwInOut{Output}{Output}

    \Input{A quantum circuit $Q$ that induces $\cbra{Q_r}_{r\in\binset^\ell}$ by initializing $\sfI$ to $\ket{r}$ for each $r\in\binset^\ell$.}
    \Output{ACCEPT or REJECT.}
    \medskip

    \textbf{1.} $V$ prepares $m$ EPR pairs $\ket{\Phi} \coloneqq \frac{1}{2^{m/2}} \sum_{s\in\binset^m} \ket{s,s}_{\sfS\sfR}$, and sends the register $\sfS$ to $P$ while keeping the register $\sfR$.
    \medskip
    
    \textbf{2.} $V$ receives an $\ell$-bit string $r\in\binset^\ell$ from $P$.
    An honest prover chooses an optimal feasible family $\cbra{\sigma_t^\star}_{t\in\binset^\ell}$ for $\MSV{\ell}{Q}$ and measures $\sfS$ using the POVM $\cbra{M_t}_{t\in\binset^\ell}$ with $M_t\coloneqq 2^m(\sigma_t^\star)^T$.
    \medskip
    
    \textbf{3.} $V$ applies the circuit $Q_r$ to the registers $(\sfA,\sfR)$, with $\sfA$ initialized to $\ket{\bar{0}}$, and accepts iff the output qubit of $Q_r$ is measured to be $1$.
\end{algorithm}
\endgroup

    As in the proof of \Cref{lemma:MultiQSD-is-in-QIP-ell-bit}, every prover strategy in \Cref{protocol:SV-in-qcQAM[ell]} can be described by a POVM $\cbra{M_r}_{r\in\binset^\ell}$ on the received register $\sfS$, where $M_r = d_\sfR \sigma_r^T$ and $d_\sfR = \dim(\sfR)=2^m$, and every such POVM gives a valid prover strategy. Consequently, 
    \begin{subequations}
    \label{eq:MSV-protocol-pacc}
    \begin{align}
    \max_P \Pr[(\protocol{P}{V})(Q)\text{ accepts}] 
    &= \max_{\text{POVM } \{M_r\}}\sum_{r\in\binset^\ell} \bra{\Phi}(M_r\otimes\Pi_r)\ket{\Phi}\\
    &= \max_{\text{POVM } \{M_r\}}\sum_{r\in\binset^\ell} \frac{\Tr(\Pi_r M_r^T)}{d_\sfR}
    = \MSV{\ell}{Q}\enspace. 
    \end{align}
    \end{subequations}
    Here, the last equality follows because $\sigma_r=M_r^T/d$ gives a bijection between the POVM $\cbra{M_r}$ and the feasible families in \Cref{def:steering-prob}. 
    Therefore, we establish the correctness of \Cref{protocol:SV-in-qcQAM[ell]}:
    \begin{itemize}
        \item For \emph{yes} instances, an optimal feasible family $\cbra{\sigma_r^\star}$ for $\MSV{\ell}{Q}$ specifies an honest prover through $M_r=d_\sfR \rbra{\sigma^\star_r}^T$.\footnote{The transpose is taken in the computational basis defining $\ket{\Phi}$.} 
        \item For \emph{no} instances, every feasible family $\cbra{\sigma'_r}$ for $\MSV{\ell}{Q}$ induces a $2^\ell$-outcome POVM $\cbra{M'_r}$, so the soundness bound follows directly from $\MSV{\ell}{Q}\leq s(n)$.
    \end{itemize}

    \sloppypar
    \parheading{The $\qcQAM[\ell]$-hardness.}
    We now establish that $\SteerVal[\ell(n),c(n),s(n)]$ is hard for $\qcQAM[\ell(n),c(n),s(n)]$. Let $\protocol{P}{V}$ be a $\qcQAM[\ell(n),c(n),s(n)]$ protocol in the above EPR-message normal form, and fix an input $x$. It suffices to construct a quantum circuit $Q$, with induced effective acceptance effects $\cbra{\Pi_r}_{r\in\binset^\ell}$, such that the maximum acceptance probability of $\protocol{P}{V}$, denoted by $\omega(V)$, equals the steering-game value $\MSV{\ell}{Q}$. 

    Let $Q\coloneqq Q_x$ be the verifier's action $V_2(x)$ after receiving the prover's $\ell$-bit message $r$ stored in the register $\sfI$. The same calculation as in \Cref{eq:MSV-protocol-pacc} shows that $\MSV{\ell}{Q}=\omega(V)$ as desired. Consequently, the reduction preserves $\omega(V)$ exactly, as well as the prover's response length and both promise thresholds.
\end{proof}


\subsection{\texorpdfstring{$\qcQAM[1] = \BQP$}{}}
\label{subsec:qcQAM[1]=BQP}

\begin{theorem}\label{thm:qcQAM[1]=BQP}
    For any efficiently computable functions $c(n)$ and $s(n)$ such that $0 \leq s(n) < c(n) \leq 1$ and $c(n)-s(n) \geq 1/\poly(n)$, 
    \[\qcQAM\sbra*{1, c(n), s(n)} = \BQP \enspace.\]
\end{theorem}

Before establishing \Cref{thm:qcQAM[1]=BQP}, we need a closed-form expression for the binary steering-game value, which is a direct consequence of the Holevo--Helstrom formula for \emph{binary} generalized state discrimination~\cite[Section~II.C]{AM14}:\footnote{See \Cref{footnote:binary-generalized-state-discrimination}.}
\begin{lemma}[Closed form of the binary steering value]
    \label{lemma:BinSteerVal-closed-form}
    Let $(\Pi_0,\Pi_1)$ be the pair of effective acceptance effects induced on the quantum register $\sfR$. Then the binary steering value admits the closed-form expression
    \begin{equation}
        \label{eq:BSV-close-form}
        \BSV{\Pi_0}{\Pi_1} = \frac{\Tr(\Pi_0+\Pi_1) + \norm{\Pi_0-\Pi_1}_1}{2d}, \quad\text{where }d=\dim(\sfR)\enspace.
    \end{equation}
\end{lemma}

We also need the following exact block-encoding gadgets:

\begin{lemma}[Exact block-encoding gadgets]
    \label{lemma:exact-block-encoding-gadgets}
    Let $(Q_0,Q_1)$ be a \BinSteerVal{} instance, and let $\Pi_r$ be the corresponding effective acceptance effect on the quantum register $\sfR$ induced by $Q_r$ for $r\in\binset$, and let $a$ be the number of ancillary qubits used by $Q_r$. One can efficiently construct the following exact block-encodings, where $a(n)$ denotes the number of ancillary qubits:
    \begin{enumerate}[label={\upshape(\arabic*)}]
        \item \label{thmitem:verifier-branch-gadget} For $r\in\binset$, $U_r \coloneqq Q_r^\dagger (I_{\sfO}-2\ketbra{1}{1}_{\sfO}) Q_r$ is a $(1,a,0)$-block-encoding of $I_{\sfR}-2 \Pi_r$. 
        \item \label{thmitem:effects-difference-gadget} There is a $(1,a+1,0)$-block-encoding $U_-$ of $\Pi_0-\Pi_1$ using one query to each of controlled-$U_0$ and controlled-$U_1$, and $O(1)$ one-qubit gates.
    \end{enumerate}
\end{lemma}

\begin{proof}
    We first establish \Cref{thmitem:verifier-branch-gadget} for $r\in\binset$ via a direct calculation: 
\begin{align*}
    &\rbra{\bra{\bar{0}}_{\sfA}\otimes I_{\sfR}}U_r\rbra{\ket{\bar{0}}_{\sfA}\otimes I_{\sfR}}\\
    =~& \rbra{\bra{\bar{0}}_{\sfA}\otimes I_{\sfR}}Q_r^\dagger (I_{\sfO}-2\ketbra{1}{1}_{\sfO}) Q_r\rbra{\ket{\bar{0}}_{\sfA}\otimes I_{\sfR}} \\
    =~& \rbra{\bra{\bar{0}}_{\sfA}\otimes I_{\sfR}}Q_r^\dagger Q_r\rbra{\ket{\bar{0}}_{\sfA}\otimes I_{\sfR}}  
    - 2\rbra{\bra{\bar{0}}_{\sfA}\otimes I_{\sfR}}Q_r^\dagger \ketbra{1}{1}_{\sfO} Q_r\rbra{\ket{\bar{0}}_{\sfA}\otimes I_{\sfR}} \\
    =~& I_{\sfR}-2\Pi_r\enspace.
\end{align*}

For \Cref{thmitem:effects-difference-gadget}, we apply \Cref{lemma:lcu} to the block encodings $U_0$ and $U_1$ from \Cref{thmitem:verifier-branch-gadget}, with coefficients $-1/2$ and $1/2$, respectively.
Since the coefficient magnitudes sum to one, the resulting unitary $U_-$ is a $(1,a+1,0)$-block encoding of $-\rbra*{I_{\sfR}-2\Pi_0}/2 +\rbra*{I_{\sfR}-2\Pi_1}/2 =\Pi_0-\Pi_1$. The construction uses one controlled query to each of $U_0$ and $U_1$ and $O(1)$ one-qubit gates.
\end{proof}

We are now ready to proceed with the proof of \Cref{thm:qcQAM[1]=BQP}:

\begin{proof}[Proof of \Cref{thm:qcQAM[1]=BQP}]
    Since $\BQP \subseteq \qcQAM[1]$ holds straightforwardly, it suffices to prove that $\qcQAM[1] \subseteq \BQP$. To this end, we consider the $\qcQAM[1]$-complete problem \BinSteerVal{} (\Cref{lemma:SteerVal-qcQAM[l]-complete} with $\ell=1$). Together with the closed form of the binary steering value (\Cref{lemma:BinSteerVal-closed-form}), it suffices to construct a quantum algorithm that estimates the binary steering value within additive error strictly smaller than $\Delta/4$ and compares the estimate with the midpoint threshold $(c+s)/2=s+\Delta/2$, using polynomially many queries to $Q_0$ and $Q_1$ in $1/\Delta$, where $\Delta \coloneqq c(n)-s(n)$ denotes the promise gap of the $\qcQAM[1,c(n),s(n)]$ protocol.
    
    More specifically, for any \BinSteerVal{} instance $(Q_0,Q_1)$, with induced effective acceptance effects $\Pi_0$ and $\Pi_1$ on the register $\sfR$, where $d_\sfR=\dim(\sfR),$ it suffices to decide whether $\BSV{\Pi_0}{\Pi_1}$ is at least $c$ or at most $s$, where  
    \begin{equation}
        \label{eq:BSV-three-terms}
        \BSV{\Pi_0}{\Pi_1} = \frac{\Tr(\Pi_0)}{2d_\sfR} + \frac{\Tr(\Pi_1)}{2d_\sfR} + \frac{\norm{\Pi_0 - \Pi_1}_1}{2d_\sfR}\enspace.
    \end{equation} 

    We now provide the main algorithm, as presented in \algoref{algo:BSV-estimation}. In the remainder of the proof, we establish the correctness of \algoref{algo:BSV-estimation} and analyze its efficiency.

\begingroup
\LinesNumbered
\begin{algorithm}[!ht]
    \SetAlgorithmName{Algorithm}{algorithm}{List of Algorithms}
    \caption{Quantum algorithm for estimating $\BSV{\Pi_0}{\Pi_1}$.}
	\label{algo:BSV-estimation}
    \SetEndCharOfAlgoLine{.}
    \setlength{\parskip}{5pt}
    \SetKwFor{While}{}{:}{}
    \SetKwFor{For}{For}{:}{}
    \SetKwIF{If}{ElseIf}{Else}{If}{:}{elif}{Else:}{}%
    \SetKwComment{Comment}{// }{}
    \SetKwInOut{Input}{Input}
    \SetKwInOut{Output}{Output}

    \Input{Quantum circuits $Q_0$ and $Q_1$ corresponding to induced effective acceptance effects $\Pi_0$ and $\Pi_1$ on the $m$-qubit register $\sfR$, respectively; $\Delta \in (0,1]$.}
    \Output{An estimate of $\BSV{\Pi_0}{\Pi_1}$ with high probability.}

    {\color{gray}\Comment{Estimating $\Tr(\Pi_r)/d_{\sfR}$ for $r\in\binset$.}}
    
    For $r\in\binset$, implement a unitary operator $U_r$ that is a $(1,a,0)$-block-encoding of $I_\sfR - 2\Pi_r$ via \Cref{lemma:exact-block-encoding-gadgets}\ref{thmitem:verifier-branch-gadget}, using one query to each of $Q_r$ and $Q_r^\dagger$\;

    For $r\in\binset$, perform the normalized trace estimation (\Cref{corr:normalized-trace-estimation}) on $U_r$, and denote the estimate by $\widehat{x}_r$. 

    {\color{gray}\Comment{Estimating $\norm{\Pi_0-\Pi_1}_1/d_{\sfR}$.}}
    
    Let $P^{\rm abs}_d$ be the degree-$d$ polynomial that approximates $|x|/4$ in the range $[-1,1]$, which is determined according to $\Delta$ and $n$ and is constructed from \Cref{lemma:abs-polynomial-approx}. 

    Implement a unitary operator $U_-$ that is a $(1,a+1,0)$-block-encoding of $\Pi_0-\Pi_1$ via \Cref{lemma:exact-block-encoding-gadgets}\ref{thmitem:effects-difference-gadget}, using one query to each of controlled-$U_0$ and controlled-$U_1$\; 

    Implement a unitary operator $U_{P^{\rm abs}_d(\Pi_0-\Pi_1)}$ that is a block-encoding of $P^{\rm abs}_d(\Pi_0-\Pi_1)$ by QSVT (\Cref{thm:qsvt-hermitian}), using $O(\deg(P^{\rm abs}_d))$ queries to $U_-$.

    Perform the normalized trace estimation (\Cref{corr:normalized-trace-estimation}) on $U_{P^{\rm abs}_d(\Pi_0-\Pi_1)}$, denoted the estimate by $\widehat{y}$. 

    {\color{gray}\Comment{Putting everything together.}}

    Return $1/2-(\widehat{x}_0+\widehat{x}_1)/4+2\widehat{y}$. 
\end{algorithm}
\endgroup

\paragraph{Estimating $\Tr(\Pi_r)/d_{\sfR}$ for $r\in\binset$.} We begin by constructing a $(1,a,0)$-block-encoding of $I_\sfR - 2\Pi_r$ for $r\in\binset$, denoted by $U_r$, using \Cref{lemma:exact-block-encoding-gadgets}\ref{thmitem:verifier-branch-gadget}. 
Let $x_r \coloneqq \Tr(I_{\sfR}-2\Pi_r)/d_{\sfR}$ for $r\in\binset$. Then, a direct calculation shows that 
\begin{equation}
    \label{eq:Pi-estimation}
    \forall r\in\binset, \quad \frac{\Tr(\Pi_r)}{d_{\sfR}} = \frac{1}{2} \rbra*{ \frac{\Tr(I_{\sfR})}{d_{\sfR}} - \frac{\Tr(I_{\sfR}-2\Pi_r)}{d_{\sfR}}} = \frac{1-x_r}{2}\enspace.
\end{equation}

For each $r\in\binset$, we apply the normalized trace estimation (\Cref{corr:normalized-trace-estimation}) to $U_r$ with $\epsilon=\epsilon_1 \coloneqq \Delta/8$, repeat the procedure $5$ times independently, and take the median. Denote the resulting estimate by $\widehat{x}_r$, which has additive error at most $\epsilon_1$ with
probability at least $11/12$.\footnote{Since a single run succeeds with probability at least $8/\pi^2 > 4/5$, it follows that the $5$-run median succeeds with probability at least $\sum_{j=3}^5 \binom{5}{j} \rbra[\big]{\frac{4}{5}}^j \rbra[\big]{\frac{1}{5}}^{5-j} > 11/12$. \label{footnote:median-repeation}} 
Conditioned on this success event, \Cref{eq:Pi-estimation} implies that 
\begin{equation}
    \label{eq:Pi-estimation-bound}
    \forall r\in\binset, \quad \abs*{ \frac{1-\widehat{x}_r}{2} - \frac{\Tr(\Pi_r)}{d_{\sfR}} } = \frac{\abs*{x_r-\widehat{x}_r}}{2} \leq \frac{\epsilon_1}{2}\enspace.
\end{equation}

\paragraph{Estimating $\norm{\Pi_0-\Pi_1}_1/d_{\sfR}$.} 
We start by constructing a $(1,a+1,0)$-block-encoding of $\Pi_0-\Pi_1$, denoted by $U_-$, using  \Cref{lemma:exact-block-encoding-gadgets}\ref{thmitem:effects-difference-gadget}. 

We now set $\eta\coloneq \Delta/16$, which satisfies $\eta \in (0,1/16] \subseteq (0,1/2)$ since $0<\Delta(n)<1$. 
Let $P_\eta$ be the polynomial from \Cref{lemma:abs-polynomial-approx}, whose parity coincides with its degree. We define $P^{\rm abs}_d \coloneqq P_\eta/2$, where the degree $d=\deg\rbra[\big]{P^{\rm abs}_d} = \deg(P_\eta) \leq \ceil*{\beta/\eta} = \ceil*{16\beta/\Delta}$.
Furthermore, $P^{\rm abs}_d$ satisfies the following bounds:
\begin{subequations}
\label{eq:abs-polynomial-approx-bound}
\begin{align}
    \max_{x\in[-1,1]} \abs*{ P^{\rm abs}_d(x) } &= \frac{1}{2}\max_{x\in[-1,1]} \abs*{P_{\eta}(x)} \leq \frac{1}{2}\enspace,\\
    \max_{x\in[-1,1]}  \abs*{ P^{\rm abs}_d(x) - \frac{|x|}{4} } &= \frac{1}{2} \max_{x\in[-1,1]}  \abs*{ P_\eta(x) - \frac{|x|}{2} } \leq \frac{\eta}{2} = \frac{\Delta}{32}\enspace. 
\end{align}
\end{subequations}

For convenience, we set $A\coloneqq \Pi_0-\Pi_1$.
Since $\norm{A} \leq 1$,\footnote{The fact $\norm{A} \leq 1$ follows directly from $-I_\sfR \preceq A \preceq I_\sfR$, which can be obtained by combining $A \preceq \Pi_0 \preceq I_\sfR$ and $A \succeq -\Pi_1 \succeq -I_\sfR$.} one can construct a $(1,a+3,\delta)$-block-encoding of $P^{\rm abs}_d(A)$, denoted by $U_{P^{\rm abs}_d(A)}$, whose ancillary qubit register is $\sfB$, by using the QSVT (\Cref{thm:qsvt-hermitian}) associated with $P^{\rm abs}_d$ and $\delta\coloneqq\Delta/64$. 
It follows that
\begin{equation}
    \label{eq:implementation-error-bound}
    \norm*{ \bra{\bar{0}}_{\sfB} U_{P^{\rm abs}_d(A)} \ket{\bar{0}}_{\sfB} - P^{\rm abs}_d(A) } \leq \delta\enspace.
\end{equation}
Because $A$ is Hermitian and its spectrum lies in $[-1,1]$, \Cref{eq:abs-polynomial-approx-bound} implies 
\begin{equation}
    \label{eq:td-polynomial-error-bound}
    \norm*{ P^{\rm abs}_d(A) - \frac{\abs*{A}}{4}} \leq \frac{\Delta}{32}\enspace.
\end{equation}

\sloppypar
Let $y \coloneqq \norm{A}_1/(4 d_{\sfR})$. 
We apply the normalized trace estimation (\Cref{corr:normalized-trace-estimation}) to $\bra{\bar{0}}_\sfB U_{P^{\rm abs}_d(A)} \ket{\bar{0}}_{\sfB}$ with $\epsilon = \epsilon_2 \coloneqq \Delta/64$, repeat the procedure $5$ times independently, and take the median. Let $\widehat{y}$ denote the resulting estimate. By \Cref{footnote:median-repeation}, this estimate has additive error at most $\epsilon_2$ with probability at least $11/12$, and on the success event we obtain:
\begin{subequations}
\label{eq:Pi-diff-estimation-bound}
\begin{align}
    \abs*{ \widehat{y} - y } 
    &= \abs*{ \widehat{y} - \frac{\Tr\rbra*{ \abs{A}/4 }}{d_\sfR} } \\
    &\leq \abs*{ \widehat{y} - \frac{\Tr\rbra[\big]{ \bra{\bar{0}}_{\sfB} U_{P^{\rm abs}_d(A)} \ket{\bar{0}}_{\sfB} }}{d_\sfR}  } + \frac{1}{d_\sfR} \abs*{ \Tr\rbra[\big]{ \bra{\bar{0}}_{\sfB} U_{P^{\rm abs}_d(A)} \ket{\bar{0}}_{\sfB} - P^{\rm abs}_d(A) } }\\
     &\qquad+ \frac{1}{d_\sfR} \abs*{ \Tr\rbra*{ P^{\rm abs}_d(A) - \frac{\abs*{A}}{4} } }\\
     &\leq \epsilon_2 + \norm*{ \bra{\bar{0}}_{\sfB} U_{P^{\rm abs}_d(A)} \ket{\bar{0}}_{\sfB}  - P^{\rm abs}_d(A) } + \norm*{ P^{\rm abs}_d(A) - \frac{\abs*{A}}{4} }\\
     &\leq \epsilon_2 + \delta + \frac{\Delta}{32}\enspace. 
\end{align}
\end{subequations}
Here, the second line uses the triangle inequality, the fourth line follows from the fact that $\abs{\Tr(X)} \leq d_\sfR \norm{X}$ for an operator $X$ acting on $\sfR$, where $\norm{X}$ equals the largest singular value $\sigma_{\max}(X)$, and the last line is guaranteed by the implementation error bound in \Cref{eq:implementation-error-bound} and the approximation error bound in \Cref{eq:td-polynomial-error-bound}. 

\paragraph{Combining the estimates.}
We first combine the estimates, which correspond to three terms on the right-hand side of \Cref{eq:BSV-three-terms}. The total additive error is strictly smaller than $\Delta/4$, and each estimate succeeds with probability at least $11/12$. Consequently, by the union bound, all three estimates succeed simultaneously with probability at least $1 - 3 (1-11/12) = 3/4$. On this simultaneous success event, combining \Cref{eq:BSV-three-terms,eq:Pi-estimation-bound,eq:Pi-diff-estimation-bound} implies that 
\begin{align*}
    \abs*{\rbra*{ \frac{1}{2}-\frac{{\widehat{x}_0}+\widehat{x}_1}{4}+2\widehat{y} } - \BSV{\Pi_0}{\Pi_1}}
    &\leq \frac{1}{2}\sum_{r\in\binset}\abs*{ \frac{1-\widehat{x}_r}{2} - \frac{\Tr(\Pi_r)}{d_\sfR} } + 2\abs*{ \widehat{y} - \frac{\norm{A}_1/4}{d_\sfR} } \\
    &\leq \frac{\epsilon_1}{2} + 2\epsilon_2 + 2\delta + \frac{\Delta}{16}\\
    &= \frac{3\Delta}{16} < \frac{\Delta}{4}\enspace.
\end{align*}
Here, the first line uses the triangle inequality, and the last line follows directly from the chosen parameters $\epsilon_1$, $\epsilon_2$, and $\delta$. Therefore, \Cref{algo:BSV-estimation} returns an estimate of $\BSV{\Pi_0}{\Pi_1}$ to within additive error $3\Delta/16$ with probability at least $3/4 > 2/3$. 

\paragraph{Efficiency analysis.}
To analyze the efficiency of \Cref{algo:BSV-estimation}, we observe that the exact block-encodings of $I_{\sfR}-2\Pi_r$ for $r\in\binset$ and $\Pi_0-\Pi_1$ are efficiently implementable by \Cref{lemma:exact-block-encoding-gadgets}. Also, \Cref{corr:normalized-trace-estimation} implies that the number of queries to $U_0$ and $U_1$ is $O(1/\epsilon_1)=O(1/\Delta)$ for estimating $\Tr(\Pi_r)/d_{\sfR}$ for $r\in\binset$. Similarly, \Cref{corr:normalized-trace-estimation} gives $O(1/\epsilon_2)=O(1/\Delta)$ queries to $U_{P^{\rm abs}_d(\Pi_0-\Pi_1)}$ for estimating $\norm{\Pi_0-\Pi_1}_1/d_{\sfR}$, while \Cref{thm:qsvt-hermitian,lemma:abs-polynomial-approx} show that the number of queries to $U_0$ and $U_1$ for implementing $U_{P^{\rm abs}_d(\Pi_0-\Pi_1)}$ is $O(\deg(P^\mathrm{abs}_d))=O(1/\Delta)$. Putting these estimates together, implementing \Cref{algo:BSV-estimation} uses $O(1/\Delta^2)$ queries to the given circuits $Q_0$ and $Q_1$, which gives the desired \BQP{} containment.
\end{proof}

\subsection{\texorpdfstring{$\qcQAM\sbra*{ O\rbra*{ \sqrt{\log{n}} } }$}{} with constant gap is in \BQP{}}
\label{subsec:qcQAM[sqrt(logn)]=BQP}

\begin{theorem}\label{thm:qcQAM[sqrt log n]=BQP}
For any efficiently computable functions $c(n)$, $s(n)$, and $\ell(n)$ such that $\ell(n) = O(\sqrt{\log n})$, $0 \leq s(n) < c(n) \leq 1$, and the promise gap $\Delta(n) \coloneq c(n) - s(n) = \Theta(1)$,
\[\qcQAM\sbra*{\ell(n), c(n), s(n)} = \BQP \enspace.\]
The underlying quantum algorithm runs in time $\poly(n,\ell,1/\Delta) \exp\rbra*{O\rbra[\big]{\ell^2/\Delta^6}}$, where $n$ denotes the description length specified in \Cref{def:steering-prob}.
\end{theorem}

To prove \Cref{thm:qcQAM[sqrt log n]=BQP}, we consider the problem $\SteerVal\sbra*{\ell(n),c(n),s(n)}$, which is shown to be $\qcQAM\sbra*{\ell(n), c(n), s(n)}$-complete by \Cref{lemma:SteerVal-qcQAM[l]-complete}. Our approach is to apply matrix multiplicative weights to approximate the $2^\ell$-ary steering-game value $\MSV{\ell}{Q}$, and then use the QSVT techniques to estimate the quantities arising in the multiplicative-weights updates without explicitly materializing the exponentially large matrices.

\subsubsection{Approximating the \texorpdfstring{$2^\ell$}{}-ary steering-game value}\label{subsubsec:appoximate-MSV}

To obtain a matrix-multiplicative-weights-friendly approximation of $\MSV{\ell}{Q}$, we first combine the variables appearing in $\MSV{\ell}{Q}$ into a single operator. Define
\[
C \coloneq \sum_{r \in \binset^\ell}\ketbra{r}{r}_{\sfI} \otimes \rbra*{\Pi_r}_{\sfR}\enspace,
\qquad
\sigma \coloneq \sum_{r \in \binset^\ell}\ketbra{r}{r}_{\sfI} \otimes \rbra*{\sigma_r}_{\sfR}\enspace.
\]
Then the objective becomes $\sum_{r\in\binset^\ell} \Tr(\Pi_r\sigma_r)=\Tr(C\sigma)$ and the constraints on $\sigma_r$ imply that $\sigma \in \Dens(\sfI, \sfR)$ and $\Tr_{\sfI}(\sigma)=I_{\sfR}/d_{\sfR}$.
Conversely, for any $\sigma\in\Dens(\sfI,\sfR)$ satisfying $\Tr_{\sfI}(\sigma)=I_{\sfR}/d_{\sfR}$, dephasing the register $\sfI$ in the computational basis preserves both the constraints and the objective $\Tr(C\sigma)$, since $C$ is block diagonal in $\sfI$. Hence it suffices to optimize over all density operators satisfying constraints $\sigma \in \Dens(\sfI, \sfR)$ and $\Tr_{\sfI}(\sigma)=I_{\sfR}/d_{\sfR}$.

Therefore, $\MSV{\ell}{Q}$ can be expressed as
\begin{align}\label{eqn:MSV-group-form}
\MSV{\ell}{Q}=\max_{\substack{\sigma \in \Dens(\sfI,\sfR)\\\Tr_{\sfI}(\sigma)=I_{\sfR}/d_{\sfR}}}\Tr(C\sigma)\enspace.
\end{align}

The SDP in \Cref{eqn:MSV-group-form} fits into the framework studied in~\cite{GW13}. Their rounding theorem allows us to replace the constraint $\Tr_{\sfI}(\sigma)=I_{\sfR}/d_{\sfR}$ by a trace-norm penalty at the cost of an additive approximation error, yielding the following lemma.\footnote{Take $k=2$, the channel $\Phi_1(X)=\Tr(X)I_{\sfR}/d_{\sfR}$, and $\calP=\{I-C\}$ in \cite[Theorem~5 of the full version]{GW13}.}

\begin{lemma}[Adapted from~{\cite[Theorem~3.1]{GW13}}]\label{lemma:approximate-MSV-by-introducing-penalty-on-ell-1-norm}
For $\varepsilon > 0$, define \[\calF_{\varepsilon}(\rho) \coloneq \Tr(C\rho) - \frac{1}{\varepsilon}\norm{\Tr_{\sfI}(\rho) - I_{\sfR}/d_{\sfR}}_1\]
for every state $\rho \in \Dens(\sfI, \sfR)$. Then $\max_{\rho \in \Dens(\sfI, \sfR)}\calF_{\varepsilon}(\rho)$ approximates $\MSV{\ell}{Q}$ as follows: 
\[\MSV{\ell}{Q} \leq \max_{\rho \in \Dens(\sfI, \sfR)}\calF_{\varepsilon}(\rho) \leq \MSV{\ell}{Q} + \varepsilon\enspace.\]
\end{lemma}

\Cref{lemma:approximate-MSV-by-introducing-penalty-on-ell-1-norm} reduces the problem of estimating $\MSV{\ell}{Q}$ to the unconstrained optimization problem $\max_{\rho \in \Dens(\sfI,\sfR)} \calF_{\varepsilon}(\rho)$, in which the marginal constraint $\Tr_{\sfI}(\rho)=I_{\sfR}/d_{\sfR}$ is replaced by a trace-norm penalty. However, the quantity $\norm{\Tr_{\sfI}(\rho)-I_{\sfR}/d_{\sfR}}_1$ does not by itself provide a linear direction along which to reduce the constraint violation, as required by the matrix multiplicative weights method. To obtain such a direction, we use the dual characterization of the trace norm:
\begin{align}\label{eqn:trace-norm-duality}
\max_{\substack{H = H^\dagger\\ \norm{H}_{\infty} \leq 1}}\Tr\rbra*{H\rbra*{\Tr_{\sfI}(\rho)-I_{\sfR}/d_{\sfR}}} = \norm{\Tr_{\sfI}(\rho)-I_{\sfR}/d_{\sfR}}_1\enspace.
\end{align}
Motivated by this duality, we introduce the following function. For $\varepsilon > 0$, $\rho\in\Dens(\sfI,\sfR)$, and Hermitian contraction $H$ on $\sfR$, define
\[
\calG_{\varepsilon}(\rho,H) \coloneq \Tr(C\rho) - \frac{1}{\varepsilon} \Tr\rbra*{H\rbra*{\Tr_{\sfI}(\rho)-I_{\sfR}/d_{\sfR}}}\enspace.
\]
When $H$ approximately attains the maximum in \Cref{eqn:trace-norm-duality}, the quantity $\calG_{\varepsilon}(\rho,H)$ approximates $\calF_{\varepsilon}(\rho)$ from above.

We now apply the matrix multiplicative weights method to $\calG_{\varepsilon}$ where $\varepsilon$ is an accuracy parameter to be specified later. At each iteration, we choose a Hermitian contraction that approximately attains the maximum in \Cref{eqn:trace-norm-duality}. More precisely, let $\xi\in(0,1)$ be an accuracy parameter and let $\eta\in(0,1)$ be a learning rate, to be specified later. We recursively define the states $\rho_t$ and Hermitian contractions $H_t$ as follows:
\begin{itemize}
    \item Initialize $\rho_1 \coloneq \frac{I_{\sfI\sfR}}{2^\ell d_{\sfR}}$ and set $K_1 \coloneq 0$.
    \item For round $t \geq 1$, we choose an arbitrary Hermitian contraction $H_t$ such that \[\Tr(H_t(\Tr_{\sfI}(\rho_t) - I_{\sfR}/d_{\sfR})) \geq \norm{\Tr_{\sfI}(\rho_t) - I_{\sfR}/d_{\sfR}}_1 - \xi\enspace.\]
    Define $A_t \coloneq C - \frac{1}{\varepsilon} I_{\sfI} \otimes H_t$ and $K_{t+1} \coloneq K_t + \eta A_t$. 
    Then the state for the next iteration is defined as \[\rho_{t + 1} \coloneq \frac{e^{K_{t + 1}}}{\Tr(e^{K_{t + 1}})}\enspace.\]
\end{itemize}

To lower bound the average value obtained by this procedure, we use a regret bound for the matrix multiplicative weights update method, stated in terms of the quantum relative entropy in \cite{ALO15}, yielding the following lemma.\footnote{Take $\alpha=\eta$, $F_{t-1}=-A_t$, $X_{t-1}=\rho_t$, and $U=\rho$, using $V_{X_0}(U)=D(\rho\|\rho_1)$ and $\norm{A_t}_\infty\leq W$.}

\begin{lemma}[Adapted from~{\cite[Theorem~3.1 of the full version]{ALO15}}]\label{lemma:estimate-MSV-with-average-G-lower-bound}
Let $W \coloneq 1 + 1/\varepsilon$ and suppose that $0 < \eta W \leq 1$. For every integer $T \geq 1$ and state $\rho$, the states $\rho_1, \ldots, \rho_T$, and operators $A_1, \ldots, A_T$ generated by the above procedure satisfy
\[\sum_{t = 1}^T \Tr(A_t (\rho - \rho_t)) \leq \eta W^2 T + \frac{1}{\eta}D(\rho \| \rho_1)\enspace,\]
where $D(\rho \| \sigma) \coloneq \Tr(\rho (\ln \rho - \ln \sigma))$ denotes the quantum relative entropy.
\end{lemma}

The bound in \Cref{lemma:estimate-MSV-with-average-G-lower-bound} allows us to replace the usual dimension term $\ln(2^\ell d_{\sfR})$ with $\ell\ln2$, removing the dependence on $d_{\sfR}$.
To see this, we can choose a state $\rho$ that is block-diagonal over $\sfI$ and attains the maximum in \Cref{eqn:MSV-group-form}. Then the constraint $\Tr_{\sfI}(\rho)=I_{\sfR}/d_{\sfR}$ implies $\rho\preceq I_{\sfI\sfR}/d_{\sfR}$, and thus $D(\rho\|\rho_1)\leq\ell\ln2$. Moreover, the optimality of $\rho$ implies $\Tr(A_t(\rho-\rho_t)) =\MSV{\ell}{Q}-\calG_{\varepsilon}(\rho_t,H_t)$. Applying \Cref{lemma:estimate-MSV-with-average-G-lower-bound} therefore yields the following corollary.

\begin{corollary}\label{corollary:estimate-MSV-with-average-G-lower-bound}
Let $W \coloneq 1 + 1/\varepsilon$ and suppose that $0 < \eta W \leq 1$.
For every integer $T \geq 1$, the states $\rho_1, \ldots, \rho_T$, and Hermitian contractions $H_1, \ldots, H_T$ generated by the above procedure satisfy
\[\MSV{\ell}{Q} - \frac{\ell \ln 2}{\eta T} - \eta W^2\leq \frac{1}{T}\sum_{t = 1}^T \calG_{\varepsilon}(\rho_t, H_t)\enspace.\] 
\end{corollary} 

The following lemma gives an upper bound on the average value obtained by the above procedure.

\begin{lemma}\label{lemma:estimate-MSV-with-average-G-upper-bound}
For every integer $T \geq 1$, the states $\rho_1, \ldots, \rho_T$, and Hermitian contractions $H_1, \ldots, H_T$ generated by the above procedure satisfy
\[\frac{1}{T}\sum_{t = 1}^T \calG_{\varepsilon}(\rho_t, H_t) \leq \MSV{\ell}{Q} + \varepsilon + \frac{\xi}{\varepsilon} \enspace.\]
\end{lemma}

\begin{proof}
Since $H_t$ is chosen to satisfy $\Tr(H_t(\Tr_{\sfI}(\rho_t) - I_{\sfR}/d_{\sfR})) \geq \norm{\Tr_{\sfI}(\rho_t) - I_{\sfR}/d_{\sfR}}_1 - \xi$, we obtain
\[\calG_{\varepsilon}(\rho_t, H_t) \leq  \calF_{\varepsilon}(\rho_t) + \frac{\xi}{\varepsilon}  \leq \MSV{\ell}{Q} + \varepsilon + \frac{\xi}{\varepsilon} \enspace,\] 
where we use \Cref{lemma:approximate-MSV-by-introducing-penalty-on-ell-1-norm} in the last inequality. Averaging over $t$ completes the proof.
\end{proof}

\subsubsection{Proof of \texorpdfstring{\Cref{thm:qcQAM[sqrt log n]=BQP}}{}}

\begin{proof}[Proof of \Cref{thm:qcQAM[sqrt log n]=BQP}]
Since the verifier may simply ignore the prover's message and decide a $\BQP$ problem directly, we immediately obtain $\BQP \subseteq \qcQAM\sbra*{\ell(n), c(n), s(n)}$. It therefore remains to establish the reverse inclusion $\qcQAM\sbra*{\ell(n), c(n), s(n)} \subseteq \BQP$. Notice that the problem $\SteerVal\sbra*{\ell(n),c(n),s(n)}$ is $\qcQAM\sbra*{\ell(n), c(n), s(n)}$-complete (see \Cref{lemma:SteerVal-qcQAM[l]-complete}). Hence it suffices to estimate $\MSV{\ell}{Q}$ to within additive error $\Delta/8$ in time $\poly(n,\ell,1/\Delta) \exp\rbra*{O\rbra[\big]{\ell^2/\Delta^6}}$.

Set parameters $\varepsilon \coloneq \frac{\Delta}{32}$, $B \coloneq \frac{1}{\varepsilon}$, $W \coloneq 1 + B$, $\xi \coloneq \frac{\Delta\varepsilon}{32}$, $\eta \coloneq \frac{\Delta}{32W^2}$, and $T \coloneq \lceil \frac{32 \ell \ln 2}{\eta \Delta}\rceil$. Then \Cref{corollary:estimate-MSV-with-average-G-lower-bound,lemma:estimate-MSV-with-average-G-upper-bound} give the following bound
\[\MSV{\ell}{Q} - \Delta/16 \leq \frac{1}{T}\sum_{t = 1}^T \calG_{\varepsilon}(\rho_t, H_t) \leq \MSV{\ell}{Q} + \Delta/16 \enspace.\]

It remains to estimate $\frac{1}{T}\sum_{t = 1}^T \calG_{\varepsilon}(\rho_t, H_t)$ within additive error $\Delta/16$ in time $\poly(n,\ell,1/\Delta) \exp\rbra*{O\rbra[\big]{\ell^2/\Delta^6}}$. To this end, we use the QSVT techniques without explicitly materializing the exponentially large matrices $\rho_t$ and $H_t$. We use the notations from \Cref{subsubsec:appoximate-MSV}. 

We show how to obtain block-encodings of the relevant operators.

\paragraph{Block-encoding of $C$} We may assume without loss of generality that $Q$ is controlled by the register $\sfI$, which stores $r$. (We can ensure this by first copying $r$ from $\sfI$ into a fresh register $\sfJ$ using $\CNOT$ gates and then running $Q$ with $\sfJ$ in place of $\sfI$.)

Then $Q = \sum_r \ketbra{r}{r}_{\sfI} \otimes Q_r$. Thus by direct calculation, the unitary $U_Q \coloneq Q^\dagger (I_{\sfO} - 2\ketbra{1}{1}_{\sfO}) Q$ is a $(1, a, 0)$-block-encoding of $I_{\sfI\sfR} - 2C$ where $a$ is the number of ancillary qubits used by $Q$:
\begin{subequations}
\begin{align*}
    &\rbra*{\bra{\bar{0}}_{\sfA} \otimes I_{\sfR\sfI}} U_Q \rbra*{\ket{\bar{0}}_{\sfA} \otimes I_{\sfR\sfI}}\\
    =& \rbra*{\bra{\bar{0}}_{\sfA} \otimes I_{\sfR\sfI}} Q^\dagger (I_{\sfO} - 2\ketbra{1}{1}_{\sfO}) Q \rbra*{\ket{\bar{0}}_{\sfA} \otimes I_{\sfR\sfI}}\\
    =& \sum_r \ketbra{r}{r}_{\sfI} \otimes \rbra*{\bra{\bar{0}}_{\sfA} \otimes I_{\sfR}} Q_r^\dagger (I_{\sfO} - 2\ketbra{1}{1}_{\sfO}) Q_r \rbra*{\ket{\bar{0}}_{\sfA} \otimes I_{\sfR}}\\
    =& \sum_r \ketbra{r}{r}_{\sfI} \otimes (I_{\sfR} - 2\Pi_r)\\
    =& I_{\sfI\sfR} - 2C\enspace.
\end{align*}
\end{subequations}

Applying \Cref{lemma:lcu} on the identity $I_{\sfI\sfR}$ and $I_{\sfI\sfR} - 2C$ gives a $(1, a^{(C)}, 0)$-block-encoding of $\frac{1}{2}(I_{\sfI\sfR} - (I_{\sfI\sfR} - 2C)) = C$, where $a^{(C)} \coloneq a + 1$. Let $S^{(C)}$ denote the gate complexity of this $(1, a^{(C)}, 0)$-block-encoding of $C$. Then by the efficiency of $Q$, $S^{(C)} = \poly(n)$.

\paragraph{Rewriting $\rho_t$} 
To identify the operators associated with $\rho_t$ that need to be block-encoded, we first rewrite $\rho_t$ in terms of an operator with norm at most one.

Define $\alpha_t \coloneq 2\eta W(t - 1)$. Then by the triangle inequality, $\alpha_t$ upper bounds $\norm{K_t}$: \[\norm{K_t} \leq \eta(t - 1) \norm{C} +  \eta\frac{1}{\varepsilon}\sum_{j < t}\norm{I_{\sfI} \otimes H_j} \leq \eta W(t - 1) \leq \alpha_t \enspace.\] Since scalar shifts of $K_t$ cancel upon normalization, we have
\[\rho_t = \frac{e^{K_t}}{\Tr\rbra*{e^{K_t}}} = \frac{E_t}{2^\ell \cdot z_t d_{\sfR}}\enspace,\]
where $E_t \coloneq e^{K_t - \alpha_t I_{\sfI\sfR}}$ is an operator with norm at most one and $z_t \coloneq \Tr(E_t)/\rbra*{2^\ell d_{\sfR}}$ denotes the normalized trace of $E_t$. 

Moreover, by construction, $K_t$ is block-diagonal over $\sfI$, and so is $\rho_t$. Thus 
\[\Tr_{\sfI}(\rho_t) - I_{\sfR}/d_{\sfR} = 2^\ell \rbra*{\bra{u}_{\sfI} \otimes I_{\sfR}} \rho_t  \rbra*{\ket{u}_{\sfI} \otimes I_{\sfR}} - I_{\sfR}/d_{\sfR} = \frac{D_t}{z_t d_{\sfR}}\enspace,\]
where $\ket{u}_{\sfI} \coloneq \frac{1}{2^{\ell/2}}\sum_{r \in \binset^\ell} \ket{r}_{\sfI}$ and $D_t \coloneq \rbra*{\bra{u}_{\sfI} \otimes I_{\sfR}} E_t \rbra*{\ket{u}_{\sfI} \otimes I_{\sfR}} - z_t I_{\sfR}$. 

For the subsequent error analysis, we need a uniform lower bound on $z_t$. Since $\norm{K_t}\leq\alpha_t$, we have $E_t\succeq e^{-2\alpha_t}I_{\sfI\sfR}$ and hence $z_t\geq e^{-2\alpha_t}$.
Defining $\mu\coloneq e^{-2\alpha_T}$ and using $\alpha_t\leq\alpha_T$ for every $t\in[T]$, we obtain
\[
    \mu\leq z_t\leq1
    \qquad\text{for every }t\in[T]\enspace.
\]

\paragraph{Block-encodings of $E_t$, $D_t$ and $H_t$}

Next we iteratively construct $(2, a^{(E)}_t, \delta)$-block-encoding of $E_t$ with $S^{(E)}_t$ gates, $(4, a^{(D)}_t, \delta')$-block-encoding of $D_t$ with $S^{(D)}_t$ gates, and $(2, a_{t}^{(H)}, 0)$-block-encoding of $H_t$ with $S^{(H)}_t$ gates via induction as follows, where $\delta \coloneq \frac{\mu^2 \xi}{32}$ and $\delta' \coloneq 3\delta$, and $a^{(E)}_t$, $a^{(D)}_t$, $a_{t}^{(H)}$, $S^{(E)}_t$, $S^{(D)}_t$, and $S_{t}^{(H)}$ are integers to be determined.

For the base case $t = 1$, we set $a^{(E)}_1 \coloneq 2$, $a^{(D)}_1 \coloneq 1$, and $a^{(H)}_{1} \coloneq 2$. By the initialization of $\rho_1$ and $K_1$, $E_1 = I_{\sfI\sfR}$, $D_1 = 0$, and $\Tr_{\sfI}(\rho_1) - I_{\sfR}/d_{\sfR} = 0$. Thus $H_1 \coloneq I_{\sfR}$ satisfies that
\[\Tr(H_1(\Tr_{\sfI}(\rho_1) - I_{\sfR}/d_{\sfR})) \geq \norm{\Tr_{\sfI}(\rho_1) - I_{\sfR}/d_{\sfR}}_1 - \xi\enspace.\]
Therefore, $E_1$ has a trivial $(2, a^{(E)}_1, \delta)$-block-encoding $H_{\sfA}(\ketbra{0}{0}_{\sfA} \otimes I_{\sfA'\sfI\sfR} + \ketbra{1}{1}_{\sfA} \otimes X_{\sfA'} \otimes I_{\sfI \sfR}) H_{\sfA}$, $H_1$ has a trivial $(2, a_{1}^{(H)}, 0)$-block-encoding $H_{\sfA}(\ketbra{0}{0}_{\sfA} \otimes I_{\sfA'\sfR} + \ketbra{1}{1}_{\sfA} \otimes X_{\sfA'} \otimes I_{\sfR}) H_{\sfA}$, and $D_1$ has a trivial $(4, a_1^{(D)}, \delta')$-block-encoding $X_{\sfA} \otimes I_{\sfR}$. Let $S^{(E)}_1$, $S^{(D)}_1$, and $S_{1}^{(H)}$ denote the number of gates in the above block-encodings of $E_1$, $D_1$, and $H_1$, respectively. Then $S^{(E)}_1 = O(1)$, $S^{(D)}_1 = O(1)$, and $S_{1}^{(H)} = O(1)$.

For the induction step $t$, suppose that we are given a $(2, a^{(E)}_j, \delta)$-block-encoding of $E_j$ with $S_j^{(E)}$ gates, a $(4, a^{(D)}_j, \delta')$-block-encoding of $D_j$  with $S_j^{(D)}$ gates, and a $(2, a_{j}^{(H)}, 0)$-block-encoding of $H_j$ with $S_j^{(H)}$ gates for $j < t$. We construct block-encodings of $E_t$, $D_t$, and $H_t$.

Set $\delta_1 \coloneq \delta/2$. Applying \Cref{lemma:exp-polynomial-approx} on $K_t/\alpha_t$ which has operator norm at most 1, there exists a polynomial $P_{\delta_1, \alpha_t}^{(\sf exp)}$ with degree $d^{(\sf exp)}_t = O(\sqrt{\alpha_t \ln (1/\delta_1)} + \ln(1/\delta_1))$ such that
\[\norm{P_{\delta_1, \alpha_t}^{(\sf exp)}(K_t/\alpha_t) - e^{\alpha_t(K_t/\alpha_t - I_{\sfI\sfR})}} \leq \delta_1\enspace.\]
In other words,
\begin{align}\label{eqn:approx-exp-Et}
    \norm{P_{\delta_1, \alpha_t}^{(\sf exp)}(K_t/\alpha_t) - E_t} \leq \delta_1\enspace.
\end{align}

We now construct a block-encoding of $K_t$, in order to apply \Cref{corr:qsvt-hermitian-contraction} to get a block-encoding of $P_{\delta_1, \alpha_t}^{(\sf exp)}(K_t/\alpha_t)/2$. Notice that $K_t$ can be written as a linear combination of operators for which we have the block-encodings. Specifically, $K_t = \eta(t - 1)C - \eta \frac{1}{\varepsilon}\sum_{j < t} I_{\sfI} \otimes H_j$. Thus by \Cref{lemma:lcu}, we can construct $(\alpha_t, a^{(K)}_t, 0)$-block-encoding of $K_t$ based on the $(2, a_{j}^{(H)}, 0)$-block-encoding of $H_j$ for $j < t$ and the $(1, a^{(C)}, 0)$-block-encoding of $C$, where $a^{(K)}_t \coloneq \lceil \log_2 t\rceil + \max (a^{(C)}, \max_{j < t}a_{j}^{(H)})$. This block-encoding of $K_t$ has $O\rbra*{\log^2 t\sum_{j = 1}^{t - 1}\rbra*{t + S^{(C)} + S^{(H)}_j}}$ gates.

Set $\delta_2\coloneq\delta/2$. Applying \Cref{corr:qsvt-hermitian-contraction} to the $(\alpha_t,a_t^{(K)},0)$-block-encoding of $K_t$ and the polynomial $P_{\delta_1,\alpha_t}^{(\sf exp)}$ gives a $(2,a_t^{(E)},0)$-block-encoding $U_{\widehat{E}_t}$ of a Hermitian contraction $\widehat E_t$, such that
\[
    \norm*{\widehat E_t- P_{\delta_1,\alpha_t}^{(\sf exp)}(K_t/\alpha_t)} \leq\delta_2\enspace,
\]
and $U_{\widehat{E}_t}$ uses $S^{(E)}_t \coloneq O\rbra*{d^{(\sf exp)}_t\rbra*{\log^2 t\sum_{j = 1}^{t - 1}\rbra*{t + S^{(C)} + S^{(H)}_j} + a_t^{(E)}}}$ gates, where $a_t^{(E)}\coloneq a_t^{(K)}+4$.
By \Cref{eqn:approx-exp-Et},
\[
    \norm*{\widehat E_t-E_t} \leq\delta_1+\delta_2 =\delta\enspace.
\]
Therefore, $U_{\widehat{E}_t}$ is also a $(2,a_t^{(E)},\delta)$-block-encoding of $E_t$.

To construct a block-encoding of $D_t = \rbra*{\bra{u}_{\sfI} \otimes I_{\sfR}} E_t \rbra*{\ket{u}_{\sfI} \otimes I_{\sfR}} - z_t I_{\sfR}$, we estimate the normalized trace $z_t$ with \Cref{corr:normalized-trace-estimation}. We apply \Cref{corr:normalized-trace-estimation} to $U_{\widehat{E}_t}$ with precision $\delta_3 \coloneq \delta$, repeat the procedure $L \coloneq 2\lceil4\ln(30T)\rceil+1$ times independently, and take the median. Denote the resulting estimate of $z_t$ by $\widehat{z}_t \in [0, 1]$, which has additive error at most $\delta_3 + \delta$ with probability at least $1 - \frac{1}{30T}$.\footnote{Since each
run succeeds with probability at least $8/\pi^2>3/4$,
Hoeffding's inequality bounds the probability that at least
half of the runs fail by $e^{-L/8}\leq1/(30T)$.
The median is within additive error $\delta_3 + \delta$ whenever a majority of the runs succeed.}

Define the approximation to $D_t$ by $\widehat{D}_t \coloneq \rbra*{\bra{u}_{\sfI} \otimes I_{\sfR}} \widehat{E}_t \rbra*{\ket{u}_{\sfI} \otimes I_{\sfR}} - \widehat{z}_t I_{\sfR}$. Then when the estimation of $z_t$ is successful, i.e. when $\abs{z_t - \widehat{z}_t} \leq \delta + \delta_3$,
\begin{align}\label{eqn:Dt-hatDt-closeness}
    \norm{D_t - \widehat{D}_t} \leq \norm{E_t - \widehat{E}_t} + \abs{z_t - \widehat{z}_t} \leq 2\delta + \delta_3 = \delta'\enspace.
\end{align}

The block-encoding of $\widehat{D}_t$ can be constructed via linear combination of unitaries. Notice that $\rbra*{\bra{u}_{\sfI} \otimes I_{\sfR}} \widehat{E}_t \rbra*{\ket{u}_{\sfI} \otimes I_{\sfR}}$ has a $(2, a_t^{(E)} + \ell, 0)$-block-encoding $H^{\otimes \ell}_{\sfI}U_{\widehat{E}_t}H^{\otimes \ell}_{\sfI}$, and $\widehat{z}_t I_{\sfR}$ has a trivial $(2, 2, 0)$-block-encoding since $\widehat{z}_t \leq 1$. By \Cref{lemma:lcu}, we can construct a $(4, a_t^{(E)} + \ell + 1, 0)$-block-encoding of $\widehat{D}_t$, denoted as $U_{\widehat{D}_t}$, such that $U_{\widehat{D}_t}$ uses $S^{(D)}_t \coloneq O(S_t^{(E)} + \ell)$ gates. By \Cref{eqn:Dt-hatDt-closeness}, $U_{\widehat{D}_t}$ is also a $(4, a_t^{(D)}, \delta')$-block-encoding of $D_t$,  where $a_t^{(D)} \coloneq a_t^{(E)} + \ell + 1$.

It remains to construct a block-encoding of $H_t$. Recall that
\begin{align*}
    &\Tr_{\sfI}(\rho_t) - I_{\sfR}/d_{\sfR}
    = \frac{D_t}{z_t d_{\sfR}}\enspace.
\end{align*}

Set $ \delta_4 \coloneq \delta$. By \Cref{lemma:sign-polynomial-approx}, there exists a polynomial $P^{(\sf sgn)}_{\delta_4, \delta_4}$ with degree $d^{(\sf sgn)}_t = O(\ln (1/\delta_4)/\delta_4)$ such that for every $x \in [-2, 2]\setminus(-\delta_4, \delta_4)$, $\abs{P^{(\sf sgn)}_{\delta_4, \delta_4}(x) x - \abs{x}} = \abs{x} \cdot \abs{P^{(\sf sgn)}_{\delta_4, \delta_4}(x) - \sign{x}} \leq 2 \cdot \delta_4$, and for every $x \in (-\delta_4, \delta_4)$, $\abs{P^{(\sf sgn)}_{\delta_4, \delta_4}(x) x - \abs{x}} = \abs{x} \cdot \abs{P^{(\sf sgn)}_{\delta_4, \delta_4}(x) - \sign{x}} \leq \delta_4 \cdot 2$. Thus for the operator $\widehat{D}_t$,
\begin{align}\label{eqn:sign-approximate}
    \norm{P^{(\sf sgn)}_{\delta_4, \delta_4}(\widehat{D}_t/4) \widehat{D}_t/4 - \abs{\widehat{D}_t/4}} \leq 2\delta_4\enspace.
\end{align}

Fix a precision parameter $\delta_5 \coloneq \delta$. Applying \Cref{corr:qsvt-hermitian-contraction} on polynomial $P^{(\sf sgn)}_{\delta_4, \delta_4}$,  the Hermitian operator $\widehat{D}_t$ and its $(4, a_t^{(D)}, 0)$-block-encoding, there exists a Hermitian contraction $H_t$ with a $(2, a_t^{(D)} + 4, 0)$-block-encoding, denoted by $U_{H_t}$, such that
\begin{align}\label{eqn:Psgn-approx-Ht}
    \norm{H_t - P^{(\sf sgn)}_{\delta_4, \delta_4}(\widehat{D}_t/4)} \leq \delta_5\enspace.
\end{align}
and $U_{H_t}$ uses $S^{(H)}_t \coloneq O\rbra*{d^{(\sf sgn)}_t \rbra*{S^{(D)}_t + a_t^{(D)} }}$ gates.
We set $a_t^{(H)} \coloneq a_t^{(D)} + 4$. Then $U_{H_t}$ is a $(2, a_t^{(H)}, 0)$-block-encoding of $H_t$. 

It remains to show that $H_t$ satisfies that 
\[\Tr(H_t(\Tr_{\sfI}(\rho_t) - I_{\sfR}/d_{\sfR})) \geq \norm{\Tr_{\sfI}(\rho_t) - I_{\sfR}/d_{\sfR}}_1 - \xi\enspace.\]

This can be established by the triangle inequality.
\begin{align*}
    &\abs{\Tr(H_t \widehat{D}_t) - \norm{\widehat{D}_t}_1}\\
    \leq& d_{\sfR} \norm{H_t \widehat{D}_t - \abs{\widehat{D}_t}}\\
    \leq& d_{\sfR} \norm{P^{(\sf sgn)}_{\delta_4, \delta_4}(\widehat{D}_t/4)\widehat{D}_t - \abs{\widehat{D}_t}} + 4d_{\sfR} \delta_5 \tag{By \Cref{eqn:Psgn-approx-Ht} and $\norm{\widehat{D}_t} \leq 4$} \\
    \leq& 8d_{\sfR} \delta_4 + 4d_{\sfR} \delta_5\enspace. \tag{By \Cref{eqn:sign-approximate}}
\end{align*}
Then as $\widehat{D}_t$ is an approximation of $D_t$,
\begin{align*}
&\abs{\Tr(H_t D_t) - \norm{D_t}_1}\\
\leq & \abs{\Tr(H_t \widehat{D}_t) - \norm{\widehat{D}_t}_1} + \abs{\norm{D_t}_1 - \norm{\widehat{D}_t}_1} + \abs{\Tr(H_t (\widehat{D}_t - D_t))}\\
\leq & 8d_{\sfR}\delta_4 + 4d_{\sfR} \delta_5 + 2d_{\sfR} \norm{\widehat{D}_t - D_t} \tag{By $\norm{H_t} \leq 1$}\\
\leq & 8d_{\sfR}\delta_4 + 4d_{\sfR} \delta_5 + 2d_{\sfR}\delta'\enspace, \tag{By \Cref{eqn:Dt-hatDt-closeness}}
\end{align*}
as long as the estimation of $z_t$ is successful, which implies 
\[\Tr(H_t(\Tr_{\sfI}(\rho_t) - I_{\sfR}/d_{\sfR})) \geq \norm{\Tr_{\sfI}(\rho_t) - I_{\sfR}/d_{\sfR}}_1 - 18\delta/z_t \geq \norm{\Tr_{\sfI}(\rho_t) - I_{\sfR}/d_{\sfR}}_1 - \xi \enspace,\]
as long as the estimation of $z_t$ is successful, where we use that $z_t \geq \mu$, $\delta \leq \frac{\mu\xi}{32}$, and $\delta' = 3\delta$. 

\paragraph{Computing $\calG_{\varepsilon}(\rho_t, H_t)$}
With all the block-encodings of $H_t$, $D_t$, and $E_t$ for $t \in [T]$, we now estimate $\frac{1}{T}\sum_{t \in [T]}\calG_{\varepsilon}(\rho_t, H_t)$.

We write $v_t \coloneq \calG_{\varepsilon}(\rho_t, H_t)$ as follows.
\begin{align*}
v_t &= \Tr(C\rho_t) - \frac{1}{\varepsilon} \Tr\rbra*{H_t\rbra*{\Tr_{\sfI}(\rho_t)-I_{\sfR}/d_{\sfR}}}\\
&= \frac{\Tr(CE_t)}{2^\ell d_{\sfR}} \cdot \frac{1}{z_t} - \frac{1}{\varepsilon} \frac{\Tr\rbra*{H_tD_t}}{d_{\sfR}} \cdot \frac{1}{z_t}\enspace,
\end{align*}
where we have $(1, a^{(C)}, 0)$-block-encoding of $C$, $(2, a^{(E)}_t, \delta)$-block-encoding of $E_t$, $(4, a^{(D)}_t, \delta')$-block-encoding of $D_t$, and $(2, a_{t}^{(H)}, 0)$-block-encoding of $H_t$.

Let $a_t \coloneq \frac{\Tr(CE_t)}{2^\ell d_{\sfR}}$ and $b_t \coloneq \frac{\Tr\rbra*{H_tD_t}}{d_{\sfR}}$. Then we have $(2, a^{(C)} + a^{(E)}_t, \delta)$-block-encoding of $CE_t$, and $(8, a^{(D)}_t + a^{(H)}_t, \delta')$-block-encoding of $H_tD_t$. Since $C$, $E_t$, $D_t$, and $H_t$ are all Hermitian, $\Tr(CE_t) \in \bbR$ and $\Tr(D_tH_t) \in \bbR$, and thus $a_t$ and $b_t$ can be estimated using the normalized trace estimation: we apply \Cref{corr:normalized-trace-estimation} to $CE_t$ and $H_tD_t$ and their block-encodings with precision $\delta_6 \coloneq \delta$, repeat the procedure $L$ times independently, and take the median. Denote the resulting estimates by $\widehat{a}_t$ and $\widehat{b}_t$. Then $\abs{\widehat{a}_t - a_t} \leq \delta_6 + \delta$ with probability at least $1 - \frac{1}{30T}$, and $\abs{\widehat{b}_t - b_t} \leq \delta_6 + \delta'$ with probability at least $1 - \frac{1}{30T}$.

Conditioned on all the estimations of $a_t$, $b_t$, and $z_t$ being successful for $t \in [T]$, which happens with probability at least $1 - \frac{1}{30T} \cdot 3T = \frac{9}{10}$, our estimation of $v_t$, defined as \[\widehat{v}_t \coloneq \frac{\widehat{a}_t}{\widehat{z}_t'} - \frac{1}{\varepsilon}\frac{\widehat{b}_t}{\widehat{z}_t'}\enspace,\]
where $\widehat{z}_t' \coloneq \min (1, \max(\mu, \widehat{z}_t))$, satisfies that
\begin{align*}
    &\abs{\widehat{v}_t - v_t}\\
    =& \abs{\rbra*{\frac{\widehat{a}_t}{\widehat{z}_t'} - \frac{a_t}{z_t}  }- \frac{1}{\varepsilon}\rbra*{\frac{\widehat{b}_t}{\widehat{z}_t'} - \frac{{b}_t}{{z}_t}}}\\
    \leq& \abs{\frac{\widehat{a}_t}{\widehat{z}_t'} - \frac{a_t}{z_t}} +  \frac{1}{\varepsilon}\abs{\frac{\widehat{b}_t}{\widehat{z}_t'} - \frac{{b}_t}{{z}_t}} \tag{By the triangle inequality}\\
    \leq& \frac{\abs{\widehat{a}_t - a_t}}{\widehat{z}_t'} + \frac{\abs{a_t}\abs{\widehat{z}_t' - z_t}}{z_t \widehat{z}_t'} +  \frac{1}{\varepsilon}\frac{\abs{\widehat{b}_t - b_t}}{\widehat{z}_t'} +  \frac{1}{\varepsilon}\frac{\abs{b_t}\abs{\widehat{z}_t' - z_t}}{z_t \widehat{z}_t'}  \tag{By the triangle inequality}\\
    \leq& \frac{\delta_6 + \delta}{\mu} + \frac{2\delta}{\mu^2} +  \frac{1}{\varepsilon}\rbra*{\frac{\delta_6 + \delta'}{\mu} + \frac{2\delta}{\mu^2}}\\
    \leq & \frac{\Delta}{16}\enspace,
\end{align*}
where we plug in $\delta = \frac{\mu^2 \xi}{32}$ and $\xi = \frac{\Delta\varepsilon}{32}$.

That is to say, $\frac{1}{T}\sum_{t \in [T]}\widehat{v}_t$ is an estimation of $\frac{1}{T}\sum_{t = 1}^T \calG_{\varepsilon}(\rho_t, H_t)$ within additive error $\frac{\Delta}{16}$ with probability at least $\frac{9}{10}$, and thus $\frac{1}{T}\sum_{t \in [T]}\widehat{v}_t$ is an estimation of $\MSV{\ell}{Q}$ within additive error $\frac{\Delta}{8}$ with probability at least $\frac{9}{10}$. Therefore, we can solve $\SteerVal\sbra*{\ell(n),c(n),s(n)}$ by comparing $\frac{1}{T}\sum_{t \in [T]}\widehat{v}_t$ with the threshold $\frac{c(n) + s(n)}{2}$.

In particular, the following algorithm solves $\SteerVal\sbra*{\ell(n),c(n),s(n)}$. The correctness of \algoref{algo:MSV-estimation} follows from the above arguments. In the  remainder of the proof, we analyze the efficiency of \algoref{algo:MSV-estimation}.

\begingroup
\LinesNumbered
\begin{algorithm}[!ht]
    \SetAlgorithmName{Algorithm}{algorithm}{List of Algorithms}
    \caption{Quantum algorithm for $\SteerVal[\ell,c,s]$.}
    \label{algo:MSV-estimation}
    \SetEndCharOfAlgoLine{.}
    \setlength{\parskip}{5pt}
    \SetKwFor{While}{}{:}{}
    \SetKwFor{For}{For}{:}{}
    \SetKwIF{If}{ElseIf}{Else}{If}{:}{elif}{Else:}{}
    \SetKwComment{Comment}{// }{}
    \SetKwInOut{Input}{Input}
    \SetKwInOut{Output}{Output}

    \Input{A quantum circuit $Q$ defining a
    $\SteerVal[\ell,c,s]$ instance.}
    \Output{Accept if $\MSV{\ell}{Q}\geq c$ and reject if
    $\MSV{\ell}{Q}\leq s$ with probability at least $9/10$.}

    \medskip
    {\color{gray}\Comment{Initialization.}}

    Set $\Delta\coloneq c-s$ and choose parameters $\varepsilon, B, W, \xi, \eta, T, \mu, \delta$, and $L$ as above.

    Construct $U_C$ from $Q$. Initialize the trivial
    block-encodings of $\widehat E_1=I_{\sfI\sfR}$,
    $\widehat D_1=0$, and $H_1=I_{\sfR}$,
    and set $\widehat z_1=1$.

    \medskip
    {\color{gray}\Comment{Constructing the block-encodings.}}

    \For{$t=2,\ldots,T$}{
        Construct $U_{\widehat E_t}$ from $U_C$ and
        $\{U_{H_j}\}_{j<t}$ as described above\;

        Estimate $z_t=\Tr(E_t)/(2^\ell d_{\sfR})$
        using $U_{\widehat E_t}$ by taking the median of $L$ independent estimates produced by \Cref{corr:normalized-trace-estimation} and denote the resulting
        estimate as $\widehat z_t$\;

        Construct $U_{\widehat D_t}$ from
        $U_{\widehat E_t}$ and $\widehat z_t$,
        and then construct $U_{H_t}$
        as described above\;
    }

    \medskip
    {\color{gray}\Comment{Estimating the payoffs.}}

    For each $t\in[T]$, use product block-encodings and
    normalized trace estimation to obtain estimates
    $\widehat a_t$ and $\widehat b_t$ of $a_t=\frac{\Tr(CE_t)}{2^\ell d_{\sfR}}$ and $b_t=\frac{\Tr(H_tD_t)}{d_{\sfR}}$ by taking the median of $L$ independent estimates produced by \Cref{corr:normalized-trace-estimation}.

    For each $t\in[T]$, set $\widehat z_t' =\min\{1,\max\{\mu,\widehat z_t\}\}$ and $\widehat v_t=\frac{\widehat{a}_t}{\widehat{z}_t'} - \frac{1}{\varepsilon}\frac{\widehat{b}_t}{\widehat{z}_t'}$.

    \medskip
    {\color{gray}\Comment{Putting everything together.}}

    Accept if $\frac1T\sum_{t=1}^T\widehat v_t\geq(c+s)/2$,
    and reject otherwise.
\end{algorithm}
\endgroup

\paragraph{Efficiency analysis.}

Our parameter choices imply $\varepsilon = O(\Delta)$, $B = O(\frac{1}{\Delta})$, $W = O(\frac{1}{\Delta})$, $\xi = O(\Delta^2)$, $\eta = O(\Delta^3)$, and $T = O(\ell/\Delta^4)$. Moreover, $\alpha_T = O(T \Delta^2) = O(\ell /\Delta^2)$, $\mu = O(e^{-2\alpha_T})$, $\delta = O(\mu^2)$, and thus $\log(1/\delta) = O(\alpha_T) = O(\ell /\Delta^2)$.

We analyze the number of ancillary qubits of the block-encodings of $C$, $E_t$, $D_t$, and $H_t$. By the efficiency of $Q$, $a^{(C)} = \poly(n)$. By our inductive construction, $a^{(E)}_1 = 2$, $a^{(D)}_1 = 1$, and $a^{(H)}_{1} = 2$, and for every $t \in \{2, 3, \cdots, T\}$, $a_t^{(E)} = \lceil \log_2 t\rceil + \max (a^{(C)}, \max_{j < t}a_{j}^{(H)}) + 4$, $a_t^{(D)} = a_t^{(E)} + \ell + 1$, and $a_t^{(H)} = a_t^{(D)} + 4$. Solving these recurrences gives $a_T^{(E)} = O(a^{(C)} + T(\ell + \log_2 T))$. Therefore, for every $t \in [T]$, $a^{(E)}_t$, $a^{(D)}_t$, and $a^{(H)}_t$ are all bounded by $O(a^{(C)} + \ell(\ell - \log_2 \Delta)/\Delta^4)$.

Next we analyze the gate complexity of the block-encodings of $C$, $E_t$, $D_t$, and $H_t$. By the efficiency of $Q$, $S^{(C)}=\poly(n)$. By our inductive construction, $S^{(E)}_1 = O(1)$, $S^{(D)}_1 = O(1)$, and $S_{1}^{(H)} = O(1)$, and for every $t \in \{2, 3, \cdots, T\}$, $S^{(E)}_t= O\rbra*{d^{(\sf exp)}_t\rbra*{\log^2 t\sum_{j = 1}^{t - 1}\rbra*{t + S^{(C)} + S^{(H)}_j} + a_t^{(E)}}}$, $S^{(D)}_t = O(S_t^{(E)} + \ell)$, and $S^{(H)}_t = O\rbra*{d^{(\sf sgn)}_t \rbra*{S^{(D)}_t + a_t^{(D)}}}$. Plugging in the bounds for $a_t^{(E)}$, $a_t^{(D)}$, $S^{(C)}$, $d^{(\sf exp)}_t = O(\ell/ \Delta^2)$ and $d^{(\sf sgn)}_t = \exp\rbra*{O(\ell/\Delta^2)}$, we obtain $S^{(E)}_t = O\rbra*{\poly(\ell, 1/\Delta) \rbra*{\sum_{j = 1}^{t - 1}S^{(H)}_j + \poly(n, \ell, 1/\Delta)}}$, $S^{(D)}_t = O(S_t^{(E)} + \ell)$, and $S^{(H)}_t = O\rbra*{\exp\rbra*{O(\ell/ \Delta^2)} \rbra*{S^{(D)}_t + \poly(n, \ell, 1/\Delta)}}$. 
Solving these recurrences gives $S_T^{(E)} = \poly(n,\ell, 1/\Delta)\exp(O(\ell T/\Delta^2)) = \poly(n,\ell, 1/\Delta)\exp(O(\ell^2/\Delta^6))$. Therefore, for every $t \in [T]$, $S^{(E)}_t$, $S^{(D)}_t$, and $S^{(H)}_t$ are all bounded by $\poly(n,\ell, 1/\Delta)\exp(O(\ell^2 /\Delta^6))$.

There are at most $3T$ normalized trace estimations, each requiring $O(L/\delta)$ oracle calls after success amplification, since the normalization factors are constants. Including the cost of the oracle circuits and the additional gates, the total gate complexity is
\[
    O(TL/\delta)\cdot\poly(n,\ell, 1/\Delta)\exp(O(\ell^2 /\Delta^6))
    =\poly(n,\ell, 1/\Delta)\exp(O(\ell^2 /\Delta^6))\enspace,
\]
where we use $T=O(\ell/\Delta^4)$, $L=O(\log T)$, and $\delta^{-1}=\exp(O(\ell/\Delta^2))$.

The required classical computations, including the generation of the circuit descriptions, can also be performed in time $\poly(n,\ell, 1/\Delta)\exp(O(\ell^2 /\Delta^6))$ by the efficient constructions above. This establishes the efficiency of \algoref{algo:MSV-estimation} and, together with the correctness analysis, completes the proof.
\end{proof}

\section*{Acknowledgments}
\noindent
The authors thank Alessandro Chiesa and Thomas Vidick for insightful discussions, and Alessandro Chiesa in particular for suggesting interactive proof systems with a laconic prover in the quantum setting as a direction to consider. The authors are grateful to Prashant Nalini Vasudevan and Yunqi Li for pointing out a black-box barrier to polarizing the total variation distance beyond $a^2>b$ in~\cite{HR05}, whose underlying model captures all previously known approaches.

\sloppypar
ZH was supported in part by the Ethereum Foundation and the Global Chinese Community of Universal Digital Commons.
YL was supported in part by funding from the Swiss State Secretariat for Education, Research and Innovation (SERI). 
This work was also supported in part by a grant of access to OpenAI models through the ChatGPT for Academic Researchers program.

\section*{AI use disclosure}
\noindent
Large language model tools were used interactively at all stages of the preparation of this manuscript: from the earliest exploratory research stages to the final exposition and writing steps. The final manuscript was substantially revised by the authors, who remain solely responsible for all mathematical claims, proofs, references, and conclusions.

\bibliographystyle{alphaurlQ}
\bibliography{Qlaconic}

\appendix
\crefalias{section}{appendix}
\crefalias{subsection}{appendix}
\crefalias{subsubsection}{appendix}

\section{Polarizing the total variation distance in the natural regime}
\label{sec:SD-in-SZK}

The \textsc{Statistical Difference Problem} ($\SD[a,b]$) is the classical counterpart of \Cref{def:QSD}, with efficiently samplable distributions generated by (not necessarily reversible) Boolean circuits on uniformly random input bits and total variation distance as the closeness measure.

\SDinSZKnaturalRegime*

Combining \Cref{thm:SD-in-SZK-natural-regime} with~\cite[Lemma 4.1]{GVW02}, we obtain the following inclusion, which improves~\cite[Theorem 3.1]{GVW02}:
\begin{theorem}[$\IPbit\subseteq\SZK$ in the natural regime]
\label{thm:IPbit-in-SZK}
Let $c(n)$ and $s(n)$ be efficiently computable functions such that $0 \leq s(n) < c(n) \leq 1$. For all sufficiently large $n$, we have:
\[ \text{If } c(n)-s(n) \geq 1/O(\log{n}), \quad \IPbit[c,s] \subseteq \SZK. \]
\end{theorem}

To establish \Cref{thm:SD-in-SZK-natural-regime}, as in \Cref{thm:QSD-in-QSZK-natural-regime}, we first bound
the approximation error by a scalar error and then encode the approximation as an instance of the \textsc{Entropy Difference Problem} (\ED{}).
For distributions $D_0$ and $D_1$ over the same finite set, define
\[
    \JS_2(D_0,D_1)
    \coloneqq\H_2\rbra*{\frac{D_0+D_1}{2}}
        -\frac{\H_2(D_0)+\H_2(D_1)}2\enspace.
\]
We use the classical counterpart of \Cref{eq:state-interpolate}, with $D_+ \coloneqq (D_0+D_1)/2$ and
\begin{equation}
    \label{eq:classical-smoothing}
    \forall z\in\binset, \quad D_{z,\lambda} \coloneqq \frac{1+\lambda}{2} D_z + \frac{1-\lambda}{2} D_{1-z},
    \quad \text{where } \lambda\in\sbra*{0,\frac{1}{2}}\enspace.
\end{equation}

\begin{lemma}[Scalar error bound for total variation distance approximation]
\label{lem:classical-transfer}
Let $D_0$ and $D_1$ be probability distributions over the same finite set. For every positive integer $J$, real coefficients $c_j$, and $\lambda_j\in[0,1/2]$ for $j\in[J]$, we have
\[
    \abs*{\sum_{j=1}^J c_j\JS_2(D_{0,\lambda_j},D_{1,\lambda_j}) -\TV(D_0,D_1)}
    \leq\sup_{|t|\leq1} \abs*{\sum_{j=1}^J c_j\Phi(\lambda_j t)-|t|}\enspace.
\]
Here, the function $\Phi(x)\coloneqq 1-\H_2\rbra*{(1+x)/2}$ for $\abs{x}<1$. 
\end{lemma}

\begin{proof}
We express both quantities $\JS_2(D_{0,\lambda},D_{1,\lambda})$ and $\TV(D_0,D_1)$ as weighted sums with weights $D_+(x)$. Restrict all sums over $x$ to $D_+(x)>0$, and set $\theta(x)\coloneqq\frac{D_0(x)-D_1(x)}{D_0(x)+D_1(x)}\in[-1,1]$. 
Then, $D_{z,\lambda}(x)=D_+(x)(1+(-1)^z\lambda \theta(x))$. A direct calculation gives
\begin{align*}
    \JS_2(D_{0,\lambda},D_{1,\lambda})
    &=\frac12\sum_x\sum_{z\in\binset} D_{z,\lambda}(x)\log\frac{D_{z,\lambda}(x)}{D_+(x)}\\
    &=\sum_xD_+(x) \frac{(1+\lambda \theta(x))\log(1+\lambda \theta(x)) +(1-\lambda \theta(x))\log(1-\lambda \theta(x))}2\\
    &=\sum_xD_+(x)\Phi(\lambda \theta(x))\enspace,\\
    \TV(D_0,D_1)
    &=\frac12\sum_x\abs{D_0(x)-D_1(x)}
      =\sum_xD_+(x)\abs{\theta(x)}\enspace.
\end{align*}

Consequently, it follows that
\begin{align*}
    \abs*{\sum_{j=1}^Jc_j\JS_2(D_{0,\lambda_j},D_{1,\lambda_j})-\TV(D_0,D_1)}
    &\leq\sum_xD_+(x)\abs*{ \sum_{j=1}^J c_j\Phi(\lambda_j \theta(x)) -\abs{\theta(x)}}\\
    &\leq\sup_{|t|\leq1}\abs*{ \sum_{j=1}^J c_j\Phi(\lambda_j t) -|t|}\enspace.
\end{align*}
Here, the first line follows from the triangle inequality, and the last line uses $\sum_x D_+(x)=1$.
\end{proof}

\paragraph{Implementing the \ED{} instance.}
The \textsc{Entropy Difference Problem} ($\ED[g]$) is the classical counterpart of $\QED[g]$. Given Boolean circuits $(C_0,C_1)$ of total description length $n$, let $D_z$ be the probability distribution generated by $C_z$ when its input is chosen uniformly at random for each $z\in\binset$. The promise underlying $\ED[g]$ is the following:
\begin{itemize}
    \item \emph{Yes:} $\H_2(D_0)-\H_2(D_1)\geq g(n)$;
    \item \emph{No:} $\H_2(D_0)-\H_2(D_1)\leq-g(n)$.
\end{itemize}

Following~\cite[Theorem~1.4]{GoldreichVadhan99} and~\cite[Theorem~2]{GoldreichSahaiVadhan98}, $\ED[1]$ is in \SZK{}, and the \SZK{} containment extends directly to $g(n)\geq 1/\poly(n)$, as stated in~\cite[Theorem~3.16]{BDRV19}:\footnote{\cite[Theorem~3.16]{BDRV19} uses the common Boolean circuit output length as the parameter $n$. Appending deterministic zeros to both outputs until their lengths equal the original total description length $n$ preserves $H_2(D_0)$ and $H_2(D_1)$, and hence their entropy difference.}
\begin{lemma}[\ED{} is in \SZK{}, adapted from~{\cite[Theorem~3.16]{BDRV19}}]
    \label{lemma:ED-in-SZK}
    For every efficiently computable function $g\colon\Naturals\to\mathbb{R}_{>0}$ such that $g(n)\geq1/\poly(n)$, we have $\ED[g]\in\SZK$.
\end{lemma}

\begin{proof}[Proof of \Cref{thm:SD-in-SZK-natural-regime}]
By the $\SZK$ containment in \Cref{lemma:ED-in-SZK}, it suffices to give a Karp reduction from $\SD[a,b]$ with $\Delta(n)\coloneqq a(n)-b(n)\geq 1/O(\log{n})$ to $\ED[1/q]$ for some polynomial $q$.
Let $(C_0,C_1)$ be an $\SD[a,b]$ instance of total description length $n$, where $C_z\colon\binset^{s_z}\to\binset^r$ samples $D_z$ on uniformly random input bits for each $z\in\binset$.
Choose dyadic numbers $\varepsilon=2^{-t}$ and $\tau(n)\in 2^{-t}\mathbb{Z}\cap[0,1]$ such that $6\varepsilon<\Delta\leq16\varepsilon$ and $\abs*{\tau(n)-(a(n)+b(n))/2}\leq\varepsilon$. 
As in the proof of \Cref{thm:QSD-in-QSZK-natural-regime}, these parameters can be computed using $O(1+\log(1/\Delta))$ bits of precision.

\parheading{Constructing the \ED{} instance.}
Apply \Cref{lemma:td-efficient-signed-approx-dyadic} with accuracy $\varepsilon$ to obtain $J$ and dyadic numbers $c_j$ and $\lambda_j$ for $j\in[J]$. We use the underlying construction, where $J$ is a power of two satisfying $J\geq32\beta/\varepsilon$, $\lambda_j=j/(16J)$, and $2J^2c_j\in\mathbb{Z}$.
Using the fixed constant $\beta>1$ in \Cref{lemma:abs-polynomial-approx}, set $\Upsilon\coloneqq2^{1+\ceil*{496\beta/\varepsilon}}$.
The coefficient bound gives $\sum_{j=1}^J\abs{c_j}+\tau(n) \leq 2^{\ceil*{496\beta/\varepsilon}}+2\leq \Upsilon$.
Combining \Cref{lem:classical-transfer,lemma:td-efficient-signed-approx-dyadic} gives the error bound
\begin{equation}
    \label{eq:classical-approximation}
    \abs*{\sum_{j=1}^J c_j \JS_2(D_{0,\lambda_j},D_{1,\lambda_j})-\TV(D_0,D_1)}\leq\varepsilon \enspace.
\end{equation}

We implement dyadic convex combinations using the classical version of the interval selection in \circuitref{circuit:dyadic-convex-combi-states}. For $K$ component circuits and weights with common denominator $2^L$, partition $\{0,\ldots,2^L-1\}$ into consecutive intervals whose lengths are $2^L$ times the weights. A uniform $L$-bit integer selects an interval, and a Boolean multiplexer outputs the corresponding circuit's sample, with or without the interval label. All component circuits are evaluated on independent uniform inputs, independently of the selector. For components with $r+1$ output bits, the construction uses $O\rbra*{K(L+r+\log K)}$ additional Boolean gates.

Let $U_1$ denote the uniform distribution on $\binset$. For every $u\in\binset^k$, let $\delta_u$ denote the point-mass distribution on $\binset^k$ defined by $\delta_u(x)\coloneqq \begin{cases} 1,&x=u\\ 0,&x\neq u\end{cases}$ for each $x\in\binset^k$. 
The coefficients in \Cref{table:classical-polarization-newC0,table:classical-polarization-newC1} are nonnegative and sum to one in each table. Their denominators divide $2\Upsilon J^2$, since $2^t$ divides $2J^2$ and $\tau(n)\in2^{-t}\mathbb{Z}$.
We obtain Boolean circuits $C'_0$ and $C'_1$, sampling $D'_0$ and $D'_1$, respectively, by applying the dyadic convex combination construction to the circuit-coefficient pairs in \Cref{table:classical-polarization-newC0,table:classical-polarization-newC1}, with $J+2$ terms and $L=\log(2\Upsilon J^2)$. Both circuits output the selected sample together with its label, with the terms ordered as $j=1,\ldots,J$, followed by the threshold and remainder terms. The labels utilize $\ceil*{\log(J+2)}$ bits.

\begin{table}[ht]
    \centering
    \begin{tabular}{@{}c c c c@{}}
        \toprule
        Term & Circuit & Distribution & Coefficient \\
        \midrule
        $j\in[J]$ with $c_j\geq0$
        & $C_{\geq0}$
        & $D_{\geq0}\coloneqq U_1\otimes D_+$
        & $\pi_j\coloneqq c_j/\Upsilon$ \\
        \midrule
        $j\in[J]$ with $c_j<0$
        & $C_{<0,j}$
        & $D_{<0,j}\coloneqq\frac12\sum_{z\in\binset}\delta_z\otimes D_{z,\lambda_j}$
        & $\pi_j\coloneqq-c_j/\Upsilon$ \\
        \midrule
        Threshold offset
        & $C_{0^{r+1}}$
        & $\delta_{0^{r+1}}$
        & $\pi_{J+1}\coloneqq\tau(n)/\Upsilon$ \\
        \midrule
        Remainder
        & $C_{0^{r+1}}$
        & $\delta_{0^{r+1}}$
        & $\pi_{J+2}\coloneqq1-\rbra[\big]{\sum_{j=1}^J\abs{c_j}+\tau(n)}/\Upsilon$ \\
        \bottomrule
    \end{tabular}
    \caption{Ingredients for constructing $C'_0$.}
    \label{table:classical-polarization-newC0}
\end{table}

\begin{table}[ht]
    \centering
    \begin{tabular}{@{}c c c c@{}}
        \toprule
        Term & Circuit & Distribution & Coefficient \\
        \midrule
        $j\in[J]$ with $c_j\geq0$
        & $C_{<0,j}$
        & $D_{<0,j}=\frac12\sum_{z\in\binset}\delta_z\otimes D_{z,\lambda_j}$
        & $\pi_j\coloneqq c_j/\Upsilon$ \\
        \midrule
        $j\in[J]$ with $c_j<0$
        & $C_{\geq0}$
        & $D_{\geq0}=U_1\otimes D_+$
        & $\pi_j\coloneqq-c_j/\Upsilon$ \\
        \midrule
        Threshold offset
        & $C_{U_1}$
        & $U_1\otimes\delta_{0^r}$
        & $\pi_{J+1}\coloneqq\tau(n)/\Upsilon$ \\
        \midrule
        Remainder
        & $C_{0^{r+1}}$
        & $\delta_{0^{r+1}}$
        & $\pi_{J+2}\coloneqq1-\rbra[\big]{\sum_{j=1}^J\abs{c_j}+\tau(n)}/\Upsilon$ \\
        \bottomrule
    \end{tabular}
    \caption{Ingredients for constructing $C'_1$.}
    \label{table:classical-polarization-newC1}
\end{table}

The resulting distributions $D'_0$ and $D'_1$ are given by
\begin{subequations}
\label{eq:classical-implementation}
\begin{align}
    D'_0 &=\sum_{j\in[J]\colon c_j\geq0}\pi_j\delta_{j-1}\otimes D_{\geq0} +\sum_{j\in[J]\colon c_j<0}\pi_j\delta_{j-1}\otimes D_{<0,j}\notag\\*
    &\qquad+\pi_{J+1}\delta_J\otimes\delta_{0^{r+1}} +\pi_{J+2}\delta_{J+1}\otimes\delta_{0^{r+1}}\enspace,\\
    D'_1 &=\sum_{j\in[J]\colon c_j\geq0}\pi_j\delta_{j-1}\otimes D_{<0,j} +\sum_{j\in[J]\colon c_j<0}\pi_j\delta_{j-1}\otimes D_{\geq0}\notag\\*
    &\qquad+\pi_{J+1}\delta_J\otimes U_1\otimes\delta_{0^r} +\pi_{J+2}\delta_{J+1}\otimes\delta_{0^{r+1}}\enspace.
\end{align}
\end{subequations}
Here, $\delta_{j-1}$, $\delta_J$, and $\delta_{J+1}$ refer to the binary encodings of the interval labels.
Next, we construct the circuit ingredients in \Cref{table:classical-polarization-newC0,table:classical-polarization-newC1}:\footnote{$C_{0^{r+1}}$ outputs $r+1$ zero bits, while $C_{U_1}$ outputs one fair bit followed by $r$ zero bits.}
\begin{itemize}
    \item $C_{\geq0}$: Applying the dyadic convex combination construction to the circuit-coefficient pairs $(0\Vert C_0,1/4)$, $(0\Vert C_1,1/4)$, $(1\Vert C_0,1/4)$, and $(1\Vert C_1,1/4)$, with four terms and $L=2$, gives a circuit sampling $D_{\geq0}$. Here, $b\Vert C_z$ denotes the circuit prepending the fixed bit $b$ to the output of $C_z$, so its output distribution is $\delta_b\otimes D_z$.
    \item $C_{<0,j}$: For each $j\in[J]$, applying the dyadic convex combination construction to the circuit-coefficient pairs $\rbra[\big]{0\Vert C_0,\frac{1+\lambda_j}{4}}$, $\rbra[\big]{0\Vert C_1,\frac{1-\lambda_j}{4}}$, $\rbra[\big]{1\Vert C_0,\frac{1-\lambda_j}{4}}$, and $\rbra[\big]{1\Vert C_1,\frac{1+\lambda_j}{4}}$, with four terms and $L=\log(64J)$, gives a circuit sampling $D_{<0,j}$. The choice $L=\log(64J)$ is valid because $(1\pm\lambda_j)/4=(16J\pm j)/(64J)$.
\end{itemize}
In both ingredient constructions, only the selected sample is output; the interval label is omitted.

\parheading{Analysis.}
Since the labels $j-1$, $J$, and $J+1$ occur with respective probabilities $\pi_j$, $\pi_{J+1}$, and $\pi_{J+2}$ in $D'_0$ and $D'_1$, applying the chain rule for entropy (\Cref{lemma:joint-entropy}) to \Cref{eq:classical-implementation} yields
\begin{subequations}
\label{eq:classical-encoding}
\begin{align}
    \H_2(D'_0)-\H_2(D'_1)
    &=\frac1\Upsilon\sum_{j=1}^Jc_j\rbra*{\H_2(D_{\geq0})-\H_2(D_{<0,j})}-\frac{\tau(n)}{\Upsilon}\\
    &=\frac1\Upsilon\sum_{j=1}^Jc_j\rbra*{\H_2(D_+)-\frac{\H_2(D_{0,\lambda_j})+\H_2(D_{1,\lambda_j})}{2}}-\frac{\tau(n)}{\Upsilon}\\
    &=\frac1\Upsilon\sum_{j=1}^Jc_j\JS_2(D_{0,\lambda_j},D_{1,\lambda_j})-\frac{\tau(n)}{\Upsilon}\enspace.
\end{align}
\end{subequations}
Here, exchanging $D_{\geq0}$ and $D_{<0,j}$ when $c_j<0$ gives the signed coefficient $c_j/\Upsilon$, while the threshold branches contribute $-\tau(n)/\Upsilon$ and the remainder branches contribute zero in the first line. The second line uses $\H_2(D_{\geq0})=1+\H_2(D_+)$ and $\H_2(D_{<0,j})=1+\frac12\sum_{z\in\binset}\H_2(D_{z,\lambda_j})$. The last line follows from $(D_{0,\lambda_j}+D_{1,\lambda_j})/2=D_+$.

Utilizing \Cref{eq:classical-approximation}, the choice of $\tau(n)$, and $6\varepsilon<\Delta$, we obtain
\begin{itemize}
    \item For \emph{yes} instances, $\TV(D_0,D_1)\geq a(n)$ implies that
    \[ \frac{1}{\Upsilon}\rbra[\Bigg]{\sum_{j=1}^Jc_j\JS_2(D_{0,\lambda_j},D_{1,\lambda_j})-\tau(n)}
        \geq\frac{\TV(D_0,D_1)-\varepsilon-\tau(n)}{\Upsilon}
        \geq\frac{\Delta/2-2\varepsilon}{\Upsilon}>\frac{\varepsilon}{\Upsilon}\eqqcolon g\enspace.
    \]
    Here, the second inequality uses $\tau(n)\leq(a(n)+b(n))/2+\varepsilon$.

    \item For \emph{no} instances, $\TV(D_0,D_1)\leq b(n)$ implies that
    \[ \frac{1}{\Upsilon}\rbra[\Bigg]{\sum_{j=1}^Jc_j\JS_2(D_{0,\lambda_j},D_{1,\lambda_j})-\tau(n)} 
        \leq\frac{\TV(D_0,D_1)+\varepsilon-\tau(n)}{\Upsilon}
        \leq\frac{-\Delta/2+2\varepsilon}{\Upsilon}<-\frac{\varepsilon}{\Upsilon}\enspace.
    \]
    Here, the second inequality uses $\tau(n)\geq(a(n)+b(n))/2-\varepsilon$.
\end{itemize}
Consequently, $\log\rbra*{1/g} = 1 + \ceil*{496\beta/\varepsilon} + \log\rbra*{1/\varepsilon} = O(\log n)$, where the last equality follows from $1/\varepsilon\leq16/\Delta=O(\log n)$.
Thus, $g\geq n^{-O(1)}$, giving the required promise gap in \Cref{lemma:ED-in-SZK}.\footnote{We may pad the descriptions of $C'_0,C'_1$ so that their total length $N$ is at least $n$. Choosing a nondecreasing polynomial $q$ with $g\geq1/q(n)$ then gives $g\geq1/q(N)$, as required for an $\ED[1/q]$ instance.}

\vspace{1em}
It remains to check the computational efficiency of the reduction. Each term circuit in \Cref{table:classical-polarization-newC0,table:classical-polarization-newC1} uses $O(1)$ copies of $C_0$ and $C_1$ and $O(r+\log J)$ additional Boolean gates. The final convex combination uses $J+2=O(1/\varepsilon)$ terms and $L=\log(2\Upsilon J^2)=O(1/\varepsilon)$, with $O\rbra*{J(L+r+\log(J+2))}$ additional Boolean gates.
Since all dyadic coefficients are computable in polynomial time in $n$ and $1/\varepsilon$ and have $O(1/\varepsilon)$-bit representations, the circuit descriptions of $C'_0$ and $C'_1$ can be computed in time $\poly(n,1/\varepsilon)=\poly(n,1/\Delta)$ and have polynomial size in the same parameters. Since $1/\Delta=O(\log n)$, both bounds are polynomial in $n$.
The reduction to $\ED[1/q]$, together with \Cref{lemma:ED-in-SZK}, therefore establishes $\SD[a,b]\in\SZK$.
\end{proof}

\section{\texorpdfstring{$\QSD\in\QSZK$}{}: polarizing regime but beyond logarithmic precision}
\label{sec:QSD-in-QSZK}

Next, we improve the parameter regime for which $\QSD[a,b] \in \QSZK$ in~\cite[Theorem~4.5]{Liu23} from $a(n)^2- \sqrt{2\ln{2}} b(n) \geq 1/\poly(n)$ to $a(n)^2-b(n) \geq 1/\poly(n)$, which serves as a quantum counterpart of~\cite[Corollary~1.7]{BDRV19} and leads to a simple proof that $\SD[a,b]\in\SZK$ when $a(n)^2-b(n) \geq 1/\poly(n)$ by encoding the distributions as $n$-qubit diagonal states: 

\begin{theorem}[Improved \QSZK{} containment of \QSD{}]
    \label{thm:QSD-in-QSZK-improved}
    Let $a, b \colon \Naturals \to [0, 1]$ be efficiently computable functions such that $0 \leq b(n) < a(n) \leq 1$ for every $n\in\Naturals$. If \[a(n)^2-b(n) \geq 1/\poly(n)\] for all sufficiently large $n$, then $\QSD[a,b] \in \QSZK$. 
\end{theorem}

Similar to~\cite[Theorem~4.5]{Liu23}, our approach for proving~\Cref{thm:QSD-in-QSZK-improved} provides a new reduction from \QSD{} to \QJSP{} without the earlier parameter loss. Here, the \textsc{Quantum Jensen--Shannon Divergence Problem} (\QJSP{}) is defined similarly to \QSD{} (see \Cref{def:QSD}), except that the closeness measure is $\QJS_2(\rho_0,\rho_1)$ instead of $\TD(\rho_0,\rho_1)$. 
Then, the \QSZK{} containment follows immediately from~\cite[Lemma 4.4]{Liu23}, which essentially uses the quantum entropy extraction approach to polarize quantum distances~\cite{BASTS10}:   

\begin{lemma}[\QJSP{} is in \QSZK{}, adapted from~{\cite[Lemma 4.4]{Liu23}}]
    \label{lemma:QJSP-in-QSZK}
    Let $a, b \colon \Naturals \to [0, 1]$ be efficiently computable functions such that $0 \leq b(n) < a(n) \leq 1$ for every $n\in\Naturals$. If \[a(n)-b(n) \geq 1/\poly(n)\] for all sufficiently large $n$, then $\QJSP[a,b] \in \QSZK$.
\end{lemma}

The key ingredient underlying our approach is the following parameterized inequality relating the quantum Jensen--Shannon divergence to the trace distance: 

\begin{lemma}[$\QJS$ vs.~$\TD$, a parameterized version]
    \label{lemma:QJS-vs-TD-parameterized}
    Let $\rho_0$ and $\rho_1$ be quantum states of the same dimension. Let $\lambda \in [0,1)$. Define the quantum states $\rho_{j,\lambda}$ for $j\in\binset$ as in \Cref{eq:state-interpolate}.
    Then, with $g(\lambda) \coloneqq \ln{2} - \H\rbra[\big]{\frac{1+\lambda}{2}}$, we have
    \[ \frac{\lambda^2}{2} \TD(\rho_0,\rho_1)^2 \leq \QJS(\rho_{0,\lambda},\rho_{1,\lambda}) \leq g(\lambda) \TD(\rho_0,\rho_1)\enspace. \]
    In particular, we have the following simplified inequality: 
    \[ \frac{\lambda^2}{2} \TD(\rho_0,\rho_1)^2 \leq \QJS(\rho_{0,\lambda},\rho_{1,\lambda}) \leq \frac{\lambda^2}{2 (1-\lambda^2)} \TD(\rho_0,\rho_1)\enspace. \]
\end{lemma}

We need the following version of the Pinsker inequality to proceed with the proof:
\begin{lemma}[Quantum Pinsker inequality, adapted from~{\cite[Theorem 5.38]{Watrous18}}]
    \label{lemma:quantum-pinsker}
    Let $\rho_0$ and $\rho_1$ be quantum states of the same dimension. Then the following inequality holds: 
    \[ \D(\rho_0\|\rho_1) \geq 2 \cdot \TD(\rho_0,\rho_1)^2\enspace. \]
    Here, the quantum relative entropy $\D(\rho_0\|\rho_1)$ is defined using the \emph{natural} logarithm. 
\end{lemma}

\begin{proof}[Proof of \Cref{lemma:QJS-vs-TD-parameterized}]
    We begin by defining $\rho_\pm \coloneqq (\rho_0 \pm \rho_1)/2$. Then, $\rho_{0,\lambda}$ and $\rho_{1,\lambda}$ satisfy 
    \begin{equation}
        \label{eq:TD-parameterized}
        \forall j\in\binset, \quad
        \TD(\rho_{j,\lambda},\rho_+) 
        = \TD\rbra*{ \rho_+ +(-1)^j \lambda \rho_-, \rho_+ } = \frac{\lambda}{2} \TD(\rho_0,\rho_1)\enspace. 
    \end{equation}
    Intuitively, $\rho_{j,\lambda}$ moves continuously from $\rho_+$ toward $\rho_j$ as $\lambda$ increases from $0$ to $1$, analogous to~\cite[Proposition 4.8]{BDRV19}.
    Since $\ker \rho_+ = \ker \rho_0 \cap \ker \rho_1$, it follows that $\rho_0$, $\rho_1$, and $\rho_\pm$ are all supported on $\calH \coloneqq \supp(\rho_+)$. We restrict our attention to $\calH$, where $\rho_+$ is positive definite, and all logarithms and relative entropies appearing throughout the proof are taken on this support. 

    Noting that $\rho_{0,\lambda}+\rho_{1,\lambda} = 2 \rho_+$, the lower bound follows immediately by combining the quantum Pinsker inequality (\Cref{lemma:quantum-pinsker}) with \Cref{eq:TD-parameterized}:
    \[ \QJS(\rho_{0,\lambda},\rho_{1,\lambda}) = \frac{1}{2} \sum_{j\in\binset} \D(\rho_{j,\lambda}\|\rho_+) \geq \sum_{j\in\binset} \TD(\rho_{j,\lambda},\rho_+)^2 = \frac{\lambda^2}{2} \TD(\rho_0,\rho_1)^2\enspace. \]

    \vspace{1em}
    It remains to establish the upper bound, which is the more challenging direction. We now sketch the proof. To control how the divergence changes, we regard $\lambda$ as the endpoint of the path $\rho_+ \pm t\rho_-$ starting at $t=0$ and integrate the infinitesimal change in the divergence along this path. Symmetry makes the first-order variation vanish at $t=0$, while the operator monotonicity of the logarithm~\cite[Section 5.3.7]{Bhatia09} bounds the derivative. Integrating from $0$ to $\lambda$ therefore produces the desired quadratic scaling $O\rbra*{\lambda^2 \TD(\rho_0,\rho_1)}$. 

    Next, we proceed with the actual proof by defining the following function:
    \begin{subequations}
        \label{eq:operator-path}
        \begin{align}
            \forall t\in[0,1), \quad F(t) &\coloneqq \QJS(\rho_{0,t}, \rho_{1,t})\\ 
            &= \frac{1}{2} \Tr\rbra*{ \rbra*{\rho_+ + t \rho_-} \ln\rbra*{\rho_+ + t \rho_-} }\\
            &\qquad + \frac{1}{2} \Tr\rbra*{ \rbra*{\rho_+ - t \rho_-} \ln\rbra*{\rho_+ - t \rho_-} }- \Tr(\rho_+ \ln \rho_+)\enspace.
        \end{align}
    \end{subequations}
    
    A direct calculation shows that $F(0)=0$ and $F(\lambda) = \QJS(\rho_{0,\lambda},\rho_{1,\lambda})$, and thus $F(t)$ describes the aforementioned path parameterized by $t$. We then bound the changes of the divergence along the path using the Daleckii--Krein formula (e.g.,~\cite[Section 5.3.1]{Bhatia09}). In particular, let $H(t)$ be a differentiable positive-definite operator path, and let $f$ be a continuously differentiable real function on an interval containing its spectrum. Taking $f(x)=x \ln x$, we obtain
    \begin{equation}
        \label{eq:trace-functional-rule}
        \frac{\dd}{\dd t} \Tr(H \ln H) = \Tr\rbra*{ f'(H(t)) H'(t) } = \Tr\rbra*{ H'(\ln H + I) }, \quad\text{where } H'(t) \coloneqq \frac{\dd H(t)}{\dd t}\enspace.
    \end{equation}
    
    Applying \Cref{eq:trace-functional-rule} to the path $F(t)$ defined in \Cref{eq:operator-path} yields
    \begin{subequations}
        \label{eq:path-changes}
        \begin{align}
            F'(t) &= \frac{1}{2} \Tr\rbra*{ \rho_- \rbra*{ \ln\rbra*{\rho_+ + t\rho_-} - \ln\rbra*{\rho_+ - t\rho_-} } }\\
            &\leq \abs*{F'(t)}\\
            &\leq \frac{1}{2} \norm*{\rho_-}_1 \cdot \norm*{ \ln\rbra*{\rho_+ + t\rho_-} - \ln\rbra*{\rho_+ - t\rho_-} }_\infty\\
            &\leq \frac{1}{2} \TD(\rho_0,\rho_1) \cdot \ln\rbra*{\frac{1+t}{1-t}}\enspace.
        \end{align}
    \end{subequations}
    Here, the third line follows from the H\"older inequality for Schatten norms (\Cref{lemma:matrix-Holder}), and the last line uses the identity $\norm{\rho_-}_1 = \TD(\rho_0,\rho_1)$ and the inequality 
    \begin{equation}
        \label{eq:log-operator-norm-bound}
        \norm*{ \ln\rbra*{\rho_+ + t\rho_-} - \ln\rbra*{\rho_+ - t\rho_-} }_\infty\ \leq \ln(1+t) - \ln(1-t)\enspace.
    \end{equation}
    To prove \Cref{eq:log-operator-norm-bound}, we start by observing that $-\rho_+ \preceq \rho_- \preceq \rho_+$ since both $\rho_0$ and $\rho_1$ are positive semidefinite, which implies $(1-t)\rho_+ \preceq \rho_+ \pm t\rho_- \preceq (1+t)\rho_+$. Combining this observation with the operator monotonicity of the logarithm~\cite[Section 5.3.7]{Bhatia09}, we obtain
    \begin{equation}
        \label{eq:log-operator-bound}
        \ln \rho_+ + \ln(1-t)I \preceq \ln\rbra*{\rho_+ \pm t\rho_-} \preceq \ln\rho_+ + \ln(1+t)I\enspace. 
    \end{equation}
    Consequently, \Cref{eq:log-operator-bound} directly implies \Cref{eq:log-operator-norm-bound}. 

    Noting that $\frac{1}{2} \ln\rbra*{(1+t)/(1-t)} = \arctanh(t)$, integrating the derivative $F'(t)$ in \Cref{eq:path-changes} from $0$ to $\lambda$ gives that
    \[ F(\lambda) = \int_0^\lambda F'(t)\,\dd t \leq \TD(\rho_0,\rho_1) \int_{0}^{\lambda} \arctanh(t)\,\dd t = g(\lambda) \TD(\rho_0,\rho_1)\enspace. \]
    Here, $g(\lambda) = \ln{2} - \H\rbra[\big]{\frac{1+\lambda}{2}} = \frac{1}{2} \rbra*{(1+\lambda)\ln(1+\lambda) + (1-\lambda)\ln(1-\lambda)}$.
    Lastly, we complete the proof by obtaining the simplified upper bound from the Taylor series:
    \[ g(\lambda) = \sum_{j=1}^{\infty} \frac{\lambda^{2j}}{(2j-1)(2j)} \leq \frac{1}{2} \sum_{j=1}^\infty \lambda^{2j} = \frac{\lambda^2}{2(1-\lambda^2)}\enspace. \qedhere \]
\end{proof}

\begin{remark}[Connection to the classical approach of~\cite{BDRV19}]
    For two distributions $D_0$ and $D_1$, writing $D_\pm \coloneqq (D_0 \pm D_1)/2$, the construction of~\cite[Section 4.2.1]{BDRV19} replaces the pair $D_+ \pm D_-$ by $D_+ \pm \lambda D_-$ and analyzes the resulting Jensen--Shannon divergence through a scalar expansion around $D_+$. A direct quantum analogue of this analysis does not work, as explained in~\cite[Remark~4.6]{Liu23}: while the triangular discrimination (which is related to the total variation distance) serves as a constant-factor approximation to the Jensen--Shannon divergence, the corresponding quantum quantities do not satisfy an analogous relation. Here, we retain the same interpolation $\rho_+ \pm t \rho_-$ but bypass this obstacle by analyzing the quantum Jensen--Shannon divergence directly: we bound its infinitesimal change along the path using the Fr\'{e}chet derivative and matrix inequalities and then integrate from $t=0$ to $t=\lambda$. 
\end{remark}

With the new inequalities established in \Cref{lemma:QJS-vs-TD-parameterized}, we are ready to prove \Cref{thm:QSD-in-QSZK-improved}:

\begin{proof}[Proof of \Cref{thm:QSD-in-QSZK-improved}]
    Using the \QSZK{} containment in \Cref{lemma:QJSP-in-QSZK}, it suffices to establish a Karp reduction from $\QSD[a,b]$ with $a(n)^2 - b(n) \geq 1/p(n)$ for some polynomial $p(n) \geq 1$ to $\QJSP[a',b']$ with $a'(n)-b'(n) \geq 1/q(n)$, where $q(n) \coloneqq 64 \ln{2} \cdot p(n)^2$ is also polynomial in $n$. 

    Let the quantum circuit pair $(Q_0,Q_1)$ be an instance of \QSD{}, where each $Q_j$ prepares a purification of the state $\rho_j$ for $j\in\binset$. 
    We now construct the quantum circuit $Q'_0$ that prepares $\rho'_0 \coloneqq \rho_{0,\lambda}$ by applying the dyadic convex combination construction (\Cref{lemma:dyadic-linear-combi-states}) to the circuit-coefficient pairs $(Q_0,(1+\lambda)/2)$ and $(Q_1,(1-\lambda)/2)$, with $J=2$ and $L=k+1$. 
    Similarly, the quantum circuit $Q'_1$ that prepares $\rho'_1\coloneqq \rho_{1,\lambda}$ is obtained by applying \Cref{lemma:dyadic-linear-combi-states} to $(Q_0,(1-\lambda)/2)$ and $(Q_1,(1+\lambda)/2)$, with $J=2$ and $L=k+1$.
    Since $k = O(\log{n}) \leq \poly(n)$, both circuit descriptions can be computed in polynomial time. 
    
    Applying \Cref{lemma:QJS-vs-TD-parameterized} to the promise of $\QSD[a,b]$, we obtain the following:
    \begin{itemize}
        \item If $\TD(\rho_0,\rho_1) \geq a(n)$, then $\QJS_2(\rho'_0,\rho'_1) \geq \frac{\lambda^2}{2\ln{2}} \TD(\rho_0,\rho_1)^2 \geq \frac{\lambda^2 a(n)^2}{2\ln{2}} \coloneqq a'(n)$;
        \item If $\TD(\rho_0,\rho_1) \leq b(n)$, then we have
        \begin{align*}
            \QJS_2(\rho'_0,\rho'_1) \leq \frac{\lambda^2}{2\ln{2} (1-\lambda^2)} \TD(\rho_0,\rho_1) &\leq \frac{\lambda^2}{2\ln{2}} \cdot \frac{b(n)}{1-\lambda^2}\\
            &= \frac{\lambda^2}{2\ln{2}} \cdot \rbra*{b(n) + \frac{b(n) \lambda^2}{1-\lambda^2}}\\
            &\leq \frac{\lambda^2}{2\ln{2}} \cdot \rbra*{b(n) + 2\lambda^2}\\
            &\leq \frac{\lambda^2}{2\ln{2}} \cdot \rbra*{b(n)+\frac{1}{2p(n)}} \coloneqq b'(n)\enspace.
        \end{align*}
        Here, the third line follows from the fact that $0 \leq b(n) \leq 1$ and $\lambda^2 \leq 1/\rbra*{4p(n)} \leq 1/4$, while the latter inequality also implies the last line. 
    \end{itemize}

    It remains to check the promise gap. Since $\lambda^2 > 1/(16p(n))$, a direct calculation shows that
    \[ a'-b' = \frac{\lambda^2}{2\ln{2}} \rbra*{a^2-b-\frac{1}{2p(n)}} \geq \frac{\lambda^2}{4 \ln{2} \cdot p(n)} > \frac{1}{64\ln{2} \cdot p(n)^2} \coloneqq \frac{1}{q(n)}\enspace. \]
    Therefore, we conclude that $\QSD[a,b]\in \QSZK$ when $a(n)^2-b(n) \geq 1/p(n)$ via a Karp reduction to $\QJSP[a',b'] $ with $a'(n)-b'(n) \geq 1/q(n)$. 
\end{proof}


\section{Classical messages suffice for a laconic prover}
\label{sec:classical-messages-suffice}

We show that in two-message quantum interactive proof systems, the prover's response can be made classical without incurring a large overhead in the communication. Specifically, we generalize the statement $\qqQAM = \qcQAM$ from~\cite[Theorem 1.7(ii)]{KLGN19} to arbitrary two-message quantum interactive proof systems, while explicitly tracking the length of the prover’s response:

\begin{theorem}[$\QIP_{\ell\text{-}\qubit} = \QIP_{2\ell\text{-}\bit}$]
\label{thm:QIPellqubit=QIP2ellbit}
For all efficiently computable functions $\ell(n)$, $c(n)$, and $s(n)$ satisfying $1\leq \ell(n)\leq \poly(n)$ and $0 \leq s(n) < c(n) \leq 1$,
\[\QIP_{\ell\text{-}\qubit}\sbra*{2,c,s} = \QIP_{2\ell\text{-}\bit}\sbra*{2,c,s}\enspace.\]
\end{theorem}

As we will see in the proof of \Cref{thm:QIPellqubit=QIP2ellbit}, our construction preserves the public-coin property, yielding the following quantitative version of the statement $\qqQAM = \qcQAM$.

\begin{corollary}[$\qqQAM = \qcQAM$, a quantitative version]
\label{cor:qqQAMell=qcQAM2ell}
For all efficiently computable functions $\ell(n)$, $c(n)$, and $s(n)$ satisfying $1\leq \ell(n)\leq \poly(n)$ and $0 \leq s(n) < c(n) \leq 1$,
\[\qqQAM[\ell, c, s] = \qcQAM[2\ell, c, s]\enspace.\]
\end{corollary}

\begin{proof}[Proof of \Cref{thm:QIPellqubit=QIP2ellbit}]
We show two containments.

\paragraph{The containment $\QIP_{\ell\text{-}\qubit}\sbra*{2,c,s} \subseteq \QIP_{2\ell\text{-}\bit}\sbra*{2,c,s}$.}

Given a $\QIP\sbra*{2}$ protocol $\protocol{P_1}{V_1}$ where the prover $P_1$ sends $\ell$ qubits, we provide a $\QIP\sbra*{2}$ protocol $\protocol{P_2}{V_2}$ for the same promise problem as $\protocol{P_1}{V_1}$ in \Cref{protocol:QIP2ellbit[2]} such that the prover $P_2$ sends $2\ell$ bits. The construction is based on the idea of quantum teleportation \cite{BBCJPW93}. A similar teleportation-based construction has been used in~\cite[Lemma 7.2]{KLGN19} to prove ${\rm qcq}\text{-}\QAM \subseteq {\rm qcc}\text{-}\QAM$, and also appears in~\cite[Theorem 1.7(ii)]{KLGN19}.

\begingroup
\LinesNotNumbered
\begin{algorithm}[!ht]
    \SetAlgorithmName{Protocol}{protocol}{List of Protocols}
    \caption{A $\QIP_{2\ell\text{-}\bit}\sbra*{2}$ protocol $\protocol{P_2}{V_2}$ based on a $\QIP_{\ell\text{-}\qubit}\sbra*{2}$ protocol $\protocol{P_1}{V_1}$.}
	\label{protocol:QIP2ellbit[2]}
    \SetEndCharOfAlgoLine{.}
    \SetKwFor{While}{}{:}{}
    \SetKwFor{For}{For}{:}{}
    \SetKwIF{If}{ElseIf}{Else}{If}{:}{elif}{Else:}{}%
    \SetKwComment{Comment}{// }{}
    \SetKwInOut{Input}{Input}
    \SetKwInOut{Output}{Output}

    \smallskip
    \Input{$x \in \calI \coloneq (\calI_\yes,\calI_\no)$, where $\calI \in \QIP_{\ell\text{-}\qubit}\sbra*{2,c,s}$ corresponds to $\protocol{P_1}{V_1}$.}
    \smallskip
    \Output{ACCEPT or REJECT.}
    \medskip

    \textbf{1.} $V_2$ simulates $V_1$ on the instance $x$ to obtain a quantum state on registers $(\sfM_1, \sfV)$ where $\sfM_1$ contains the message of $V_1$ and $\sfV$ is the private register of $V_1$. $V_2$ then prepares $\ell$ EPR pairs $\ket{\Phi}_{\sfP'\sfV'} \coloneq \frac{1}{\sqrt{2^{\ell}}}\sum_{t \in \binset^\ell}\ket{t, t}_{\sfP'\sfV'}$, and sends registers $(\sfM_1, \sfP')$ to $P_2$ while keeping the registers $(\sfV, \sfV')$.
    \medskip
    
    \textbf{2.} $V_2$ receives $(a, b)$, where $a, b \in \binset^\ell$ are supposed to be the measurement outcome that $P_2$ obtains after simulating $P_1$ on register $\sfM_1$ to obtain an $\ell$-qubit message state in $\sfM_2$, and doing a Bell-basis measurement on registers $(\sfM_2, \sfP')$.
    \medskip
    
    \textbf{3.} $V_2$ applies $\PauliX(a)\PauliZ(b)$ on the register $\sfV'$, and invokes $V_1$ on the instance $x$ where we use the state in register $\sfV'$ as the message sent by $P_1$, and the state in register $\sfV$ as the private state of $V_1$. $V_2$ outputs the result of $V_1$.
\end{algorithm}
\endgroup

\begin{description}
    \item[Completeness.] In the protocol $\protocol{P_2}{V_2}$, the honest prover $P_2$ performs the quantum teleportation on the quantum message, and thus the verifier $V_2$ can recover the same $\ell$-qubit message as in the original protocol. Hence the completeness is preserved.
    \item[Soundness.] Suppose a cheating prover $P_2^*$ convinces $V_2$ on input $x$ with probability $\epsilon$. We show that there exists a cheating prover $P_1^*$ that convinces $V_1$ on the same input $x$ with the same probability $\epsilon$.

    Let $\sigma_{\sfV'\sfV}$ be the state on registers $(\sfV', \sfV)$ when $V_2$ interacts with $P_2^*$ after $V_2$ applies the Pauli correction $X(a)Z(b)$. Because no one has touched the register $\sfV$ since $V_2$ sends the state in the first step, we obtain $\Tr_{\sfV'}(\sigma_{\sfV'\sfV}) = \Tr_{\sfM_1}\rbra*{\ketbra{\psi}{\psi}_{\sfM_1\sfV}}$ where $\ket{\psi}_{\sfM_1\sfV}$ is the joint state of registers $(\sfM_1, \sfV)$ after $V_1(x)$ outputs the message. 

    Next we show that there exists a quantum channel $\Lambda \colon \sfM_1 \to \sfV'$ such that
    \begin{align}\label{eqn:exist-channel}
        (\Lambda\otimes \Id_\sfV) \rbra*{\ketbra{\psi}{\psi}_{\sfM_1\sfV}} = \sigma_{\sfV'\sfV}\enspace.
    \end{align}

    Let $\ket{\phi}_{\sfA\sfV'\sfV}$ be a purification of $\sigma_{\sfV'\sfV}$. Without loss of generality, we can assume the size of the environment register $\sfA$ is larger than the size of $\sfM_1$. We have that \[\Tr_{\sfA\sfV'}\rbra*{\ketbra{\phi}{\phi}_{\sfA\sfV'\sfV}} = \Tr_{\sfV'}\rbra*{\sigma_{\sfV'\sfV}} = \Tr_{\sfM_1}\rbra*{\ketbra{\psi}{\psi}_{\sfM_1\sfV}}\enspace.\]

    By the unitary equivalence of purifications~\cite[Theorem 2.12]{Watrous18}, there exists a unitary $U$ acting trivially on $\sfV$ such that
    \[(U\otimes I_\sfV)\ket{\psi}_{\sfM_1\sfV}\ket{\bar{0}}_{\sfA'} = \ket{\phi}_{\sfA\sfV'\sfV}\enspace, \]  
    where $\ket{\bar{0}}$ denotes the state of the ancillary qubits used to ensure that the left-hand side and the right-hand side have the same number of qubits. 
    Then one can implement the following channel, which satisfies \Cref{eqn:exist-channel}: initialize the ancillary qubits in the state $\ket{\bar{0}}_{\sfA'}$, apply the above unitary $U$, and trace out the register $\sfA$.

    The cheating prover $P_1^*$ can use the channel $\Lambda$ in \Cref{eqn:exist-channel} to convince $V_1$ on input $x$ with exactly the same probability $\epsilon$: on receiving $\sfM_1$, $P_1^*$ applies the channel $\Lambda$ and sends back the register $\sfV'$. Then the verifier $V_1$'s view is exactly $\sigma$, so $V_1$ outputs accept with probability $\epsilon$. Hence the soundness is preserved. 
\end{description}

\paragraph{The containment $\QIP_{2\ell\text{-}\bit}\sbra*{2,c,s} \subseteq \QIP_{\ell\text{-}\qubit}\sbra*{2,c,s}$.}
Given a $\QIP\sbra*{2}$ protocol $\protocol{P_1}{V_1}$ where the prover $P_1$ sends $2\ell$ bits, we provide a $\QIP\sbra*{2}$ protocol $\protocol{P_2}{V_2}$ for the same promise problem as $\protocol{P_1}{V_1}$ in \Cref{protocol:QIPellqubit[2]} such that the prover $P_2$ sends $\ell$ qubits. The construction is based on the idea of superdense coding \cite{BW92}.

\begingroup
\LinesNotNumbered
\begin{algorithm}[!ht]
    \SetAlgorithmName{Protocol}{protocol}{List of Protocols}
    \caption{A $\QIP_{\ell\text{-}\qubit}\sbra*{2}$ protocol $\protocol{P_2}{V_2}$ based on a $\QIP_{2\ell\text{-}\bit}\sbra*{2}$ protocol $\protocol{P_1}{V_1}$.}
	\label{protocol:QIPellqubit[2]}
    \SetEndCharOfAlgoLine{.}
    \SetKwFor{While}{}{:}{}
    \SetKwFor{For}{For}{:}{}
    \SetKwIF{If}{ElseIf}{Else}{If}{:}{elif}{Else:}{}%
    \SetKwComment{Comment}{// }{}
    \SetKwInOut{Input}{Input}
    \SetKwInOut{Output}{Output}

    \smallskip
    \Input{$x \in \calI \coloneq (\calI_\yes,\calI_\no)$, where $\calI \in \QIP_{2\ell\text{-}\bit}\sbra*{2}$ corresponds to $\protocol{P_1}{V_1}$.}
    \smallskip
    \Output{ACCEPT or REJECT.}
    \medskip

    \textbf{1.} $V_2$ simulates $V_1$ on the instance $x$ to obtain a quantum state on registers $(\sfM_1, \sfV)$ where $\sfM_1$ contains the message of $V_1$ and $\sfV$ is the private register of $V_1$. $V_2$ then prepares $\ell$ EPR pairs $\ket{\Phi}_{\sfP'\sfV'} \coloneq \frac{1}{\sqrt{2^{\ell}}}\sum_{t \in \binset^\ell}\ket{t, t}_{\sfP'\sfV'}$, and sends registers $(\sfM_1, \sfP')$ to $P_2$ while keeping the registers $(\sfV, \sfV')$.
    \medskip
    
    \textbf{2.} $V_2$ receives a register $\sfP'$, which is supposed to be obtained through the following procedure: $P_2$ simulates $P_1$ on register $\sfM_1$ to get a $2\ell$-bit string $s$, and then applies $X(s_1)Z(s_2)$ to the register $\sfP'$, where $s_1$ is the first $\ell$ bits of the string $s$ and $s_2$ is the second $\ell$ bits of the string $s$.
    \medskip
    
    \textbf{3.} $V_2$ measures $(\sfP', \sfV')$ in the Bell basis to get an outcome $(s_1, s_2)$ where $s_1, s_2 \in \binset^\ell$, and invokes $V_1$ on the instance $x$ where we use the state in register $\sfV$ as the private state of $V_1$, and use $s_1 \parallel s_2$ as the $2\ell$-bit message sent by $P_1$. $V_2$ outputs the result of $V_1$.
\end{algorithm}
\endgroup

\begin{description}
    \item[Completeness.] In the protocol $\protocol{P_2}{V_2}$, the honest prover $P_2$ transmits the $2\ell$-bit classical message by sending an $\ell$-qubit quantum message via superdense coding, and thus the verifier $V_2$ can recover the same $2\ell$-bit classical message as in the original protocol. Hence the completeness is preserved.
    \item[Soundness.] Suppose an adversary $P_2^*$ convinces $V_2$ on input $x$ with probability $\epsilon$. We show that there exists an adversary $P_1^*$ that convinces $V_1$ on input $x$ with probability $\epsilon$.

    Let $\sigma$ be the joint state of the measurement outcome and register $\sfV$ after $V_2$ measures $(\sfP', \sfV')$ in the Bell basis. Then we can write $\sigma$ as $\sum_{s_1, s_2}\ketbra{s_1, s_2}{s_1, s_2}_{\sfM_2} \otimes \rho_{s_1, s_2}$ where $\rho_{s_1, s_2}$ is the subnormalized state on $\sfV$ conditioned on the Bell-basis measurement outcome $(s_1, s_2)$. 
    Since no one has touched register $\sfV$ after $V_2$ sends the state in the first step, we obtain $\Tr_{\sfM_2}(\sigma_{\sfM_2\sfV}) = \Tr_{\sfM_1}\rbra*{\ketbra{\psi}{\psi}_{\sfM_1\sfV}}$, where $\ket{\psi}_{\sfM_1\sfV}$ is the joint state of registers $(\sfM_1, \sfV)$ after $V_1(x)$ outputs the message. 

    Next, we show that there exists a POVM $\{M_{s_1, s_2}\}_{s_1, s_2}$ on register $\sfM_1$ such that 
    \begin{align}\label{eqn:exist-POVM-to-get-rho}
        \rho_{s_1, s_2} = \Tr_{\sfM_1} \rbra*{\rbra*{M_{s_1, s_2}\otimes I_\sfV}\ketbra{\psi}{\psi}_{\sfM_1\sfV} }\enspace.
    \end{align}

    By the same argument as in the soundness analysis of \Cref{protocol:QIP2ellbit[2]}, there exists a quantum channel $\Lambda \colon \sfM_1 \to \sfM_2$ such that 
    \[\rbra*{\Lambda\otimes\Id_{\sfV}}(\ketbra{\psi}{\psi}_{\sfM_1\sfV}) = \sigma_{\sfM_2\sfV} = \sum_{s_1, s_2}\ketbra{s_1, s_2}{s_1, s_2}_{\sfM_2} \otimes \rho_{s_1, s_2}\enspace.\]

    Then the POVM induced by first applying the channel $\Lambda$, and then making a computational basis measurement on $\sfM_2$ satisfies \Cref{eqn:exist-POVM-to-get-rho}.

    The cheating prover $P_1^*$ can use the POVM $\{M_{s_1, s_2}\}_{s_1, s_2}$ from \Cref{eqn:exist-POVM-to-get-rho} to convince $V_1$ on input $x$ with exactly the same probability $\epsilon$: upon receiving $\sfM_1$, $P_1^*$ applies the POVM $\{M_{s_1, s_2}\}_{s_1, s_2}$ to register $\sfM_1$, and sends back the measurement result. Then the verifier $V_1$'s view is exactly $\sigma$, so $V_1$ accepts with probability $\epsilon$. Hence the soundness is preserved. \qedhere
\end{description}
\end{proof}


\section{Omitted proofs}

\subsection{Proof of \texorpdfstring{\Cref{lemma:extractor-works}}{}}\label{subsection:proof-of-extractor}

\extractorWorks*

\begin{proof}
Define $A_u \coloneq \frac{1}{2^\ell}\sum_{s \in \binset^\ell}(-1)^{\innerprodF{u}{s}} X_{x, s}$. Since $\varsigma$ and $\varsigma'$ have a common component $(1 - p_x)\ketbra{1, 0, 0^\ell, 0, \bar{0}}{1, 0, 0^\ell, 0, \bar{0}}$, we can derive that
\begin{align*}
&\TD\rbra*{\varsigma, \varsigma'}\\
=& \frac{1}{2}\Tr\rbra*{\abs{\frac{1}{2^{2\ell + 2}}\sum_{u, s \in \binset^\ell, v, b \in \binset}(-1)^{h_{u, v}(s) + b + 1}\ketbra{0}{0} \otimes \ketbra{b}{b} \otimes \ketbra{u, v}{u, v} \otimes X_{x, s}}}\\
=& \frac{1}{2^{2\ell + 3}} \sum_{u \in \binset^\ell, v, b \in \binset}\Tr\rbra*{\abs{\sum_{s \in \binset^\ell}(-1)^{h_{u, v}(s)}X_{x, s}}}\\
=& \frac{1}{2^{\ell + 1}}\sum_{u \in \binset^\ell}\Tr\rbra*{\abs{A_u}}\\
=& \frac{1}{2^{\ell + 1}}p_x + \frac{1}{2^{\ell + 1}}\sum_{u \in \binset^\ell \setminus \{0^\ell\}}\Tr\rbra*{\abs{A_u}}\\
\leq & \frac{1}{2^{\ell + 1}} p_x + \frac{1}{2^{\ell + 1}}\sqrt{\rbra*{2^\ell - 1}\sum_{u \in \binset^\ell \setminus \{0^\ell\}}\Tr\rbra*{\abs{A_u}}^2}
\end{align*} 
where in the second line, we plug in $\varsigma$ and $\varsigma'$, in the third line, we use the fact that the operator in the second line is block diagonal, in the fifth line, we use the fact that $X_{x, s} \succeq 0$ and thus $\Tr\rbra*{\abs{A_{0^\ell}}} = \frac{1}{2^\ell}\Tr\rbra*{\abs{\sum_{s \in \binset^\ell}X_{x, s}}} = \frac{1}{2^\ell}\Tr\rbra*{\sum_{s \in \binset^\ell}X_{x, s}} = p_x$, and we apply the Cauchy--Schwarz inequality in the last line.

Let $\overline{X}_{x} \coloneq \frac{1}{2^\ell}\sum_{s \in \binset^\ell}X_{x, s}$. By \cite[Lemma 5.1.2]{Renner05}, since the support of $\overline{X}_x$ contains the support of $A_u$,
\begin{align}\label{eqn:upper-bound-for-sqrt(A_u)}
    \Tr\rbra*{\abs{A_u}}^2 \leq \Tr\rbra*{\overline{X}_x}\Tr\rbra*{\rbra*{\overline{X}_x^{-1/4}A_u\overline{X}_x^{-1/4}}^2} = p_x \Tr\rbra*{\rbra*{\overline{X}_x^{-1/4}A_u\overline{X}_x^{-1/4}}^2}\enspace.
\end{align}
Here, all negative powers of $\overline{X}_{x}$ are understood as the corresponding powers of its Moore--Penrose inverse on $\supp\rbra{\overline{X}_{x}}$.
Thus we can continue the above inequality to get that
\begin{subequations}
\label{eqn:the-inequality-to-be-continued}
\begin{align}
&\TD\rbra*{\varsigma, \varsigma'}\\
\leq & \frac{1}{2^{\ell + 1}} p_x + \frac{1}{2^{\ell + 1}} \sqrt{\rbra*{2^\ell - 1}p_x\sum_{u \in \binset^\ell \setminus \{0^\ell\}}\Tr\rbra*{\rbra*{\overline{X}_x^{-1/4}A_u\overline{X}_x^{-1/4}}^2}}\\
= & \frac{1}{2^{\ell + 1}} p_x + \frac{1}{2^{2\ell + 1}} \sqrt{\rbra*{2^\ell - 1}p_x\sum_{\substack{u \in \binset^\ell \setminus \{0^\ell\}, \\s, s' \in \binset^\ell}}(-1)^{\innerprodF{u}{s \oplus s'}}\Tr\rbra*{\overline{X}_x^{-1/2}X_{x, s}\overline{X}_x^{-1/2}X_{x, s'}}}\enspace,
\end{align}
\end{subequations}
where in the second line, we use \Cref{eqn:upper-bound-for-sqrt(A_u)}, and in the last line, we plug in $A_u$.

Furthermore, we simplify the term $\sum_{u \in \binset^\ell \setminus \{0^\ell\}, s, s' \in \binset^\ell}(-1)^{\innerprodF{u}{s \oplus s'}}\Tr\rbra*{\overline{X}_x^{-1/2}X_{x, s}\overline{X}_x^{-1/2}X_{x, s'}}$ as shown below: 
\begin{align*}
    &\sum_{\substack{u \in \binset^\ell \setminus \{0^\ell\}, \\s, s' \in \binset^\ell}}(-1)^{\innerprodF{u}{s \oplus s'}}\Tr\rbra*{\overline{X}_x^{-1/2}X_{x, s}\overline{X}_x^{-1/2}X_{x, s'}}\\
    =& (2^\ell - 1)\sum_{s \in \binset^\ell}\Tr\rbra*{\overline{X}_x^{-1/2}X_{x, s}\overline{X}_x^{-1/2}X_{x, s}} - \sum_{s, s' \in \binset^\ell \text{ s.t. } s \ne s'}\Tr\rbra*{\overline{X}_x^{-1/2}X_{x, s}\overline{X}_x^{-1/2}X_{x, s'}}\\
    =& 2^\ell(2^\ell - 1)\sum_{s \in \binset^\ell}\Tr\rbra*{M_sX_{x, s}} - 2^\ell\sum_{s, s' \in \binset^\ell \text{ s.t. } s \ne s'}\Tr\rbra*{M_s X_{x, s'}}
\end{align*}
where we split the summation according to whether $s = s'$ in the first line, and in the last line, for $s \in \binset^\ell$, we define $M_s \coloneq \frac{1}{2^{\ell}}\overline{X}_x^{-1/2} X_{x, s}\overline{X}_x^{-1/2}$.

Since each element $M_s$ is positive semidefinite and $\sum_{s}M_s = \overline{X}_x^{-1/2} \overline{X}_x \overline{X}_x^{-1/2} \preceq I$, we can derive a POVM $\{N_s\}_s$ with $N_{0^\ell} = M_{0^\ell} + I - \sum_s M_s$ and $N_s = M_s$ for $s \ne 0^\ell$. Moreover, $\sum_s M_s$ is the projection onto the space $\supp(\overline{X}_x)$ and thus $\Tr\rbra*{N_s X_{x, s'}} = \Tr\rbra*{M_s X_{x, s'}}$ for every $s, s' \in \binset^\ell$. As a result, the above term can be further simplified and upper bounded as below: 
\begin{subequations}
\label{eqn:simplification-of-the-term}
\begin{align}
    &\sum_{\substack{u \in \binset^\ell \setminus \{0^\ell\}, \\s, s' \in \binset^\ell}}(-1)^{\innerprodF{u}{s \oplus s'}}\Tr\rbra*{\overline{X}_x^{-1/2}X_{x, s}\overline{X}_x^{-1/2}X_{x, s'}}\\
    =& 2^\ell(2^\ell - 1)\sum_{s \in \binset^\ell}\Tr\rbra*{N_sX_{x, s}} - 2^\ell\sum_{s, s' \in \binset^\ell \text{ s.t. } s \ne s'}\Tr\rbra*{N_s X_{x, s'}}\\
    =& 2^\ell(2^\ell - 1)\sum_{s \in \binset^\ell}\Tr\rbra*{N_sX_{x, s}} - 2^\ell\sum_{s \in \binset^\ell}\Tr\rbra*{(I - N_s) X_{x, s}}\\
    =& 2^{2\ell} \sum_{s \in \binset^\ell}\Tr\rbra*{N_sX_{x, s}} - 2^{2\ell} p_x\\
    \leq& 2^{2\ell} b - 2^{2\ell} p_x\enspace,
\end{align}
\end{subequations}
where the last line uses $\max_{\{M^*_s\}_{s \in \binset^\ell}} \sum_{s \in \binset^\ell}\Tr\rbra*{M^*_s X_{x, s}} \leq b$.

\Cref{eqn:the-inequality-to-be-continued,eqn:simplification-of-the-term} yield
\begin{align*}
    \TD\rbra*{\varsigma, \varsigma'} \leq& \frac{1}{2^{\ell + 1}} p_x + \frac{1}{2^{\ell + 1}} \sqrt{\rbra*{2^\ell - 1}p_x(b - p_x)}\\
    =& \frac{b}{2^{\ell + 1}} \cdot \rbra*{ \frac{p_x}{b} + \sqrt{\rbra*{2^\ell - 1}\frac{p_x}{b}\rbra*{1 - \frac{p_x}{b}}} }\\
    \leq& \frac{b}{2^{\ell + 1}} \cdot \max_{y \in [0, 1]} \rbra*{y + \sqrt{\rbra*{2^\ell - 1}y(1 - y)}}\\
    \leq & \frac{b}{2^{\ell + 1}} \cdot \frac{1 + 2^{\ell/2}}{2}\enspace.
\end{align*}
Here, we use that $b \geq \sum_{s \in \binset^\ell}\Tr(\frac{I}{2^\ell} \cdot X_{x, s}) = p_x$ in the third line; when $b = 0$, we have $p_x = 0$ and the inequality $\TD\rbra*{\varsigma, \varsigma'} \leq \frac{b}{2^{\ell + 1}} \cdot \frac{1 + 2^{\ell/2}}{2}$ holds trivially.
\end{proof}

\end{document}